\documentclass[aps, notitlepage, pra, twocolumn, nofootinbib, floatfix, superscriptaddress, english, 10pt]{revtex4-1}
\usepackage[cm]{fullpage}
\usepackage[T1]{fontenc}
\usepackage[utf8]{inputenx}
\usepackage{amssymb}
\usepackage{amsmath}
\usepackage{amsthm}
\makeatletter
\def\amsbb{\use@mathgroup \M@U \symAMSb}
\makeatother
\usepackage[mathscr]{euscript}
\usepackage{mathtools}
\usepackage{verbatim}
\usepackage{bbold}
\usepackage{ifthen}
\usepackage{array}
\usepackage{multirow}
\usepackage{placeins}
\usepackage[usenames,dvipsnames,table]{xcolor}
\definecolor{darkred}{RGB}{200, 0, 0}
\definecolor{darkgreen}{RGB}{0, 100, 0}
\definecolor{darkblue}{RGB}{0, 0, 200}
\usepackage{hyperref}
\hypersetup{colorlinks, breaklinks, linkcolor=darkred, urlcolor=darkblue, citecolor=darkgreen}

\newcommand{\nbox}[2][9]{\hspace{#1pt} \mbox{#2} \hspace{#1pt}}
\usepackage{paralist}
\usepackage{gensymb}
\usepackage{aliascnt}
\newtheoremstyle{lemeq}
{}
{}
{}
{}
{\bfseries}
{:}
{\parindent}
{\thmname{#1} \thmnumber{\normalfont{#2}} \thmnote{\normalfont#3}}
\theoremstyle{lemeq}

\newaliascnt{axiom}{theo}
\newtheorem{df}{Definition}[section]
\newtheorem{prop}{Proposition}[section]
\newtheorem{obs}{Observation}[section]

\DeclareMathOperator{\rank}{rank}

\DeclareMathOperator{\tr}{tr}

\DeclareMathOperator{\clos}{Clos}
\DeclareMathOperator{\conv}{Conv}
\def \diracspacing {0.7pt}
\newcommand{\expec}[1]{\langle #1 \rangle}
 
\newcommand{\ket}[1]{| \hspace{\diracspacing} #1 \rangle} 
\newcommand{\braket}[2]{\langle #1 \hspace{\diracspacing} | \hspace{\diracspacing} #2 \rangle} 
 
\newcommand{\ketbra}[2]{| \hspace{\diracspacing} #1 \rangle \langle #2 \hspace{\diracspacing} |} 
 
\newcommand{\ketbraq}[1]{\ketbra{#1}{#1}} 
\newcommand{\bramatket}[3]{\langle #1 \hspace{\diracspacing} | #2 | \hspace{\diracspacing} #3 \rangle}

\newcommand{\abs}[2][]{#1| #2 #1|}
\newcommand{\cA}{\mathcal{A}}

\newcommand{\cF}{\mathcal{F}}

\newcommand{\cL}{\mathcal{L}}

\newcommand{\cN}{\mathcal{N}}

\newcommand{\cQ}{\mathcal{Q}}

\newcommand{\cS}{\mathcal{S}}

\newcommand{\sH}{\mathscr{H}}

\newcommand{\cAbnd}{\cA_{\textnormal{bnd}}}
\newcommand{\cAext}{\cA_{\textnormal{ext}}}
\newcommand{\cAexp}{\cA_{\textnormal{exp}}}
\newcommand{\cLdet}{\cL_{\textnormal{det}}}

\newcommand{\cQfin}{\cQ_{\textnormal{finite}}}
\usepackage{xspace}
\mathchardef\mhyphen="2D
\newcommand{\scen}[3][2]{\ensuremath{(#1\mathord{\scriptstyle^{_{\mhyphen}}}#2\mathord{\scriptstyle^{_{\mhyphen}}}#3)}}
\newcommand{\vecg}{\vec{g}}
\newcommand{\vecu}{\vec{u}}
\newcommand{\vecb}{\vec{B}}
\newcommand{\vecp}{\vec{P}}
\newcommand{\subfontsize}{\tiny}
\newcommand{\PCHSH}{\vecp_{\textnormal{\subfontsize CHSH}}}

\newcommand{\PPR}{\vecp_{\textnormal{\subfontsize PR}}}
\newcommand{\PPRk}[1]{\vecp_{\textnormal{\subfontsize PR, #1}}}
\newcommand{\Pdet}[1]{\vecp_{d, #1}}
\newcommand{\PL}[1]{\vecp_{\subfontsize L, #1}}
\newcommand{\PNE}{\vecp_{\textnormal{\subfontsize NE}}}
\newcommand{\Pnoise}{\vecp_{\textnormal{noise}}}
\newcommand{\PHardy}{\vecp_{\textnormal{\subfontsize Hardy}}}
\newcommand{\psiH}{ \psi_{\textnormal{H}} }
\newcommand{\realreal}{$\big( \sH_{A}, \sH_{B}, \rho_{AB}, \{ E_{a}^{x} \}, \{ F_{b}^{y} \} \big)$}
\newcommand{\idealreal}{$\big( \sH_{A'}, \sH_{B'}, \Psi_{A'B'}, \{ P_{a}^{x} \}, \{ Q_{b}^{y} \} \big)$}
\newcommand{\betaCHSH}{\beta_{\textnormal{\subfontsize CHSH}}}
\newcommand{\NS}{\cN\cS}

\usepackage{bm}
\usepackage{units}
\usepackage{soul,color}

\usepackage{microtype}
\let\savedegree\degree
\let\degree\relax
\usepackage[matha]{mathabx} 
\let\degree\savedegree
\begin{document}
\title{Geometry of the set of quantum correlations}
\author{Koon Tong Goh}
\email{ktgoh@u.nus.edu}
\affiliation{Centre for Quantum Technologies\(,\) National University of Singapore\(,\) Singapore\(,\) 117543}
\author{J{\k{e}}drzej Kaniewski}
\email{jkaniewski@math.ku.dk}
\affiliation{QMATH, Department of Mathematical Sciences, University of Copenhagen, Universitetsparken 5, 2100 Copenhagen, Denmark}
\author{Elie Wolfe}
\email{ewolfe@perimeterinstitute.ca}
\affiliation{Perimeter Institute for Theoretical Physics\(,\) Waterloo\(,\) Ontario\(,\) Canada\(,\) N2L 2Y5}
\author{Tam\'as V\'ertesi}
\affiliation{Institute for Nuclear Research\(,\) Hungarian Academy of Sciences\(,\) Debrecen\(,\) Hungary\(,\) 4001}
\author{Xingyao Wu}
\affiliation{Joint Center for Quantum Information and Computer Science\(,\) University of Maryland\(,\) College Park\(,\) MD 20742\(,\) USA}
\author{Yu Cai}
\affiliation{Centre for Quantum Technologies\(,\) National University of Singapore\(,\) Singapore\(,\) 117543}
\author{Yeong-Cherng Liang}
\email{ycliang@mail.ncku.edu.tw}
\affiliation{Department of Physics, National Cheng Kung University\(,\) Tainan\(,\) Taiwan\(,\) 701}
\author{Valerio Scarani}
\affiliation{Centre for Quantum Technologies\(,\) National University of Singapore\(,\) Singapore\(,\) 117543}
\affiliation{Department of Physics\(,\) National University of Singapore\(,\) Singapore\(,\) 117542}
\begin{abstract}
It is well known that correlations predicted by quantum mechanics cannot be explained by any classical (local-realistic) theory. The relative strength of quantum and classical correlations is usually studied in the context of Bell inequalities, but this tells us little about the geometry of the quantum set of correlations. In other words, we do not have good intuition about what the quantum set actually looks like. In this paper we study the geometry of the quantum set using standard tools from convex geometry. We find explicit examples of rather counter-intuitive features in the simplest non-trivial Bell scenario (two parties, two inputs and two outputs) and illustrate them using 2-dimensional slice plots. We also show that even more complex features appear in Bell scenarios with more inputs or more parties. Finally, we discuss the limitations that the geometry of the quantum set imposes on the task of self-testing.
\end{abstract}
\date{\today}
\maketitle
\section{Introduction}
Local measurements performed on entangled particles can give rise to correlations which are stronger than those present in any classical theory, a phenomenon known as Bell nonlocality. This seminal result, often referred to as Bell's theorem, was proven by Bell more than five decades ago~\cite{Bell:1964} and the existence of such correlations has recently been confirmed unequivocally in a couple of technologically-demanding experiments~\cite{Hensen:2015,Shalm:2015, Giustina:2015,Rosenfeld:2017}.

In addition to its fundamental significance Bell nonlocality has also found real-life applications, notably in secure communication, generation of certifiably secure randomness and more generally device-independent quantum information processing (see Ref.~\cite{Brunner:2014} for a review on Bell nonlocality and its applications). In the device-independent setting we do not have a complete description of our physical setup and draw conclusions based only on the observed correlations instead. Thorough understanding of the sets of correlations allowed by different physical theories is thus essential to comprehend the power of device-independent quantum information processing.

In this language, Bell's theorem simply states that the set of correlations allowed by quantum theory $\cQ$ is a \emph{strict superset} of the set of correlations allowed by classical theories $\cL$. The difference between these two objects is often investigated via \emph{Bell inequalities}, i.e.~linear constraints that must be satisfied by classical correlations, but may be violated by quantum mechanics. Usual examples include the inequalities derived by Clauser, Horne, Shimony and Holt (CHSH)~\cite{CHSH} and Mermin~\cite{Mermin:1990:PRL}.

Although quantum correlations may be stronger than their classical counterpart, they cannot be arbitrarily strong. In particular, they obey the no-signalling principle proposed by Popescu and Rohrlich~\cite{Popescu1994}. Imposing this principle alone gives the no-signalling set $\NS$, which turns out to be a {\em strict superset} of $\cQ$, i.e.~we arrive at the following well-known strict inclusions:
\begin{equation}
\label{eq:correlation-sets-inclusion}
\cL \subsetneq \cQ \subsetneq \NS.
\end{equation}
Bell scenarios are parametrised by the number of inputs and outputs at each site. In this work we assume that all these numbers are finite and then the local and no-signalling sets are polytopes, while the quantum set is a convex set, but not a polytope. It is known that all of them span the same affine space, i.e.~$\dim \cL = \dim \cQ = \dim \NS$~\cite{Tsirelson:1993, pironio05a}.

In the literature the relation between the three sets is sometimes represented by a simple diagram which consists of a circle sandwiched between two squares similar to Fig.~\ref{fig:circle-plot}, see e.g.~Ref.~\cite{Branciard:2011, Brunner:2014}. While this picture accurately represents a particular 2-dimensional slice (cross-section) of the quantum set, it does not capture all the intricacies related to its geometry, see e.g.~Refs.~\cite{Allcock:2009, Allcock:2009b, GYNI, Yang:2011, Fritz:2013, Putz2014, Chaves:2015aa, deVicente:2015, Christensen:2015, Lin:2017, Zhou:2017}. The quantum set is arguably the most important object in the field of quantum correlations and while some special subspaces are rather well-understood~\cite{Tsirelson:1987, Landau:1988, Werner:2001, Masanes:2006, Cabello:2005}, in general surprisingly little is known about its geometry~\cite{Tsirelson:1993} beyond the fact that it is convex~\cite{Pitowsky:Book}. In this work we explore the unusual features of the quantum set and use standard notions from convex geometry to formalise them. Understanding the geometry of the quantum set has immediate implications for (at least) two distinct lines of research.

The first one is related to the question of whether the quantum set admits a ``physical description'', i.e.~whether there exists a simple, physically-motivated principle which singles out precisely the quantum set without referring to operators acting on Hilbert spaces. Several such rules have been proposed \cite{Popescu1994, Navascues2009, Pawlowski2009, Fritz:2013, Navascus:2015, Zhou:2017} (the no-signalling principle being the first), but none of them has been shown to recover the quantum set. Checking whether these physical principles correctly reproduce the unusual features of the true quantum set, as was done in Refs.~\cite{Yang:2011,Fritz:2013,Chaves:2015aa, deVicente:2015, Christensen:2015, Zhou:2017, Lin:2017}, will give us a better understanding of their strengths and weaknesses and might help us in the search of the correct physical principle.

The geometry of the quantum set is also related to the task of self-testing \cite{Tsirelson:1987, Tsirelson:1993, Popescu:1992, Mayers:2004} in which we aim to deduce properties of the quantum system under consideration from the observed correlations alone. The fact that the geometry of the quantum set is more complex than that of the circle in Fig.~\ref{fig:circle-plot} imposes concrete limitations on our ability to make self-testing statements, which we discuss in the relevant sections.

In Section~\ref{sec:preliminaries} we define the three correlation sets and propose a geometrically-motivated classification of Bell functions. Sections~\ref{sec:faces-2222} and~\ref{sec:nonlocal-faces} contain examples of various unusual geometrical features of the quantum set. In Section~\ref{sec:conclusions} we summarise the most important findings and discuss some open questions. In the appendices, one can find (1) a simple proof that the quantum set in the simplest Bell scenario is closed, (2) connections between the CHSH Bell inequality violation and some distance measures in the same Bell scenario, (3) further examples of unusual slices of the quantum set, (4) tools that we have developed to identify unusual quantum faces,  (5) a proof that the optimal quantum distribution realising the Hardy paradox~\cite{Hardy1992} is not {\em exposed}, and (6) some other technical details of our results.
\section{Preliminaries}
\label{sec:preliminaries}
\begin{figure}[ht]
\includegraphics[width=8cm]{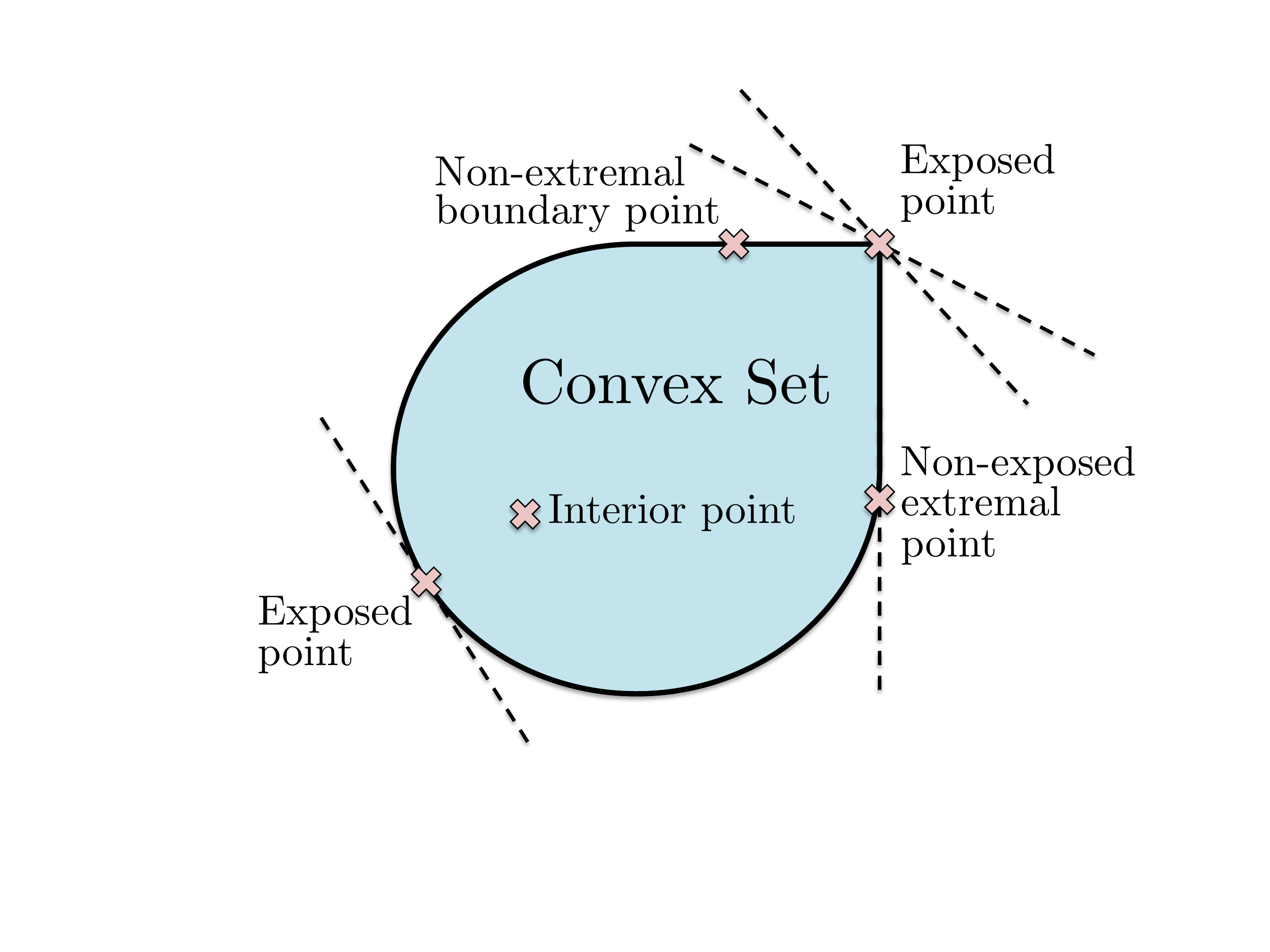}
\caption{Different types of points of a compact convex set.}
\label{fig:convex-set}
\end{figure}
While discussing the geometry of the quantum set it is natural to employ standard tools from convex geometry, e.g.~the notions of \textbf{exposed}, \textbf{extremal} and \textbf{boundary} points. Let us recall that for a compact convex set $\cA$ we have
\begin{equation}
\label{eq:hierarchy}
\cAexp \subseteq \cAext \subseteq \cAbnd \subseteq \cA
\end{equation}
and as shown in Fig.~\ref{fig:convex-set}
these inclusions are in general strict. We will also use the notion of an \textbf{exposed face} of a convex set. A short and self-contained introduction to these notions can be found in Appendix~\ref{app:convex-sets}.
\subsection{Three important correlation sets}
\label{sec:correlations-sets}
Following the convention of Ref.~\cite{Donohue:2015} we denote the Bell scenario of $n$ parties who (each) have $m$ measurement settings with $\Delta$ possible outcomes by $\scen[n]{m}{\Delta}$. In this work we focus predominantly on the bipartite case, i.e.~$n = 2$, and then the entire statistics can be assembled into a real vector $\vecp := \big( P(ab | xy) \big) \in \amsbb{R}^{ m^{2} \Delta^{2} }$, which we will refer to as the \textbf{behaviour}\footnote{Tsirelson first used the term ``behaviour''~\cite{Tsirelson:1993} to describe a family of probability distributions indexed by tuples of setting values. The term has become widely adopted, e.g.~in Refs.~\cite{Navascues:2008,Brunner:2014}. The terms ``box'' \cite{Barrett:2005:Feb,Barrett:2005:Sep}, ``probability model'' \cite{Acin:2015} and ``correlation'' \cite{Jones:2005,Barnum:2010} are also commonly used.}, \textbf{probability point} or simply a \textbf{point}. It is clear that all conditional probability distributions must be non-negative
\begin{equation}
\label{eq:positivity}
P(a b | x y) \geq 0 \quad \forall a, b, x, y
\end{equation}
and normalised
\begin{equation}
\label{eq:normalisation}
\sum_{ab} P(a b |x y) = 1 \quad \forall x, y.
\end{equation}
\subsubsection{The no-signalling set $\NS$}
A probability point belongs to the no-signalling set if it satisfies
\begin{equation}
\label{eq:no-signalling}
\begin{split}
    &\forall a, x, y, y' \; \sum_{b} P(a b |x y)=\sum_{b} P(a b| x y')\,\text{ and }\\
    &\forall b, x, x', y \; \sum_{a} P(a b |x y)=\sum_{a} P(a b| x' y)\,.
\end{split}
\end{equation}
The term no-signalling~\cite{Barrett:2005:Feb} refers to the fact that the choice of local settings of one party does not affect the outcome distribution of the other party. We denote the set of all no-signalling behaviours by $\NS$ and since it is characterised by a finite number of linear inequalities and equalities, namely \eqref{eq:positivity}, \eqref{eq:normalisation} and \eqref{eq:no-signalling}, the no-signalling set is a polytope.
\subsubsection{The quantum set $\cQ$}
\label{sec:quantum-set}
The quantum set $\cQ$ is the set of correlations which can be achieved by performing local measurements on quantum systems. Following the standard tensor-product paradigm each party is assigned a Hilbert space $\sH$ of finite dimension $d := \dim(\sH) < \infty$. A valid quantum state corresponds to a $d^{2} \times d^{2}$ matrix which is positive semidefinite and of unit trace. A local measurement with $\Delta$ outcomes is a decomposition of the $d$-dimensional identity into positive semidefinite operators, i.e.~$\{ E_{a} \}_{a = 1}^{\Delta}$ such that $E_{a} \geq 0$ for all $a$ and
\begin{equation}
\sum_{a = 1}^{\Delta} E_{a} = \mathbb{1}_{d},
\end{equation}
where $\mathbb{1}_{d}$ denotes the $d$-dimensional identity matrix.

We define $\cQfin$ to be the set of behaviours which can be generated when local Hilbert spaces are finite-dimensional, i.e.~$\vecp \in \cQfin$ if there exists a finite-dimensional quantum state $\rho$ and measurements $\{ E_{a}^{x} \}, \{ F_{b}^{y} \}$ such that
\begin{equation}
P(ab|xy) = \tr \big[ ( E_{a}^{x} \otimes F_{b}^{y} ) \rho \big]
\end{equation}
for all $a, b, x, y$.

To make the underlying mathematics neater, we define the quantum set $\cQ$ as the closure of $\cQfin$, i.e.~we explicitly include all the limit points, which makes the quantum set $\cQ$ compact.\footnote{A detailed treatment of this issue, which is of secondary importance in the context of this work, can be found in Appendix~\ref{app:quantum-set}.} The fundamental result that $\cQ \neq \cQfin$ for some finite Bell scenarios was only recently established by \citet{Slofstra:2017} (see also the recent work of Dykema et al.~\cite{dykema17a}).
\subsubsection{The local set $\cL$}

We call a probability point \textbf{deterministic} if the output of each party is a (deterministic) function of their input and we denote the set of deterministic points by $\cLdet$. The local set is defined as the convex hull of the deterministic points
\begin{equation*}
	\cL := \conv\left[ \cLdet \right].
\end{equation*}
Since $\cLdet$ is a finite set, the local set $\cL$ is a polytope.

\subsection{Bell functions and exposed faces of the correlation sets}
\label{sec:bell-functions}
A \textbf{Bell function} is a real vector $\vec{B} \in \amsbb{R}^{ m^{2} \Delta^{2} }$ and the value of the function on a specific behaviour $\vecp$ is simply the inner product $\vecb \cdot \vecp$. For a given correlation set $\cS = \cL, \cQ, \NS$ we denote the maximum value of the Bell function by
\begin{equation*}
\beta_{\cS}(\vecb) := \max_{\vecp \in \cS} \vecb \cdot \vecp.
\end{equation*}
Note that since all these sets are compact, the maximum is always achieved. To simplify the notation we will simply write $\beta_{\cS}$ whenever the Bell function is clear from the context.

Bell functions are useful for studying the three correlation sets, but they suffer from the problem of non-uniqueness, i.e.~the same function can be written in multiple ways which are not always easily recognised as equivalent (see, however, Ref.~\cite{Rosset2014a}). To overcome this obstacle instead of studying the inequalities we study the (exposed) faces they give rise to. For a correlation set $\cS = \cL, \cQ, \NS$ every Bell function $\vecb$ identifies a face
\begin{equation*}
\cF_{\cS}(\vecb) := \{ \vecp \in \cS : \vecb \cdot \vecp = \beta_{\cS} \}.
\end{equation*}
All the sets considered here are compact, so the face is always non-empty, i.e.~it contains at least one point. Since the local set $\cL$ and the no-signalling set $\NS$ are polytopes, all their faces are also guaranteed to be polytopes, while for the quantum set $\cQ$ this is not necessarily the case.

The dimension of the face $\cF_{\cS}(\vecb)$ is simply the dimension of the affine subspace spanned by the points in $\cF_{\cS}(\vecb)$. While the dimensions of local and no-signalling faces are easy to compute (we simply find which vertices saturate the maximal value and then check how many of them are affinely independent), there is no generic way of computing the dimension of a quantum face. However, if the quantum value coincides with either the local or the no-signalling value, an appropriate bound follows directly from the set inclusion relation~\eqref{eq:correlation-sets-inclusion}:
\begin{align}
\label{eq:dimQlb}
\beta_{\cQ}(\vecb) = \beta_{\cL}(\vecb) &\implies \dim \big( \cF_{\cQ}(\vecb) \big) \geq \dim \big( \cF_{\cL}(\vecb) \big)
\\
\shortintertext{and}
\label{eq:dimQub}
\beta_{\cQ}(\vecb) = \beta_{\NS}(\vecb) &\implies \dim \big( \cF_{\cQ}(\vecb) \big) \leq \dim \big( \cF_{\NS}(\vecb) \big).
\end{align}
A \textbf{flat boundary region} is an exposed face which contains more than a single point, but is strictly smaller than the entire set.\footnote{One should take care not to confuse the face $\cF_{\cS}(\vecb)$ with the hyperplane $\vecb \cdot \vecp = \beta_{\cS}$: the face  $\cF_{\cS}(\vecb)$ is the intersection of $\cS$ with the hyperplane.}

Focusing on faces, rather than Bell functions, reduces the undesired ambiguity; the following example shows that multiple Bell functions may give rise to the same face. Let $\vecp$ be a deterministic point and recall that for such a point all the conditional probabilities are either zero or one. Consider a Bell function $\vecb$ whose coefficients satisfy
\begin{equation}
\label{eq:exposing-family}
B(abxy)
\begin{cases}
    > 0 \nbox{if} P( ab | xy ) = 1,\\
    = 0 \nbox{otherwise.}
\end{cases}
\end{equation}
Every Bell function in this (continuous) family has a unique maximiser, which is precisely the point $\vecp$ (regardless of the choice of the correlation set). This simple example illustrates that (i) multiple Bell functions can have precisely coinciding sets of maximisers (i.e.~they give rise to the same face) and (ii) all deterministic probability points are exposed in all three sets $\cL, \cQ$ and $\NS$ (see Appendix~\ref{app:convex-sets} for the definition of an exposed point).

Bell functions can be split into $4$ (disjoint) classes by comparing the maximal values over the three correlation sets:
\begin{compactenum}
    \item[(1)] $\beta_{\cL} < \beta_{\cQ} < \beta_{\NS}$: all three values differ,
    \item[(2)] $\beta_{\cL} = \beta_{\cQ} < \beta_{\NS}$: only local and quantum values coincide,
    \item[(3)] $\beta_{\cL} < \beta_{\cQ} = \beta_{\NS}$: only quantum and no-signalling values coincide,
    \item[(4)] $\beta_{\cL} = \beta_{\cQ} = \beta_{\NS}$: all three values coincide.
\end{compactenum}
Whenever two of these values coincide, we can make the classification finer by asking whether the resulting faces coincide or not. In the list below we fine-grain the four Bell-value classes into nine face-comparison classes, using $=$ vs.~$\subsetneq$ to distinguish exact coincidence from strict containment. Enumerating all possible cases leads to:
\begin{compactenum}
    \item[(1)\hphantom{a}] $\beta_{\cL} < \beta_{\cQ} < \beta_{\NS}\,$,
    \item[(2a)] $\beta_{\cL} = \beta_{\cQ} < \beta_{\NS}$ and $\cF_{\cL} \subsetneq \cF_{\cQ}\,$,
    \item[(2b)] $\beta_{\cL} = \beta_{\cQ} < \beta_{\NS}$ and $\cF_{\cL} = \cF_{\cQ}\,$,
    \item[(3a)] $\beta_{\cL} < \beta_{\cQ} = \beta_{\NS}$ and $\cF_{\cQ} \subsetneq \cF_{\NS}\,$,
    \item[*(3b)] $\beta_{\cL} < \beta_{\cQ} = \beta_{\NS}$ and $\cF_{\cQ} = \cF_{\NS}\,$,
    \item[(4a)] $\beta_{\cL} = \beta_{\cQ} = \beta_{\NS}$ and  $\cF_{\cL} \subsetneq \cF_{\cQ} \subsetneq \cF_{\NS}\,$,
    \item[(4b)] $\beta_{\cL} = \beta_{\cQ} = \beta_{\NS}$ and  $\cF_{\cL} = \cF_{\cQ} \subsetneq \cF_{\NS}\,$, 
    \item[*(4c)]  $\beta_{\cL} = \beta_{\cQ} = \beta_{\NS}$ and  $\cF_{\cL} \subsetneq \cF_{\cQ} = \cF_{\NS}\,$,
    \item[(4d)] $\beta_{\cL} = \beta_{\cQ} = \beta_{\NS}$ and $\cF_{\cL} = \cF_{\cQ} = \cF_{\NS}\,$.
\end{compactenum}
It turns out that two (those marked with $*$) among these nine cases are actually \emph{not} realisable. Indeed, no nonlocal vertex of the no-signalling polytope can be obtained by measuring quantum systems~\cite{Ramanathan:2016}, which places some restrictions on the cases satisfying $\beta_{\cQ} = \beta_{\NS}$. The face $\cF_{\NS}$ is the convex hull of some vertices of the no-signalling polytope. If all these vertices are local, we must have $\cF_{\cL} = \cF_{\NS}$, which puts us in the class (4d). On the other hand, if there is at least one nonlocal vertex, this vertex does not belong to the quantum set. This immediately implies that $\cF_{\cQ} \subsetneq \cF_{\NS}$, which eliminates classes (3b) and (4c). All the remaining classes exist and some of them we discuss in detail, e.g.~class (1) in Section~\ref{sec:b1}, (2a) in Section~\ref{sec:b2}, (2b) in Section~\ref{sec:b3} (see the generic cases considered therein), (4a) in Section~\ref{app:positivity} and (4b) in Section~\ref{sec:b5}. Class (3a) does not appear in the \scen{2}{2} scenario (see Appendix~\ref{app:no-3a-in-2222} for a proof), but it is easy to check that the magic square game~\cite{cleve04a} or the (tripartite) Mermin inequality~\cite{Mermin:1990:PRL} belong precisely to that class. Class (4d) is represented by the family of inequalities defined in Eq.~\eqref{eq:exposing-family} (although the resulting face is just a single point; a 1-dimensional example of this kind is given in Appendix~\ref{app:type-4d}).
\subsection{Self-testing of quantum systems}
Some probability points in the quantum set have the surprising property that there is essentially only one way of realising them in quantum mechanics, a phenomenon known as \emph{self-testing}. In this paper we do not prove any new self-testing results, but we provide explicit examples where the usual self-testing statements cannot be made. Moreover, we prove a relation between self-testing and extremality (see Appendix~\ref{app:self-testing} for details), which we use to deduce that certain points are extremal in the quantum set.
\section{Faces of the quantum set in the \scen{2}{2} scenario}
\label{sec:faces-2222}
In this paper we focus predominantly on the simplest non-trivial Bell scenario, i.e.~the case of $m = \Delta = 2$. It is well known~\cite{Fine1982} that in this scenario the local set is fully described by the positivity inequalities~\eqref{eq:positivity}, no-signalling constraints~\eqref{eq:no-signalling} and 8 additional inequalities, which are all equivalent (up to permutations of inputs and outputs) to the CHSH inequality~\cite{CHSH}. The existence of a single type of (facet) Bell inequalities and the fact that any no-signalling probability point $\vecp \in \NS$ can violate at most one of these inequalities~\cite{Liang:2010} means that we can interpret the CHSH violation as a measure of distance from the local set. More specifically, Bierhorst showed that the total variation distance from the local set and the local content~\cite{Scarani2008} can be written as linear functions of the violation~\cite{bierhorst16a}. In Appendix~\ref{app:chsh-violation-distance-measure} we show that the same property holds for various notions of \emph{visibility}.

The structure of the quantum set turns out to be significantly more complex. Let us start by introducing convenient notation for the \scen{2}{2} scenario. Correlations in the \scen{2}{2} scenario are described by vectors $\vecp \in \amsbb{R}^{16}$, but due to the no-signalling constraints these vectors span only an 8-dimensional subspace and it is convenient to use a representation which takes advantage of this dimension reduction. Following the convention $a, b \in \{0, 1\}$ we define the local marginals as
\begin{align*}
\expec{ A_{x} } &:= P(a = 0 | x) - P(a = 1 | x),\\
\expec{ B_{y} } &:= P(b = 0 | y) - P(b = 1 | y)
\end{align*}
and the correlators as
\begin{equation*}
\expec{ A_{x} B_{y} } := P(a = b | x y) - P(a \neq b | x y).
\end{equation*}
The inverse relation is given by
\begin{align}
\label{Eq:CorToProb}
    P(&a b|x y)=\\\nonumber
    &\frac{1}{4} \big(1 + (-1)^a \expec{A_{x}} + (-1)^b \expec{B_{y}}+ (-1)^{a+b} \expec{A_{x} B_y} \big).
\end{align}
While this transformation is valid for any no-signalling point, the notation is inspired by quantum mechanics, since for a quantum behaviour the local marginals and correlators are simply expectation values ($\expec{X} = \tr (X \rho)$) of the local observables
\begin{align*}
A_{x} &= E_{0}^{x} - E_{1}^{x},\\
B_{y} &= F_{0}^{y} - F_{1}^{y}
\end{align*}
and their products. The expectation values are conveniently represented in a table
\begin{equation}
    \vecp =
    \begin{array}{c|c|c}
        & \expec{B_0} &  \expec{B_1} \\ \hline
        \expec{A_0} &  \expec{A_0B_0} &  \expec{A_0B_1} \\ \hline
        \expec{A_1} &  \expec{A_1B_0} &  \expec{A_1B_1}
    \end{array} \;
\end{equation}
and it is natural to use the same representation when writing down Bell functions. It is worth pointing out that the coordinate transformation that takes us from the conditional probabilities of events $P(ab|xy)$ to the local marginals ($\expec{A_{x}}, \expec{B_{y}}$) and correlators ($\expec{A_{x} B_{y}}$) is a linear transformation, but it is not isometric. In other words, the transformation does not change any qualitative features of the set, e.g.~whether a point is extremal or exposed, but it might affect measures of distance or volume. In order to make this transformation isometric we would need to define our coordinate system as $( c_{1} \expec{A_{x}}, c_{1} \expec{B_{y}}, c_{2} \expec{A_{x} B_{y}} )$ for suitably chosen constants $c_{1}, c_{2}$.

In the remainder of the section we look at various quantum faces in the \scen{2}{2} scenario ordered according to the classification introduced in Section~\ref{sec:bell-functions}.
\subsection{A quantum face with $\beta_{\cL} < \beta_{\cQ} < \beta_{\NS}$}
\label{sec:b1}
Our first example is the CHSH Bell function~\cite{CHSH}, which reads
\begin{equation}
\label{eq:b1}
    \vecb_{1} :=
    \begin{array}{r|r|r}
        &  0 &  0\\ \hline
        0 &  1 &  1\\ \hline
        0 &  1 &  -1
    \end{array} \; , \;
    \vecb_{1} \cdot \vecp \leq
    \begin{cases}
        2 & \cL\\
        \bm{2\sqrt{2}} & \cQ\\
        4 & \NS
    \end{cases}.
\end{equation}
This inequality is known to have a unique quantum maximiser~\cite{McKague:2012, Brunner:2014}
\begin{equation}
\label{eq:pCHSH}
	\PCHSH :=
	\begin{array}{r|r|r}
	    &  0 &  0\\ \hline
	    0 & \tfrac{1}{\sqrt{2}} & \tfrac{1}{\sqrt{2}}\\ \hline
	    0 & \tfrac{1}{\sqrt{2}} & -\tfrac{1}{\sqrt{2}}
	\end{array} \, ,
\end{equation}
which implies that $\PCHSH$ is an exposed point of the quantum set. In Fig.~\ref{fig:circle-plot} we show $\PCHSH$ in the 2-dimensional slice spanned by 4 variants of the Popescu-Rohrlich (PR) box \cite{Barrett:2005:Feb, Barrett:2005:Sep, Bancal:2014, Brunner:2014, Ramanathan:2016}
\begin{equation}
\label{eq:PR1234}
\begin{split}
	\PPR &:=
	\begin{array}{r|r|r}
	    & 0 &  0\\ \hline
	    0 &  \phantom{-} 1 & 1\\ \hline
	    0 &  1 & -1
	\end{array} \; , \quad
	\PPRk{2} :=
	\begin{array}{r|r|r}
	    &   0 &  0\\ \hline
	    0 &  - 1 & - 1\\ \hline
	    0 &  - 1 & 1
	\end{array} \; ,\\
	\PPRk{3} &:=
	\begin{array}{r|r|r}
	    &   0 &  0\\ \hline
	    0 &  -1 & \phantom{-} 1\\ \hline
	    0 &  1 & 1
	\end{array} \; , \quad
	\PPRk{4} :=
	\begin{array}{r|r|r}
	    &   0 &  0\\ \hline
	    0 &  1 & -1\\ \hline
	    0 &  -1 & -1
	\end{array} \; .
\end{split}	
\end{equation}
\begin{figure}[ht]
    \includegraphics[width=9cm]{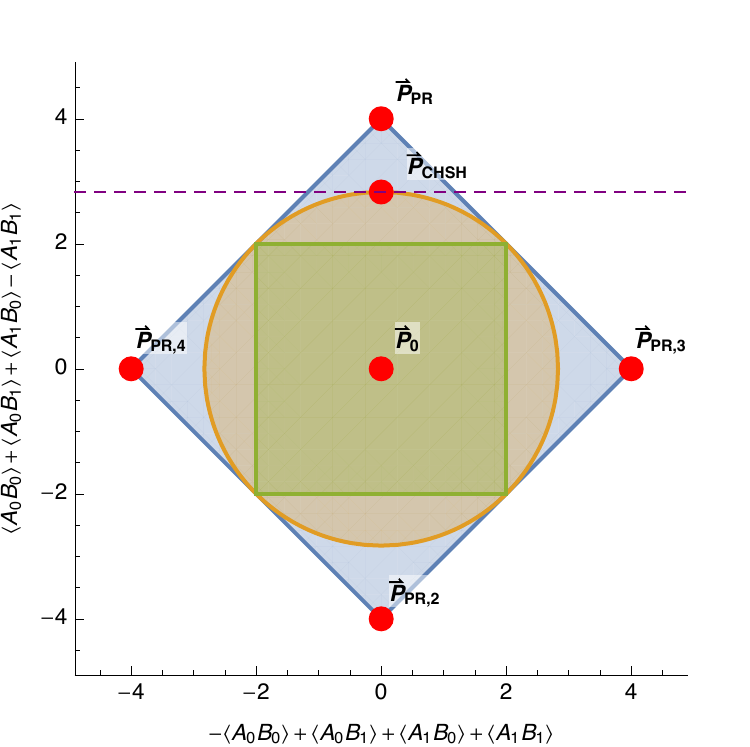}
    \caption{A 2-dimensional slice in which the quantum set exhibits no flat boundaries (first presented as Fig.~1 in Ref.~\cite{Branciard:2011}). Points on this slice can also be conveniently parametrised by two different versions of the CHSH Bell functions as in Fig.~3 of Ref.~\cite{Christensen:2015}.}
    \label{fig:circle-plot}
\end{figure}
The central point of this plot corresponds to the uniformly random distribution, which can be written as a uniform mixture of the 4 PR boxes, i.e.
\begin{equation}
\vecp_{0} :=
    \begin{array}{r|r|r}
        & 0 & 0\\ \hline
        0 & 0 & 0\\ \hline
        0 & 0 & 0
    \end{array} \, .
\end{equation}
\subsection{Quantum faces with $\beta_{\cL} = \beta_{\cQ} < \beta_{\NS}$} 
In this section we consider Bell functions satisfying $\beta_{\cL} = \beta_{\cQ}$, which implies that the corresponding quantum faces will contain some local points.
\subsubsection{Quantum faces with $\cF_{\cL} \subsetneq \cF_{\cQ}$ containing the CHSH point}
\label{sec:b2}
Consider the Bell function
\begin{equation}
\label{eq:b2}
\vecb_{2} :=
    \begin{array}{c|c|c}
        & 1 \mathopen{-} \sqrt{2} &  1 \\ \hline
        1 \mathopen{-} \sqrt{2} &  \sqrt{2} &  \sqrt{2} \\ \hline
        1 &  \sqrt{2} & -\sqrt{2}
    \end{array} \; , \;
	\vecb_{2} \cdot \vecp \leq
    \begin{cases}
        4 & \cL\\
        \bm{4} & \cQ\\
        4\sqrt{2} & \NS
    \end{cases},
\end{equation}
where the local and no-signalling bounds have been computed by enumerating the vertices of the polytopes, while the quantum bound has been computed using the analytic technique of Wolfe and Yelin~\cite{Wolfe:2012}. The quantum bound is saturated by $\PCHSH$ but also by the deterministic point
\begin{equation}
\label{eq:Pdet1}
    \Pdet{1} :=
    \begin{array}{r|r|r}
        & 1 &  1\\ \hline
        1 &  1 &  1\\ \hline
        1 &  1 &  1
    \end{array} \; .
\end{equation}
This implies that the resulting quantum face is \emph{at least} 1-dimensional and we conjecture that this lower bound is actually tight, i.e.~that the quantum face is a line. Fig.~\ref{fig:pchsh-pd-slice} shows this quantum face in the slice containing $\PCHSH$, $\vecp_{0}$ and $\Pdet{1}$ (the same feature was presented in Fig.~3(c) of Ref.~\cite{Donohue:2015}).
\begin{figure}[h!t]
    \includegraphics[width=9cm]{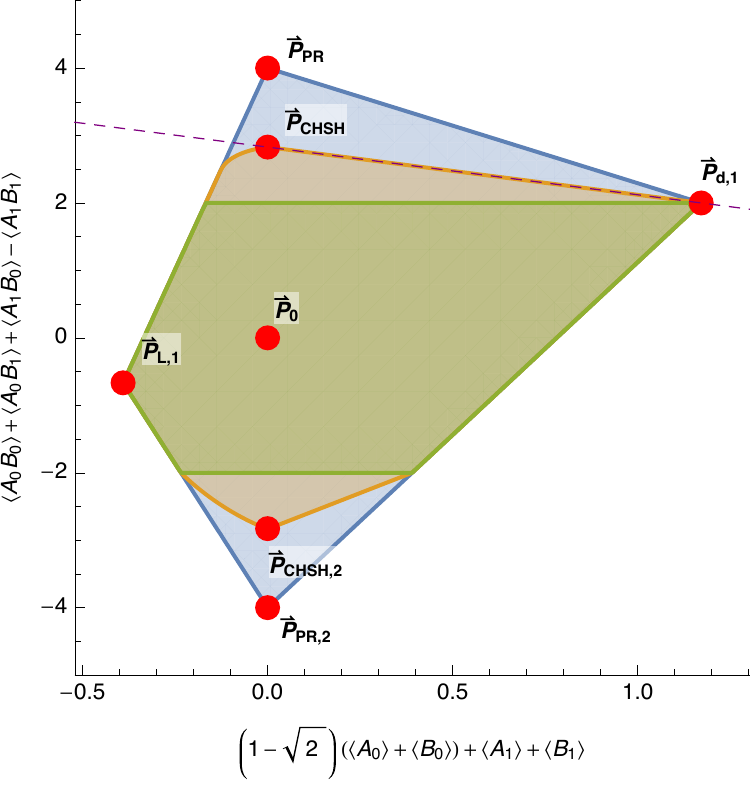}
    \caption{A slice containing $\PCHSH$, $\vecp_{0}$ and $\Pdet{1}$ is singled out by the following six equations: ${\expec{A_0}=\expec{A_1}=\expec{B_0}=\expec{B_1}}$, ${\expec{A_0 B_0}=\expec{A_0 B_1}=\expec{A_1 B_0}}$ and ${\expec{A_0}+\expec{A_1}+\expec{B_0}+\expec{B_1}=2 ( \expec{A_0 B_0}+\expec{A_1 B_1}} )$. The point $\PL{1}$ is given by $\expec{A_{x}} = \expec{B_{y}}=\expec{A_{x} B_{y}} = -\nicefrac{1}{3}$. Apart from the flat quantum face  $\cF_{\cQ}(\vecb_2)$ that connects $\PCHSH$ and $\Pdet{1}$, our numerical results suggest a few other flat regions on the boundary of (this slice of) the quantum set.}
    \label{fig:pchsh-pd-slice}
\end{figure}

\begin{figure}[ht]
    \includegraphics[width=9cm]{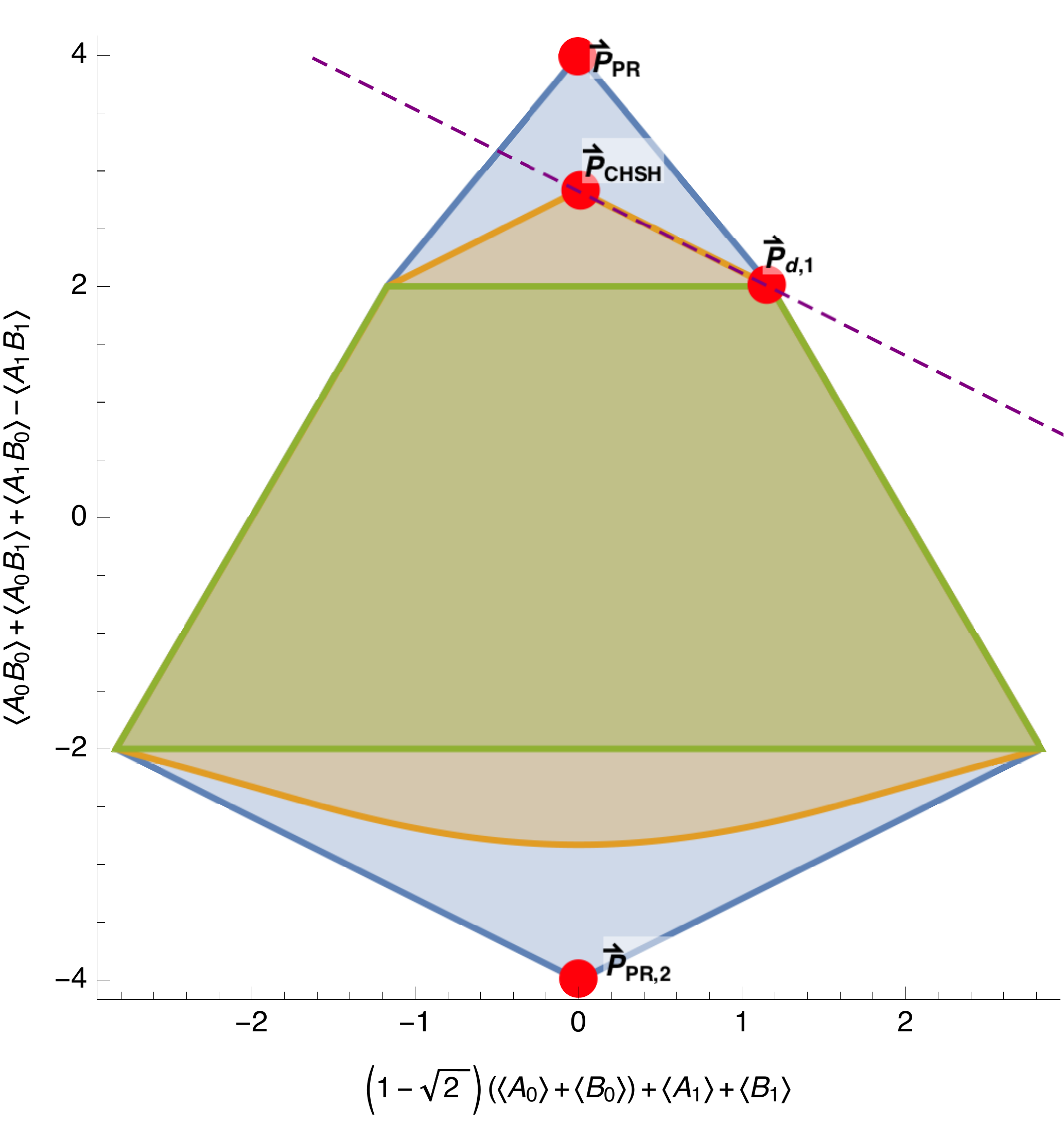}
    \caption{Projection illustrating the flat boundary identified by $\vecb_1$. While every point in a slice plot corresponds to precisely one behaviour, a point in a projection plot may simultaneously represent multiple behaviors. In this projection the behaviours $\PCHSH$ and $\Pdet{1}$ are the only behaviours that lie on the points $( 0, 2\sqrt{2} )$ and $( 4 - 2\sqrt{2}, 2 )$, respectively. We see that all three correlation sets are symmetric with respect to the reflection about the $x = 0$ line. This symmetry arises from the relabelling of the outcomes of all measurements, which flips the marginals, but leaves the correlators unchanged.}
    \label{fig:pchsh-pd-proj}
\end{figure}

The above quantum face is not the only flat face containing $\PCHSH$. To see this we swap the outcomes of all the measurements: this results in flipping the horizontal axis of Fig.~\ref{fig:pchsh-pd-slice} while leaving the vertical axis unchanged. This relabelling transforms the Bell function $\vecb_{2}$ to
\begin{equation*}
    \vecb_{2^{*}} :=
    \begin{array}{c|c|c}
        & \sqrt{2} \mathopen{-} 1 &  -1 \\ \hline
        \sqrt{2} \mathopen{-} 1 &  \sqrt{2} &  \sqrt{2} \\ \hline
        -1 &  \sqrt{2} & -\sqrt{2}
    \end{array} \, , \;
    \vecb_{2^{*}} \cdot \vecp \leq
    \begin{cases}
        4 & \cL\\
        \bm{4} & \cQ\\
        4\sqrt{2}  & \NS
    \end{cases},
\end{equation*}
whose quantum value is achieved by $\PCHSH$ and
\begin{equation}
\label{eq:Pdet2}
    \Pdet{2} :=
    \begin{array}{r|r|r}
        & -1 &  -1\\ \hline
        -1 &  1 &  1\\ \hline
        -1 &  1 &  1
    \end{array} \; .
\end{equation}
This immediately implies that a \emph{projection} plot of $\cQ$ using the same set of axes must look different from the \emph{slice} plot, a difference which is clearly seen in Figs.~\ref{fig:pchsh-pd-slice} and~\ref{fig:pchsh-pd-proj}. More generally, we can apply a suitably chosen relabelling of the inputs and/or outputs to obtain an equivalent face that connects $\PCHSH$ with any of the 8 deterministic points which saturate the local value of the CHSH inequality given in Eq.~\eqref{eq:b1}.

It is worth pointing out that the \emph{non-extremal} probability points lying on these quantum faces \emph{cannot} be obtained by performing measurements on two-qubit states. In other words, only the $\PCHSH$ and the corresponding local deterministic point are quantumly achievable using local Hilbert spaces of dimension $2$ --- as can be verified using the technique of~\citet[Appendix D]{Donohue:2015}. Aside from this example, there are many more additional faces of the quantum set which contain a deterministic point and an extremal nonlocal point. For a systematic method of finding them please refer to Appendix~\ref{app:identifying-certifying}.
\subsubsection{Higher-dimensional quantum faces with $\cF_{\cL} \subsetneq \cF_{\cQ}$}
\label{sec:b3}
Interestingly, higher-dimensional quantum faces containing nonlocal points can also be found in this simplest Bell scenario. Consider, for example, the following two-parameter family of Bell functions: 
\begin{align}
\nonumber
    &\vecb_{3} :=
    \begin{array}{c|c|c}
        & -a &  1 \\ \hline
        -a &  c &  c \\ \hline
        1 & c & -(c + 1 - 2a)
    \end{array} \, ,\\
	\label{eq:b3}
	&\vecb_{3} \cdot \vecp \leq
    \begin{cases}
        2c \mathopen{+} 1 & \cL\\
        \bm{2c \mathopen{+} 1} & \cQ\\
        4c \mathopen{+} 1 \mathopen{-} 2 a & \NS
    \end{cases},
\end{align}
where $a \in [0, 1]$ and $c \in [a, c_{\textnormal{max}}]$ with $c_{\textnormal{max}}$ being---for any given value of $a$---the largest value of $c$ for which the following inequality holds:
\begin{align}
\nonumber
(c-2 a+1) &\left(2 a^3-3 a^2+(3 a-1) c^2-5 (a-1) a c-c^3\right)\\
\geq &\frac{a^2}{4 c^2 } \left(-2 a^2+3 (a-1) c+a+c^2\right)^2.\label{eq:ACregion}
\end{align}
The quantum bound of Eq.~\eqref{eq:b3} is saturated by 3 deterministic points: $\Pdet{1}$ given in Eq.~\eqref{eq:Pdet1} and
\begin{equation}
\label{eq:Pdet34}
    \Pdet{3} :=
    \begin{array}{r|r|r}
        &  -1 &  1\\ \hline
        -1 &  1 & -1\\ \hline
        -1 & 1 & -1
    \end{array} \; ,\quad
    \Pdet{4} :=
    \begin{array}{r|r|r}
        &  -1 & -1\\ \hline
        -1 &  1 &  1\\ \hline
        1 & -1 & -1
    \end{array} \; .
\end{equation}
Note that $\Pdet{3}$ and $\Pdet{4}$ are related by simply swapping Alice and Bob, i.e.~transposing the matrix given in Eq.~\eqref{eq:Pdet34}.

Interestingly, for generic pairs $(a, c)$, the quantum inequality is saturated \emph{exclusively} by local points. On the other hand when we consider special pairs of $(a, c)$ at the limit of the region constrained by Eq.~\eqref{eq:ACregion}, we find that the quantum bound is saturated by an additional extremal nonlocal point. Table~\ref{table:b3} lists several functions from this \emph{one-parameter} family together with some properties of the extremal nonlocal maximisers.
\begin{center}
\begin{table}[h!]
\begin{tabular}{lccccc}
$a$ & $c$ & $\beta_{\cQ}$ & $\betaCHSH$ & $\lambda$ & $\phi$\\
\hline
 0.1 & 0.325 & 1.649 & 2.719 & 0.558 & $106.945\degree$ \\
 0.2 & 0.592 & 2.185 & 2.769 & 0.525 & $103.131\degree$ \\
 0.280776 & 0.781 & 2.562 & 2.792 & $\bf{0.5}$ & $100.358\degree$ \\
 0.5 & 1.193 & 3.386 & 2.808 & 0.564 & $93.091\degree$ \\
 0.586417 & 1.318 & 3.635 & 2.798 & 0.591 & $\bf{90\degree}$ \\
 0.6 & 1.335 & 3.670 & 2.795 & 0.596 & $89.486\degree$ \\
 0.7 & 1.444 & 3.889 & 2.766 & 0.631 & $85.336\degree$ \\
 0.8 & 1.510 & 4.020 & 2.711 & 0.674 & $80.075\degree$ \\
 0.846074 & 1.519 & 4.038 & 2.671 & 0.699 & $76.924\degree$ \\
 0.9 & 1.502 & 4.004 & 2.599 & 0.737 & $72.036\degree$ \\
 0.99 & 1.272 & 3.543 & 2.249 & 0.878 & $50.291\degree$ \\
 1 & 1 & 3 & n/a & n/a & n/a \\
\end{tabular}
\caption{Each row corresponds to a Bell function from the family defined in Eq.~\eqref{eq:b3}. The parameters $a$ and $c$, chosen to saturate Eq.~\eqref{eq:ACregion}, determine the quantum bound $\beta_{\cQ}$ of the function. The quantum face identified by each function consists of (at least) 4 extremal points: three deterministic points and one nonlocal point. The remaining 3 columns contain information about the nonlocal maximiser: the CHSH violation $\betaCHSH$, the entanglement of the optimal two-qubit state $\lambda$ (quantified by the square of the larger Schmidt coefficient, e.g.~$\lambda = 1/2$ corresponds to the maximally entangled state) and the optimal angle $\phi$ between the two local observables ($\phi = 90\degree$ corresponds to maximally incompatible observables). The optimal angle is always the same for both parties, which reflects the symmetry of the Bell function. The final row corresponds to a Bell function saturated exclusively by local points, i.e.~there is no nonlocal maximiser. It is interesting to examine the trends in the properties of the nonlocal maximiser. The maximally entangled state appears only for $a=\frac{\sqrt{17}-3}{4}\approx 0.280776$, with the entanglement of the state dropping monotonically as $a$ varies away from that special value in either direction. The angle between the observables decreases monotonically as $a$ increases while maximal incompatibility is observed for $a = \frac{1}{3} ( 14 \sqrt{13} \, \sin z - 23 )$ where $z := \frac{1}{6} \big( \pi - 2 \tan ^{-1} \big( 87 \sqrt{3}/4591 \big) \big)$, i.e.~$a \approx 0.586417$. Finally, the CHSH violation initially increases and then goes down, peaking at $\betaCHSH \approx 2.810$ for $a \approx 0.45$.}
\label{table:b3}
\end{table}
\end{center}

For non-maximal values of $c$ we have $\cF_{\cL} = \cF_{\cQ}$, whereas for maximal $c$ the quantum face extends into an additional dimension beyond the local subspace, i.e. $\cF_{\cL} \subsetneq \cF_{\cQ}$. The increase in the dimension can be easily seen by noting that all three local points saturating Eq.~\eqref{eq:b3} give the CHSH value of 2, which is exceeded by the additional quantum point. As we further increase $a$, we reach the values $a = c = 1$, which corresponds to a linear combination of positivity facets.
\subsection{Bell functions with $\beta_{\cL} = \beta_{\cQ} = \beta_{\NS}$} 
In Section~\ref{sec:bell-functions} we have already seen a family of Bell functions with $\beta_{\cL} = \beta_{\cQ} = \beta_{\NS}$ for which the resulting faces are identical $\cF_{\cL} = \cF_{\cQ} = \cF_{\NS}$. In this section we present examples of the remaining two classes.
\subsubsection{A Bell function satisfying $\cF_{\cL} \subsetneq \cF_{\cQ} \subsetneq \cF_{\NS}$ containing an extremal but non-exposed point}
Consider the Bell function
\begin{equation*}
    \vecb_{4} :=
    \begin{array}{c|c|c}
        & 0 & 0 \\ \hline
        0 & 0 & 0 \\ \hline
        0 & 0 & -1
    \end{array} \; , \;
    \vecb_{4} \cdot \vecp \leq
    \begin{cases}
        1 & \cL\\
        \bm{1} & \cQ\\
        1 & \NS
    \end{cases}.
\end{equation*}
We do not have an analytic characterisation of the corresponding quantum face, but we can show that it is not a polytope. More specifically, we show that already in the slice of unbiased marginals $\expec{A_0} = \expec{A_1} = \expec{B_0} = \expec{B_1} = 0$ the quantum face has an infinite number of extremal points. The analytic characterisation of the quantum set in the correlator space due to Tsirelson, Landau and Masanes~\cite{Tsirelson:1993, Landau:1988, Masanes:2005} states that correlators $\expec{ A_{x} B_{y} } $ belong to the quantum set if and only if
\begin{equation}
\label{eq:tlm-constraint}
1 + \prod_{xy} \expec{ A_{x} B_{y} } + \prod_{xy} \sqrt{ 1 - \expec{ A_{x} B_{y} }^{2} } - \frac{1}{2} \sum_{xy} \expec{ A_{x} B_{y} }^{2} \geq 0,
\end{equation}
where the sums and products go over $x, y \in \{0, 1\}$.\footnote{This elegant and symmetric form is obtained by simply squaring the inequality derived by Landau~\cite{Landau:1988}.} The quantum set in the correlator space is a \emph{projection} of the (full) quantum set onto the coordinates $( \expec{ A_{0} B_{0} }, \expec{ A_{0} B_{1} }, \expec{ A_{1} B_{0} }, \expec{ A_{1} B_{1} } )$. However, since all such correlations can be achieved with unbiased marginals~\cite{Tsirelson:1980}, the possible values of correlators in the slice of unbiased marginals are described precisely by constraint~\eqref{eq:tlm-constraint}. The quantum face is characterised by $\expec{ A_{1} B_{1} } = -1$, which leads to a cubic inequality
\begin{equation*}
2 \expec{ A_{0} B_{0} } \cdot \expec{ A_{0} B_{1} } \cdot \expec{ A_{1} B_{0} } + \expec{ A_{0} B_{0} }^{2} + \expec{ A_{0} B_{1} }^{2} + \expec{ A_{1} B_{0} }^{2} \leq 1.
\end{equation*}
Any point which saturates this inequality and additionally satisfies
\begin{equation*}
\max \big\{ \abs{ \expec{ A_{0} B_{0} } }, \abs{ \expec{ A_{0} B_{1} } }, \abs{ \expec{ A_{1} B_{0} } } \big\} < 1
\end{equation*}
is a self-test~\cite{wang16a} and, hence, must be an extremal point of the quantum set (see Appendix~\ref{app:self-testing} for a proof). It is easy to verify that there is an infinite number of such points and, therefore, the quantum face corresponding to the Bell function $\vecb_{4}$ must have an infinite number of extremal points.

The strict inclusions $\cF_{\cL} \subsetneq \cF_{\cQ} \subsetneq \cF_{\NS}$, although intuitively clear, are neatly presented in a particular slice. In Fig.~\ref{fig:extremal-not-exposed} we present a 2-dimensional slice containing the PR box $\PPR$ and two local behaviours
\begin{align*}
    \PL{2} =
    \begin{array}{c|c|c}
        & 0 &  0 \\ \hline
        0  &  \nicefrac{1}{3} & \nicefrac{1}{3}  \\ \hline
        0 &  \nicefrac{1}{3} & -1
    \end{array}
    \quad \textnormal{and} \quad
    \PL{3} =
    \begin{array}{c|c|c}
        & 0 &  0 \\ \hline
        0  &  1 & 1  \\ \hline
        0 &  1 & 1
    \end{array} \; .
\end{align*}
This slice is singled out by unbiased marginals and three of the correlators being equal $\expec{ A_{0} B_{0} } = \expec{ A_{0} B_{1} } = \expec{ A_{1} B_{0} }$. Points in this slice are conveniently parametrised by $\expec{A_{1} B_{1}}$ and
\begin{equation*}
    \alpha = \expec{ A_{0} B_{0} } = \expec{ A_{0} B_{1} } = \expec{ A_{1} B_{0} }.
\end{equation*}
\begin{figure}[h]
    \includegraphics[width=9cm]{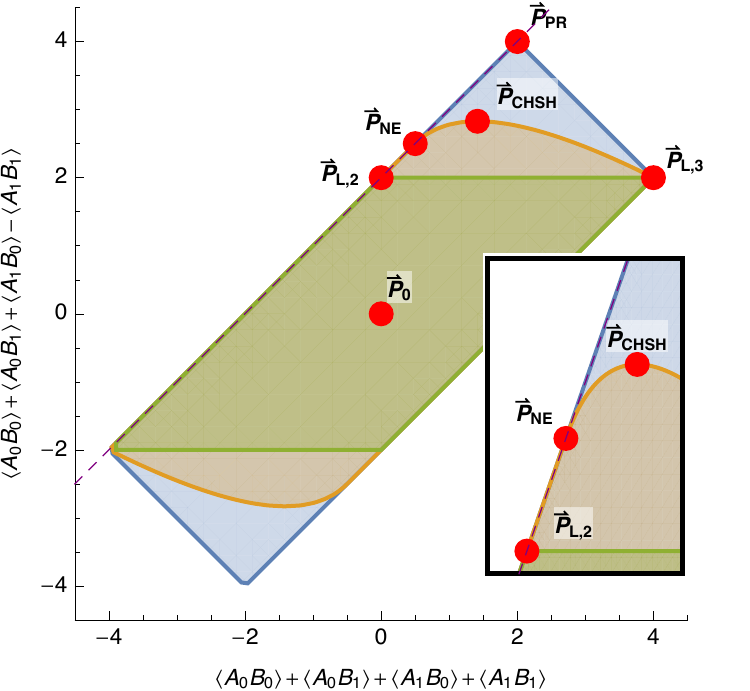}
    \caption{A highly symmetric 2-dimensional slice of the quantum set. Thanks to the analytic characterisation of the quantum set given in Eq.~\eqref{eq:curved-boundary}, we can rigorously show that $\PNE$ is a non-exposed point of the quantum set.}
    \label{fig:extremal-not-exposed}
\end{figure}
The constraint given in Eq.~\eqref{eq:tlm-constraint} implies the following
tight upper bound on $(- \expec{A_{1} B_{1}})$ as a function of $\alpha$:
\begin{align}
\label{eq:curved-boundary}
    - \expec{A_1 B_1} \leq
    \begin{cases}
        1 & -1\leq \alpha< \nicefrac{1}{2} \\
        3 \alpha - 4 \alpha^{3} & \nicefrac{1}{2} \leq \alpha \leq 1
    \end{cases} \, .
\end{align}
Let us consider the point at the boundary which corresponds to $\alpha = \nicefrac{1}{2}$, i.e.
\begin{align}
    \PNE =
    \begin{array}{c|c|c}
        & 0 &  0 \\ \hline
        0  &  \nicefrac{1}{2} & \nicefrac{1}{2}  \\ \hline
        0 &  \nicefrac{1}{2} & -1
    \end{array} \; .
\end{align}
This is precisely where the flat and curved parts of the boundary meet and it is easy to check that the gradients on both sides are equal, which implies that the point $\PNE$ is not exposed in this slice. Since an exposed point must remain exposed in every slice, we conclude that $\PNE$ is not exposed in the entire set. However, we know that $\PNE$ is an extremal point, because it is a self-test~\cite{wang16a}. We conjecture that any probability point which (i) saturates the constraint~\eqref{eq:tlm-constraint} and (ii) has precisely one correlator of unit modulus is of this type, i.e.~extremal but not exposed.

It turns out that similar geometric features are exhibited by the so-called Hardy point, i.e.~the (unique) point that maximally violates the Hardy paradox~\cite{Hardy1992}. The Hardy point is a self-test~\cite{Rabelo:2012} and therefore extremal, but it is not exposed (see Appendix~\ref{app:hardy-not-exposed} for a proof). This explains why previous attempts to find a Bell function that fully captures the nonlocal nature of the Hardy paradox failed. The authors of Ref.~\cite{Mancinska2014} proposed a sequence of Bell functions, whose maximiser approaches the Hardy point, but when one tries to take the limit, the coefficients of the functions diverge. This is precisely the behaviour one would expect when dealing with an extremal but not exposed point. A family of quantum faces that the Hardy point lies on is discussed in Appendix~\ref{app:positivity}.
\subsubsection{A quantum face with $\cF_{\cL} = \cF_{\cQ} \subsetneq \cF_{\NS}$}
\label{sec:b5}
Let us finish the discussion of the \scen{2}{2} scenario with an example of a quantum face which completely coincides with its local counterpart, but is strictly contained within the no-signalling face. Consider the Bell function
\begin{equation}
\label{eq:b5}
    \vecb_{5} :=
    \begin{array}{c|c|c}
        & 0 &  0 \\ \hline
        0  &  1 &  1 \\ \hline
        0 &  0 & 0
    \end{array} \; , \;
    \vecb_{5} \cdot \vecp \leq
    \begin{cases}
        2 & \cL\\
        \bm{2} & \cQ\\
        2 & \NS
    \end{cases}.
\end{equation}
We start by determining the quantum face $\cF_{\cQ}$. The maximal quantum value is achieved iff $\expec{ A_{0} B_{0} } = \expec{ A_{0} B_{1} } = 1$, which by constraint~\eqref{eq:tlm-constraint} implies that $\expec{A_{1} B_{0}} = \expec{A_{1} B_{1}}$. It is, however, straightforward to verify that such correlations are local (they cannot violate any variant of the CHSH inequality) and thus $\cF_{\cL} = \cF_{\cQ}$. On the other hand, the PR box also saturates the bound given in Eq.~\eqref{eq:b5}, thus showing that $\cF_{\cL} = \cF_{\cQ} \subsetneq \cF_{\NS}$.
\section{Nonlocal faces of positive dimension}
\label{sec:nonlocal-faces}
Our numerical studies of the quantum set in the \scen{2}{2} scenario suggest that every Bell function for which $\beta_{\cQ} > \beta_{\cL}$ has a unique maximiser in the quantum set. While we conjecture that this is indeed true in the \scen{2}{2} scenario, it is easy to see that it does not hold in general, e.g.~if we take the CHSH inequality and ``embed'' it in a Bell scenario with more inputs, then the maximal violation does not carry any information about the statistics corresponding to the additional inputs. A more natural family of such functions was proposed by Slofstra~\cite{Slofstra:2011}, but these require a large number of measurement settings on each side. On the other hand, a simple example was recently found in the tripartite \scen[3]{2}{2} scenario by Ramanathan and Mironowicz~\cite{ramanathan17a}. In this section we give an example in the bipartite scenario \scen{3}{2} and two additional examples in the \scen[3]{2}{2} scenario. What is particularly appealing about the tripartite examples is that we were able to fully determine the corresponding quantum faces.
\subsection{The \scen{3}{2} scenario}
Consider the correlation part of the $I_{3322}$ Bell function~\cite{Froissart1981, Collins2004}
\begin{equation}
\label{eq:i3322star}
\begin{split}
\vecb_{6} \cdot \vecp := &\expec{A_{0} B_{0}} + \expec{A_{0} B_{1}} + \expec{A_{0} B_{2}} + \expec{A_{1} B_{0}}\\
&+ \expec{A_{1} B_{1}} - \expec{A_{1} B_{2}} + \expec{A_{2} B_{0}} - \expec{A_{2} B_{1}}.
\end{split}
\end{equation}
The local and no-signalling values of this inequality have been found by enumerating the vertices of the respective polytopes, whereas the quantum value has been found using a semidefinite program~\cite{wehner06a} (see Appendix~\ref{app:b7-quantum-value} for details)
\begin{align}
    \vecb_{6} \cdot \vecp \leq
    \begin{cases}
        4 & \cL\\
        \bm{5} & \cQ\\
        8 & \NS
    \end{cases}.
\end{align}
Below we present a one-parameter family of quantum realisations which saturate the quantum bound of this Bell function. The shared state is $\ket{\Psi_{-}} = \frac{1}{\sqrt{2}} ( \ket{01} - \ket{10} )$ and the observables are
\begin{align*}
A_{0} &= \frac{1}{2} \Big( 2 \cos \frac{\pi}{6} \, \sigma_{x} + \cos \alpha \, \sigma_{y} + \sin \alpha \, \sigma_{z} \Big),\\
A_{1} &= \frac{1}{2} \Big( 2 \cos \frac{\pi}{6} \, \sigma_{x} - \cos \alpha \, \sigma_{y} - \sin \alpha \, \sigma_{z} \Big),\\
A_{2} &= \sigma_{y},\\
B_{0} &= - \cos \frac{\pi}{6} \, \sigma_{x} - \sin \frac{\pi}{6} \, \sigma_{y},\\
B_{1} &= - \cos \frac{\pi}{6} \, \sigma_{x} + \sin \frac{\pi}{6} \, \sigma_{y},\\
B_{2} &= - \cos \alpha \, \sigma_{y} - \sin \alpha \, \sigma_{z},
\end{align*}
where $\alpha \in [0, 2 \pi]$ is a free parameter. It is clear that all the marginals vanish: $\expec{A_{x}} = \expec{B_{y}} = 0$, while the correlators are given by
\begin{align*}
\expec{A_{0} B_{0}} &= \expec{A_{1} B_{1}} = \frac{3 + \cos \alpha}{4},\\
\expec{A_{0} B_{1}} &= \expec{A_{1} B_{0}} = \frac{3 - \cos \alpha}{4},\\
\expec{A_{0} B_{2}} &= \expec{A_{2} B_{0}} = \frac{1}{2},\\
\expec{A_{1} B_{2}} &= \expec{A_{2} B_{1}} = - \frac{1}{2},\\
\expec{A_{2} B_{2}} &= \cos \alpha.
\end{align*}
This family of probability points is simply a line whose extremal points correspond to $\alpha = 0$ and $\alpha = \pi$. Moreover, it is easy to check that the two extremal points are related by swapping $A_{0}$ with $A_{1}$ and flipping the sign of $B_{2}$ (since the Bell function is symmetric we could alternatively swap $B_{0}$ and $B_{1}$ and flip the sign of $A_{2}$).

We do not know whether the quantum face corresponding to $\vecb_{6}$ is strictly larger than the line, but the existence of such a 1-dimensional region already has interesting implications for self-testing. More concretely, it means that saturating the quantum bound $\beta_{\cQ} = 5$ does not imply the usual self-testing statement, simply because the maximal value can be achieved by multiple inequivalent arrangements of observables.\footnote{In most cases of self-testing it is sufficient to allow for extra degrees of freedom and local isometries (see Appendix~\ref{app:self-testing} for details), but sometimes one must also consider the transposition (complex conjugation) equivalence~\cite{mckague11a, andersson17a}. The transposition equivalence is automatically taken care of if one looks at commutation relations between the local observables~\cite{kaniewski17a}, which immediately implies that the quantum realisations presented in the main text are inequivalent even if we allow for this additional equivalence.} It is, however, still possible that saturating the quantum bound certifies the maximally entangled state of two qubits. If so, this would be an example where the maximal violation certifies the state but not the measurements.
\subsection{The tripartite scenarios}
\label{sec:tripartite-scenarios}
Finally, we discuss two tripartite examples which demonstrate that the geometry of the quantum set becomes even more complex in the multipartite scenarios.

The quantum set for multiple parties has been extensively studied, but mainly in the context of Bell inequalities. From the study of multipartite self-testing we know that certain Bell functions have unique maximisers, e.g.~the Bell function proposed by Mermin~\cite{colbeck06a} (and its generalisation due to Mermin, Ardehali, Belinskii and Klyshko~\cite{kaniewski17a}), but also the Bell functions constructed to self-test the $W$ state~\cite{Pal:2014}. However, we conjecture that this behaviour is not generic and present two Bell functions in the tripartite scenario which give rise to more complex quantum faces. In contrast to the bipartite scenario discussed before, in these cases we can explicitly map out the entire quantum face. The multiple inequivalent ways of saturating the quantum bound immediately imply specific limitations on the self-testing statements we can hope for.

In both examples Alice, Bob and Charlie perform binary measurements. In the first example Alice and Bob have two measurements, whereas Charlie only has one. Consider the Bell function
\begin{equation*}
\vecb_{7} \cdot \vecp := \expec{A_{0} B_{0} C_{0}} + \expec{A_{0} B_{1} C_{0}} + \expec{A_{1} B_{0} C_{0}} - \expec{A_{1} B_{1} C_{0}},
\end{equation*}
which first appeared as Eq.~(15) in Ref.~\cite{Werner:2001}. Note that this is nothing else than the CHSH function between Alice and Bob ``modulated'' by the outcome of Charlie, which immediately implies that $\beta_{\cL} = 2, \beta_{\cQ} = 2 \sqrt{2}$ and $\beta_{\NS} = 4$. The quantum bound is saturated when Alice and Bob perform the optimal CHSH strategy while Charlie deterministically outputs $0$, which leads to\footnote{It is well known that in a tripartite scenario with two outcomes per site the probability point is uniquely determined by the local marginals and the two- and three-body expectation values, see e.g.~Ref.~\cite{Pironio:2011}.}
\begin{gather}
\expec{ A_{x} } = \expec{ B_{y} } = 0, \; \expec{ C_{0} } = 1,\nonumber\\
\label{eq:tripartite-extremal-point-1}
\expec{A_{x} B_{y}} = \expec{A_{x} B_{y} C_{0}} = (-1)^{xy} / \sqrt{2},\\
\expec{A_{x} C_{0}} = \expec{B_{y} C_{0}} = 0.\nonumber
\end{gather}
Alternatively, the quantum bound may be saturated if Alice and Bob achieve the CHSH value of $- 2 \sqrt{2}$, while Charlie deterministically outputs $1$, which leads to
\begin{gather}
\expec{ A_{x} } = \expec{ B_{y} } = 0, \; \expec{ C_{0} } = -1,\nonumber\\
\label{eq:tripartite-extremal-point-2}
- \expec{A_{x} B_{y}} = \expec{A_{x} B_{y} C_{0}} = (-1)^{xy} / \sqrt{2},\\
\expec{A_{x} C_{0}} = \expec{B_{y} C_{0}} = 0.\nonumber
\end{gather}
Since Charlie always performs the same measurement, for the purpose of computing the resulting statistics we can assume that he is only classically correlated with Alice and Bob. Conditioned on a particular output of Charlie the statistics on Alice and Bob are unique, as they must achieve the CHSH value of $+2 \sqrt{2}$ or $-2 \sqrt{2}$. This implies that we always end up with a convex combinations of statistics given in Eqs.~\eqref{eq:tripartite-extremal-point-1} and \eqref{eq:tripartite-extremal-point-2}, i.e.~that the resulting quantum face is simply a line.

Any point on this line can be realised using a three-qubit state shared among Alice, Bob and Charlie. In fact, it suffices to look at a single arrangement of qubit observables
\begin{equation}
\label{eq:b7-observables}
\begin{array}{ll}
A_{0} = \sigma_{x}, & A_{1} = \sigma_{z},\\
B_{0} = (\sigma_{x} + \sigma_{z})/\sqrt{2}, & B_{1} = (\sigma_{x} - \sigma_{z})/\sqrt{2},\\
C_{0} = \sigma_{z}. &
\end{array}
\end{equation}
The largest eigenvalue of the resulting Bell operator equals $\lambda = 2 \sqrt{2}$ and the corresponding eigenspace is 2-dimensional and spanned by vectors $\{ \ket{ \Phi_{+} }_{AB} \ket{0}_{C}, \ket{ \Psi_{-} }_{AB} \ket{1}_{C} \}$, where $\ket{\Phi_{+}} = (\ket{00} + \ket{11})/\sqrt{2}$ and $\ket{\Psi_{-}} = (\ket{01} - \ket{10})/\sqrt{2}$. Therefore, the quantum bound is saturated by any state of the form
\begin{equation*}
\ket{\eta}_{ABC} := \cos \theta \, \ket{ \Phi_{+} }_{AB} \ket{0}_{C} + \sin \theta \, \ket{ \Psi_{-} }_{AB} \ket{1}_{C}
\end{equation*}
for $\theta \in [0, \pi/2]$. In fact, since Charlie always measures in the computational basis, the same statistics could be obtained from the mixed state\footnote{The same mixed state was recently used by Krisnanda et al.~to demonstrate that quantum systems can become entangled even if they interact only through a mediator which remains classical (diagonal in a fixed basis) at all times~\cite{krisnanda17a}.}
\begin{align*}
\rho_{ABC} = &\cos^{2} \theta \, \ketbraq{\Phi_{+}}_{AB} \otimes \ketbraq{0}_{C}\\
&+ \sin^{2} \theta \, \ketbraq{ \Psi_{-} }_{AB} \otimes \ketbraq{1}_{C},
\end{align*}
which clearly results in a convex combination of the two extremal points.

It is instructive to consider what kind of self-testing statements we can hope for in this case. Grouping Bob and Charlie brings us back to the CHSH scenario (in the sense that Bob and Charlie together have only two distinct measurement settings), so there must be a maximally entangled two-qubit state in the bipartition Alice vs.~Bob and Charlie, but we do not know exactly how the entanglement is split between Bob and Charlie. At the extremal points given by Eqs.~\eqref{eq:tripartite-extremal-point-1} and \eqref{eq:tripartite-extremal-point-2} the reduced statistics on Alice and Bob saturate the quantum bound of some CHSH function, which ensures that the relevant entanglement is shared between Alice and Bob only. In the interior of the line, however, we cannot make such precise statements. In particular, while all the interior points can be realised using genuinely tripartite entanglement, such entanglement can never be certified in this setup, simply because the entire line can be written as a convex combination of the extremal points (which can be achieved using bipartite entanglement between Alice and Bob).

In the second example there are two measurements on each site, i.e.~we are in the \scen[3]{2}{2} scenario. Consider the Bell function
\begin{equation*}
\vecb_{8} \cdot \vecp := \expec{A_{0} B_{0} C_{0}} + \expec{A_{0} B_{1} C_{1}} + \expec{A_{1} B_{0} C_{0}} - \expec{A_{1} B_{1} C_{1}}
\end{equation*}
for which $\beta_{\cL} = 2, \beta_{\cQ} = 2 \sqrt{2}$ and $\beta_{\NS} = 4$. This Bell function was found by Werner and Wolf while characterising the facets of the correlation polytope in the \scen[3]{2}{2} scenario~\cite{Werner:2001}, but as shown in Ref.~\cite{Sliwa:2003} it is also a facet Bell inequality of the full local polytope. We show in Appendix~\ref{app:tripartite-faces} that the corresponding quantum face is the convex hull of 8 discrete points and a one-parameter family of quantum points arising from the tripartite Greenberger-Horne-Zeilinger (GHZ) state~\cite{greenberger89a}.

The 8 points, denoted by $\{ \vecp_{j} \}_{j = 1}^{8}$, are achieved when Alice saturates some variant of the CHSH function with either Bob or Charlie, while the remaining party adopts a deterministic strategy. The one-parameter family corresponds to Bob and Charlie nontrivially ``sharing'' the maximal CHSH violation.

The first two points correspond to Charlie always producing the same outcome regardless of his input. The resulting statistics are analogous to those in Eqs.~\eqref{eq:tripartite-extremal-point-1} and~\eqref{eq:tripartite-extremal-point-2}:
\begin{equation*}
\begin{array}{lll}
&& \expec{ A_{x} } = \expec{ B_{y} } = 0, \; \expec{ C_{z} } = 1,\\
\vecp_{1}: &&\expec{A_{x} B_{y}} = \expec{A_{x} B_{y} C_{z}} = (-1)^{xy} / \sqrt{2},\\
&& \expec{A_{x} C_{z}} = \expec{B_{y} C_{z}} = 0
\end{array}
\end{equation*}
and
\begin{equation*}
\begin{array}{lll}
&& \expec{ A_{x} } = \expec{ B_{y} } = 0, \; \expec{ C_{z} } = - 1,\\
\vecp_{2}: &&- \expec{A_{x} B_{y}} = \expec{A_{x} B_{y} C_{z}} = (-1)^{xy} / \sqrt{2},\\
&& \expec{A_{x} C_{z}} = \expec{B_{y} C_{z}} = 0.
\end{array}
\end{equation*}
Points $\vecp_{3}$ and $\vecp_{4}$ arise if Charlie's outcome depends on his input, which implies that Alice and Bob must saturate another variant of the CHSH inequality. The resulting statistics are:
\begin{equation*}
\begin{array}{lll}
&& \expec{ A_{x} } = \expec{ B_{y} } = 0,\\
&&\expec{ C_{0} } = 1, \; \expec{ C_{1} } = -1,\\
\vecp_{3}: &&\expec{A_{x} B_{y}} = (-1)^{ (x + 1) y } / \sqrt{2},\\
&&\expec{A_{x} B_{y} C_{z}} = (-1)^{ (x + 1) y + z } / \sqrt{2},\\
&& \expec{A_{x} C_{z}} = \expec{B_{y} C_{z}} = 0
\end{array}
\end{equation*}
and
\begin{equation*}
\begin{array}{lll}
&& \expec{ A_{x} } = \expec{ B_{y} } = 0,\\
&&\expec{ C_{0} } = -1, \; \expec{ C_{1} } = 1,\\
\vecp_{4}: &&\expec{A_{x} B_{y}} = (-1)^{ (x + 1) y + 1 } / \sqrt{2},\\
&&\expec{A_{x} B_{y} C_{z}} = (-1)^{ (x + 1) y + z } / \sqrt{2},\\
&& \expec{A_{x} C_{z}} = \expec{B_{y} C_{z}} = 0.
\end{array}
\end{equation*}
Points $\{ P_{j} \}_{j = 5}^{8}$ are constructed from $\{ P_{j} \}_{j = 1}^{4}$ by exchanging the roles of Bob and Charlie.

The one-parameter family of facial points has vanishing one- and two-body expectation values
\begin{equation*}
\expec{ A_{x} }= \expec{ B_{y} } = \expec{ C_{z} } = \expec{ A_{x} B_{y} } = \expec{ A_{x} C_{z} } = \expec{ B_{y} C_{z} } = 0,
\end{equation*}
while the three-body correlations are given by:
\begin{equation}
\label{eq:one-parameter-family}
\begin{split}
\expec{ A_{0} B_{0} C_{0} } &= \expec{ A_{0} B_{1} C_{1} } = \frac{1}{\sqrt{2}},\\
\expec{ A_{0} B_{0} C_{1} } &= \expec{ A_{0} B_{1} C_{0} } = \cos \alpha,\\
\expec{ A_{1} B_{0} C_{0} } &= - \expec{ A_{1} B_{1} C_{1} } = \frac{1}{\sqrt{2}},\\
\expec{ A_{1} B_{0} C_{1} } &= - \expec{ A_{1} B_{1} C_{0} } = \sin \alpha
\end{split}
\end{equation}
for $\alpha \in [0, 2 \pi]$.

This example is important because we can explicitly compute the corresponding quantum face and we see that it is a highly non-trivial object. We conjecture that in multipartite scenarios such high-dimensional and non-polytopic quantum faces are a common phenomenon.
\section{Conclusions and open questions}
\label{sec:conclusions}
In this work we have studied the geometry of the quantum set. In particular, we have identified several flat regions lying on the boundary of the quantum set and we have found extremal points which are not exposed. We have also introduced a classification of Bell functions in terms of the facial structure they give rise to and provided an explicit example for each existing class. Finally, we have presented a simple example of a bipartite Bell function whose quantum and classical values differ for which the quantum maximiser is not unique.

Despite the progress we have made on understanding the geometry of the quantum set in the \scen{2}{2} scenario, several questions remain open. For instance, having found a 1-dimensional flat boundary region containing the CHSH point, one could ask whether it is possible to find a higher-dimensional region of that kind or, more generally, what is the highest dimension of a flat region containing the CHSH point. Let us also put forward the following conjecture about the uniqueness of the maximiser: from our numerics it seems that all Bell functions in the \scen{2}{2} scenario have at most 1 extremal nonlocal maximiser. Can one find an analytical proof of this statement?

Another interesting task would be to study the extremal points of the quantum set in the \scen{2}{2} scenario. We know that all of them can be achieved by projective measurements on a two-qubit state, but we know that the latter is a strict superset of the former. This is particularly interesting from the self-testing point of view: we know that if the marginals are uniform, then all the extremal nonlocal points are self-tests. Is this also true for correlation points with arbitrary marginals? In other words, are all extremal nonlocal points of the quantum set in the \scen{2}{2} scenario self-tests?

Another natural question arising from our results concerns the ``generic'' geometry of the quantum set. In this work we provide several examples of unexpected geometric features of the quantum set, but in order to see them one has to go beyond the standard, well-studied Bell functions. Therefore, the question is whether such features are indeed ``unusual'' or our intuition has simply been skewed by looking only at ``regular'' Bell functions for which these behaviours do not appear. We suspect that such features are indeed unusual, but we currently have no rigorous evidence to support this claim.
\begin{acknowledgements}
We would like to thank Eliahu Cohen for showing us the elegant form of the Tsirelson-Landau-Masanes criterion, Laura Man\v{c}inska for bringing to our attention Ref.~\cite{Ramanathan:2016} and Denis Rosset and Antonios Varvitsiotis for useful discussions.

This research is supported by the Singapore Ministry of Education Academic Research Fund Tier 3 (grant no.~MOE2012-T3-1-009), the National Research Fund and the Ministry of Education, Singapore, under the Research Centres of Excellence programme, the John Templeton Foundation project ``Many-box locality as a physical principle'' (grant no.~60607), the European Union's Horizon 2020 research and innovation programme under the Marie Sk{\l}odowska-Curie Action ROSETTA (grant no.~749316), the European Research Council (grant no.~337603), the Danish Council for Independent Research (Sapere Aude), the VILLUM FONDEN via the QMATH Centre of Excellence (grant no.~10059), the National Research, Development and Innovation Office NKFIH (grant nos.~K111734 and KH125096), the Ministry of Education of Taiwan, R.O.C., through ``Aiming for the Top University Project'' granted to the National Cheng Kung University (NCKU), the Ministry of Science and Technology of Taiwan, R.O.C.~(grant no.~104-2112-M-006-021-MY3), and in part by Perimeter Institute for Theoretical Physics. Research at Perimeter Institute is supported by the Government of Canada through the Department of Innovation, Science and Economic Development Canada and by the Province of Ontario through the Ministry of Research, Innovation and Science. 
\end{acknowledgements}
\setlength{\bibsep}{2pt plus 1pt minus 2pt}
\bibliographystyle{apsrev4-1}
\nocite{apsrev41Control}
\bibliography{geometry.bib}

\begin{thebibliography}{78}%
\makeatletter
\providecommand \@ifxundefined [1]{%
 \@ifx{#1\undefined}
}%
\providecommand \@ifnum [1]{%
 \ifnum #1\expandafter \@firstoftwo
 \else \expandafter \@secondoftwo
 \fi
}%
\providecommand \@ifx [1]{%
 \ifx #1\expandafter \@firstoftwo
 \else \expandafter \@secondoftwo
 \fi
}%
\providecommand \natexlab [1]{#1}%
\providecommand \enquote  [1]{``#1''}%
\providecommand \bibnamefont  [1]{#1}%
\providecommand \bibfnamefont [1]{#1}%
\providecommand \citenamefont [1]{#1}%
\providecommand \href@noop [0]{\@secondoftwo}%
\providecommand \href [0]{\begingroup \@sanitize@url \@href}%
\providecommand \@href[1]{\@@startlink{#1}\@@href}%
\providecommand \@@href[1]{\endgroup#1\@@endlink}%
\providecommand \@sanitize@url [0]{\catcode `\\12\catcode `\$12\catcode
  `\&12\catcode `\#12\catcode `\^12\catcode `\_12\catcode `\%12\relax}%
\providecommand \@@startlink[1]{}%
\providecommand \@@endlink[0]{}%
\providecommand \url  [0]{\begingroup\@sanitize@url \@url }%
\providecommand \@url [1]{\endgroup\@href {#1}{\urlprefix }}%
\providecommand \urlprefix  [0]{URL }%
\providecommand \Eprint [0]{\href }%
\providecommand \doibase [0]{http://dx.doi.org/}%
\providecommand \selectlanguage [0]{\@gobble}%
\providecommand \bibinfo  [0]{\@secondoftwo}%
\providecommand \bibfield  [0]{\@secondoftwo}%
\providecommand \translation [1]{[#1]}%
\providecommand \BibitemOpen [0]{}%
\providecommand \bibitemStop [0]{}%
\providecommand \bibitemNoStop [0]{.\EOS\space}%
\providecommand \EOS [0]{\spacefactor3000\relax}%
\providecommand \BibitemShut  [1]{\csname bibitem#1\endcsname}%
\let\auto@bib@innerbib\@empty
\bibitem [{\citenamefont {Bell}(1964)}]{Bell:1964}%
  \BibitemOpen
  \bibfield  {author} {\bibinfo {author} {\bibfnamefont {J.~S.}\ \bibnamefont
  {Bell}},\ }\bibfield  {title} {\enquote {\bibinfo {title} {On the
  {E}instein-{P}odolsky-{R}osen paradox},}\ }\href@noop {} {\bibfield
  {journal} {\bibinfo  {journal} {Physics}\ }\textbf {\bibinfo {volume} {1}},\
  \bibinfo {pages} {195} (\bibinfo {year} {1964})}\BibitemShut {NoStop}%
\bibitem [{\citenamefont {Hensen}\ \emph {et~al.}(2015)\citenamefont {Hensen},
  \citenamefont {Bernien}, \citenamefont {Dr\'eau}, \citenamefont {Reiserer},
  \citenamefont {Kalb}, \citenamefont {Blok}, \citenamefont {Ruitenberg},
  \citenamefont {Vermeulen}, \citenamefont {Schouten}, \citenamefont
  {Abell\'an}, \citenamefont {Amaya}, \citenamefont {Pruneri}, \citenamefont
  {Mitchell}, \citenamefont {Markham}, \citenamefont {Twitchen}, \citenamefont
  {Elkouss}, \citenamefont {Wehner}, \citenamefont {Taminiau},\ and\
  \citenamefont {Hanson}}]{Hensen:2015}%
  \BibitemOpen
  \bibfield  {author} {\bibinfo {author} {\bibfnamefont {B.}~\bibnamefont
  {Hensen}}, \bibinfo {author} {\bibfnamefont {H.}~\bibnamefont {Bernien}},
  \bibinfo {author} {\bibfnamefont {A.~E.}\ \bibnamefont {Dr\'eau}}, \bibinfo
  {author} {\bibfnamefont {A.}~\bibnamefont {Reiserer}}, \bibinfo {author}
  {\bibfnamefont {N.}~\bibnamefont {Kalb}}, \bibinfo {author} {\bibfnamefont
  {M.~S.}\ \bibnamefont {Blok}}, \bibinfo {author} {\bibfnamefont
  {J.}~\bibnamefont {Ruitenberg}}, \bibinfo {author} {\bibfnamefont {R.~F.~L.}\
  \bibnamefont {Vermeulen}}, \bibinfo {author} {\bibfnamefont {R.~N.}\
  \bibnamefont {Schouten}}, \bibinfo {author} {\bibfnamefont {C.}~\bibnamefont
  {Abell\'an}}, \bibinfo {author} {\bibfnamefont {W.}~\bibnamefont {Amaya}},
  \bibinfo {author} {\bibfnamefont {V.}~\bibnamefont {Pruneri}}, \bibinfo
  {author} {\bibfnamefont {M.~W.}\ \bibnamefont {Mitchell}}, \bibinfo {author}
  {\bibfnamefont {M.}~\bibnamefont {Markham}}, \bibinfo {author} {\bibfnamefont
  {D.~J.}\ \bibnamefont {Twitchen}}, \bibinfo {author} {\bibfnamefont
  {D.}~\bibnamefont {Elkouss}}, \bibinfo {author} {\bibfnamefont
  {S.}~\bibnamefont {Wehner}}, \bibinfo {author} {\bibfnamefont {T.~H.}\
  \bibnamefont {Taminiau}}, \ and\ \bibinfo {author} {\bibfnamefont
  {R.}~\bibnamefont {Hanson}},\ }\bibfield  {title} {\enquote {\bibinfo {title}
  {{Loophole-free Bell inequality violation using electron spins separated by
  1.3 kilometres}},}\ }\href {\doibase 10.1038/nature15759} {\bibfield
  {journal} {\bibinfo  {journal} {Nature}\ }\textbf {\bibinfo {volume} {526}},\
  \bibinfo {pages} {682} (\bibinfo {year} {2015})}\BibitemShut {NoStop}%
\bibitem [{\citenamefont {Shalm}\ \emph {et~al.}(2015)\citenamefont {Shalm},
  \citenamefont {Meyer-Scott}, \citenamefont {Christensen}, \citenamefont
  {Bierhorst}, \citenamefont {Wayne}, \citenamefont {Stevens}, \citenamefont
  {Gerrits}, \citenamefont {Glancy}, \citenamefont {Hamel}, \citenamefont
  {Allman}, \citenamefont {Coakley}, \citenamefont {Dyer}, \citenamefont
  {Hodge}, \citenamefont {Lita}, \citenamefont {Verma}, \citenamefont
  {Lambrocco}, \citenamefont {Tortorici}, \citenamefont {Migdall},
  \citenamefont {Zhang}, \citenamefont {Kumor}, \citenamefont {Farr},
  \citenamefont {Marsili}, \citenamefont {Shaw}, \citenamefont {Stern},
  \citenamefont {Abell\'an}, \citenamefont {Amaya}, \citenamefont {Pruneri},
  \citenamefont {Jennewein}, \citenamefont {Mitchell}, \citenamefont {Kwiat},
  \citenamefont {Bienfang}, \citenamefont {Mirin}, \citenamefont {Knill},\ and\
  \citenamefont {Nam}}]{Shalm:2015}%
  \BibitemOpen
  \bibfield  {author} {\bibinfo {author} {\bibfnamefont {L.~K.}\ \bibnamefont
  {Shalm}}, \bibinfo {author} {\bibfnamefont {E.}~\bibnamefont {Meyer-Scott}},
  \bibinfo {author} {\bibfnamefont {B.~G.}\ \bibnamefont {Christensen}},
  \bibinfo {author} {\bibfnamefont {P.}~\bibnamefont {Bierhorst}}, \bibinfo
  {author} {\bibfnamefont {M.~A.}\ \bibnamefont {Wayne}}, \bibinfo {author}
  {\bibfnamefont {M.~J.}\ \bibnamefont {Stevens}}, \bibinfo {author}
  {\bibfnamefont {T.}~\bibnamefont {Gerrits}}, \bibinfo {author} {\bibfnamefont
  {S.}~\bibnamefont {Glancy}}, \bibinfo {author} {\bibfnamefont {D.~R.}\
  \bibnamefont {Hamel}}, \bibinfo {author} {\bibfnamefont {M.~S.}\ \bibnamefont
  {Allman}}, \bibinfo {author} {\bibfnamefont {K.~J.}\ \bibnamefont {Coakley}},
  \bibinfo {author} {\bibfnamefont {S.~D.}\ \bibnamefont {Dyer}}, \bibinfo
  {author} {\bibfnamefont {C.}~\bibnamefont {Hodge}}, \bibinfo {author}
  {\bibfnamefont {A.~E.}\ \bibnamefont {Lita}}, \bibinfo {author}
  {\bibfnamefont {V.~B.}\ \bibnamefont {Verma}}, \bibinfo {author}
  {\bibfnamefont {C.}~\bibnamefont {Lambrocco}}, \bibinfo {author}
  {\bibfnamefont {E.}~\bibnamefont {Tortorici}}, \bibinfo {author}
  {\bibfnamefont {A.~L.}\ \bibnamefont {Migdall}}, \bibinfo {author}
  {\bibfnamefont {Y.}~\bibnamefont {Zhang}}, \bibinfo {author} {\bibfnamefont
  {D.~R.}\ \bibnamefont {Kumor}}, \bibinfo {author} {\bibfnamefont {W.~H.}\
  \bibnamefont {Farr}}, \bibinfo {author} {\bibfnamefont {F.}~\bibnamefont
  {Marsili}}, \bibinfo {author} {\bibfnamefont {M.~D.}\ \bibnamefont {Shaw}},
  \bibinfo {author} {\bibfnamefont {J.~A.}\ \bibnamefont {Stern}}, \bibinfo
  {author} {\bibfnamefont {C.}~\bibnamefont {Abell\'an}}, \bibinfo {author}
  {\bibfnamefont {W.}~\bibnamefont {Amaya}}, \bibinfo {author} {\bibfnamefont
  {V.}~\bibnamefont {Pruneri}}, \bibinfo {author} {\bibfnamefont
  {T.}~\bibnamefont {Jennewein}}, \bibinfo {author} {\bibfnamefont {M.~W.}\
  \bibnamefont {Mitchell}}, \bibinfo {author} {\bibfnamefont {P.~G.}\
  \bibnamefont {Kwiat}}, \bibinfo {author} {\bibfnamefont {J.~C.}\ \bibnamefont
  {Bienfang}}, \bibinfo {author} {\bibfnamefont {R.~P.}\ \bibnamefont {Mirin}},
  \bibinfo {author} {\bibfnamefont {E.}~\bibnamefont {Knill}}, \ and\ \bibinfo
  {author} {\bibfnamefont {S.~W.}\ \bibnamefont {Nam}},\ }\bibfield  {title}
  {\enquote {\bibinfo {title} {Strong loophole-free test of local realism},}\
  }\href {\doibase 10.1103/PhysRevLett.115.250402} {\bibfield  {journal}
  {\bibinfo  {journal} {Phys. Rev. Lett.}\ }\textbf {\bibinfo {volume} {115}},\
  \bibinfo {pages} {250402} (\bibinfo {year} {2015})}\BibitemShut {NoStop}%
\bibitem [{\citenamefont {Giustina}\ \emph {et~al.}(2015)\citenamefont
  {Giustina}, \citenamefont {Versteegh}, \citenamefont {Wengerowsky},
  \citenamefont {Handsteiner}, \citenamefont {Hochrainer}, \citenamefont
  {Phelan}, \citenamefont {Steinlechner}, \citenamefont {Kofler}, \citenamefont
  {Larsson}, \citenamefont {Abell\'an}, \citenamefont {Amaya}, \citenamefont
  {Pruneri}, \citenamefont {Mitchell}, \citenamefont {Beyer}, \citenamefont
  {Gerrits}, \citenamefont {Lita}, \citenamefont {Shalm}, \citenamefont {Nam},
  \citenamefont {Scheidl}, \citenamefont {Ursin}, \citenamefont {Wittmann},\
  and\ \citenamefont {Zeilinger}}]{Giustina:2015}%
  \BibitemOpen
  \bibfield  {author} {\bibinfo {author} {\bibfnamefont {M.}~\bibnamefont
  {Giustina}}, \bibinfo {author} {\bibfnamefont {M.~A.~M.}\ \bibnamefont
  {Versteegh}}, \bibinfo {author} {\bibfnamefont {S.}~\bibnamefont
  {Wengerowsky}}, \bibinfo {author} {\bibfnamefont {J.}~\bibnamefont
  {Handsteiner}}, \bibinfo {author} {\bibfnamefont {A.}~\bibnamefont
  {Hochrainer}}, \bibinfo {author} {\bibfnamefont {K.}~\bibnamefont {Phelan}},
  \bibinfo {author} {\bibfnamefont {F.}~\bibnamefont {Steinlechner}}, \bibinfo
  {author} {\bibfnamefont {J.}~\bibnamefont {Kofler}}, \bibinfo {author}
  {\bibfnamefont {J.-A.}\ \bibnamefont {Larsson}}, \bibinfo {author}
  {\bibfnamefont {C.}~\bibnamefont {Abell\'an}}, \bibinfo {author}
  {\bibfnamefont {W.}~\bibnamefont {Amaya}}, \bibinfo {author} {\bibfnamefont
  {V.}~\bibnamefont {Pruneri}}, \bibinfo {author} {\bibfnamefont {M.~W.}\
  \bibnamefont {Mitchell}}, \bibinfo {author} {\bibfnamefont {J.}~\bibnamefont
  {Beyer}}, \bibinfo {author} {\bibfnamefont {T.}~\bibnamefont {Gerrits}},
  \bibinfo {author} {\bibfnamefont {A.~E.}\ \bibnamefont {Lita}}, \bibinfo
  {author} {\bibfnamefont {L.~K.}\ \bibnamefont {Shalm}}, \bibinfo {author}
  {\bibfnamefont {S.~W.}\ \bibnamefont {Nam}}, \bibinfo {author} {\bibfnamefont
  {T.}~\bibnamefont {Scheidl}}, \bibinfo {author} {\bibfnamefont
  {R.}~\bibnamefont {Ursin}}, \bibinfo {author} {\bibfnamefont
  {B.}~\bibnamefont {Wittmann}}, \ and\ \bibinfo {author} {\bibfnamefont
  {A.}~\bibnamefont {Zeilinger}},\ }\bibfield  {title} {\enquote {\bibinfo
  {title} {{Significant-loophole-free test of Bell's theorem with entangled
  photons}},}\ }\href {\doibase 10.1103/PhysRevLett.115.250401} {\bibfield
  {journal} {\bibinfo  {journal} {Phys. Rev. Lett.}\ }\textbf {\bibinfo
  {volume} {115}},\ \bibinfo {pages} {250401} (\bibinfo {year}
  {2015})}\BibitemShut {NoStop}%
\bibitem [{\citenamefont {Rosenfeld}\ \emph {et~al.}(2017)\citenamefont
  {Rosenfeld}, \citenamefont {Burchardt}, \citenamefont {Garthoff},
  \citenamefont {Redeker}, \citenamefont {Ortegel}, \citenamefont {Rau},\ and\
  \citenamefont {Weinfurter}}]{Rosenfeld:2017}%
  \BibitemOpen
  \bibfield  {author} {\bibinfo {author} {\bibfnamefont {W.}~\bibnamefont
  {Rosenfeld}}, \bibinfo {author} {\bibfnamefont {D.}~\bibnamefont
  {Burchardt}}, \bibinfo {author} {\bibfnamefont {R.}~\bibnamefont {Garthoff}},
  \bibinfo {author} {\bibfnamefont {K.}~\bibnamefont {Redeker}}, \bibinfo
  {author} {\bibfnamefont {N.}~\bibnamefont {Ortegel}}, \bibinfo {author}
  {\bibfnamefont {M.}~\bibnamefont {Rau}}, \ and\ \bibinfo {author}
  {\bibfnamefont {H.}~\bibnamefont {Weinfurter}},\ }\bibfield  {title}
  {\enquote {\bibinfo {title} {{Event-ready Bell test using entangled atoms
  simultaneously closing detection and locality loopholes}},}\ }\href {\doibase
  10.1103/PhysRevLett.119.010402} {\bibfield  {journal} {\bibinfo  {journal}
  {Phys. Rev. Lett.}\ }\textbf {\bibinfo {volume} {119}},\ \bibinfo {pages}
  {010402} (\bibinfo {year} {2017})}\BibitemShut {NoStop}%
\bibitem [{\citenamefont {Brunner}\ \emph {et~al.}(2014)\citenamefont
  {Brunner}, \citenamefont {Cavalcanti}, \citenamefont {Pironio}, \citenamefont
  {Scarani},\ and\ \citenamefont {Wehner}}]{Brunner:2014}%
  \BibitemOpen
  \bibfield  {author} {\bibinfo {author} {\bibfnamefont {N.}~\bibnamefont
  {Brunner}}, \bibinfo {author} {\bibfnamefont {D.}~\bibnamefont {Cavalcanti}},
  \bibinfo {author} {\bibfnamefont {S.}~\bibnamefont {Pironio}}, \bibinfo
  {author} {\bibfnamefont {V.}~\bibnamefont {Scarani}}, \ and\ \bibinfo
  {author} {\bibfnamefont {S.}~\bibnamefont {Wehner}},\ }\bibfield  {title}
  {\enquote {\bibinfo {title} {{Bell nonlocality}},}\ }\href {\doibase
  10.1103/RevModPhys.86.419} {\bibfield  {journal} {\bibinfo  {journal} {Rev.
  Mod. Phys.}\ }\textbf {\bibinfo {volume} {86}},\ \bibinfo {pages} {419}
  (\bibinfo {year} {2014})}\BibitemShut {NoStop}%
\bibitem [{\citenamefont {Clauser}\ \emph {et~al.}(1969)\citenamefont
  {Clauser}, \citenamefont {Horne}, \citenamefont {Shimony},\ and\
  \citenamefont {Holt}}]{CHSH}%
  \BibitemOpen
  \bibfield  {author} {\bibinfo {author} {\bibfnamefont {J.~F.}\ \bibnamefont
  {Clauser}}, \bibinfo {author} {\bibfnamefont {M.~A.}\ \bibnamefont {Horne}},
  \bibinfo {author} {\bibfnamefont {A.}~\bibnamefont {Shimony}}, \ and\
  \bibinfo {author} {\bibfnamefont {R.~A.}\ \bibnamefont {Holt}},\ }\bibfield
  {title} {\enquote {\bibinfo {title} {{Proposed experiment to test local
  hidden-variable theories}},}\ }\href
  {http://link.aps.org/doi/10.1103/PhysRevLett.23.880} {\bibfield  {journal}
  {\bibinfo  {journal} {Phys. Rev. Lett.}\ }\textbf {\bibinfo {volume} {23}},\
  \bibinfo {pages} {880} (\bibinfo {year} {1969})}\BibitemShut {NoStop}%
\bibitem [{\citenamefont {Mermin}(1990)}]{Mermin:1990:PRL}%
  \BibitemOpen
  \bibfield  {author} {\bibinfo {author} {\bibfnamefont {N.~D.}\ \bibnamefont
  {Mermin}},\ }\bibfield  {title} {\enquote {\bibinfo {title} {Extreme quantum
  entanglement in a superposition of macroscopically distinct states},}\ }\href
  {https://link.aps.org/doi/10.1103/PhysRevLett.65.1838} {\bibfield  {journal}
  {\bibinfo  {journal} {Phys. Rev. Lett.}\ }\textbf {\bibinfo {volume} {65}},\
  \bibinfo {pages} {1838} (\bibinfo {year} {1990})}\BibitemShut {NoStop}%
\bibitem [{\citenamefont {Popescu}\ and\ \citenamefont
  {Rohrlich}(1994)}]{Popescu1994}%
  \BibitemOpen
  \bibfield  {author} {\bibinfo {author} {\bibfnamefont {S.}~\bibnamefont
  {Popescu}}\ and\ \bibinfo {author} {\bibfnamefont {D.}~\bibnamefont
  {Rohrlich}},\ }\bibfield  {title} {\enquote {\bibinfo {title} {Quantum
  nonlocality as an axiom},}\ }\href {\doibase 10.1007/BF02058098} {\bibfield
  {journal} {\bibinfo  {journal} {Found. Phys.}\ }\textbf {\bibinfo {volume}
  {24}},\ \bibinfo {pages} {379} (\bibinfo {year} {1994})}\BibitemShut
  {NoStop}%
\bibitem [{\citenamefont {Tsirelson}(1993)}]{Tsirelson:1993}%
  \BibitemOpen
  \bibfield  {author} {\bibinfo {author} {\bibfnamefont {B.}~\bibnamefont
  {Tsirelson}},\ }\bibfield  {title} {\enquote {\bibinfo {title} {{Some results
  and problems on quantum Bell-type inequalities}},}\ }\href@noop {} {\bibfield
   {journal} {\bibinfo  {journal} {Hadronic J. Supp.}\ }\textbf {\bibinfo
  {volume} {8}},\ \bibinfo {pages} {329} (\bibinfo {year} {1993})}\BibitemShut
  {NoStop}%
\bibitem [{\citenamefont {Pironio}(2005)}]{pironio05a}%
  \BibitemOpen
  \bibfield  {author} {\bibinfo {author} {\bibfnamefont {S.}~\bibnamefont
  {Pironio}},\ }\bibfield  {title} {\enquote {\bibinfo {title} {Lifting {B}ell
  inequalities},}\ }\href {\doibase 10.1063/1.1928727} {\bibfield  {journal}
  {\bibinfo  {journal} {J. Math. Phys.}\ }\textbf {\bibinfo {volume} {46}},\
  \bibinfo {pages} {062112} (\bibinfo {year} {2005})}\BibitemShut {NoStop}%
\bibitem [{\citenamefont {Branciard}(2011)}]{Branciard:2011}%
  \BibitemOpen
  \bibfield  {author} {\bibinfo {author} {\bibfnamefont {C.}~\bibnamefont
  {Branciard}},\ }\bibfield  {title} {\enquote {\bibinfo {title} {{Detection
  loophole in Bell experiments: how postselection modifies the requirements to
  observe nonlocality}},}\ }\href {\doibase 10.1103/PhysRevA.83.032123}
  {\bibfield  {journal} {\bibinfo  {journal} {Phys. Rev. A}\ }\textbf {\bibinfo
  {volume} {83}},\ \bibinfo {pages} {032123} (\bibinfo {year}
  {2011})}\BibitemShut {NoStop}%
\bibitem [{\citenamefont {Allcock}\ \emph
  {et~al.}(2009{\natexlab{a}})\citenamefont {Allcock}, \citenamefont {Brunner},
  \citenamefont {Paw{\l}owski},\ and\ \citenamefont {Scarani}}]{Allcock:2009}%
  \BibitemOpen
  \bibfield  {author} {\bibinfo {author} {\bibfnamefont {J.}~\bibnamefont
  {Allcock}}, \bibinfo {author} {\bibfnamefont {N.}~\bibnamefont {Brunner}},
  \bibinfo {author} {\bibfnamefont {M.}~\bibnamefont {Paw{\l}owski}}, \ and\
  \bibinfo {author} {\bibfnamefont {V.}~\bibnamefont {Scarani}},\ }\bibfield
  {title} {\enquote {\bibinfo {title} {Recovering part of the boundary between
  quantum and nonquantum correlations from information causality},}\ }\href
  {https://link.aps.org/doi/10.1103/PhysRevA.80.040103} {\bibfield  {journal}
  {\bibinfo  {journal} {Phys. Rev. A}\ }\textbf {\bibinfo {volume} {80}},\
  \bibinfo {pages} {040103} (\bibinfo {year} {2009}{\natexlab{a}})}\BibitemShut
  {NoStop}%
\bibitem [{\citenamefont {Allcock}\ \emph
  {et~al.}(2009{\natexlab{b}})\citenamefont {Allcock}, \citenamefont {Brunner},
  \citenamefont {Linden}, \citenamefont {Popescu}, \citenamefont {Skrzypczyk},\
  and\ \citenamefont {V\'ertesi}}]{Allcock:2009b}%
  \BibitemOpen
  \bibfield  {author} {\bibinfo {author} {\bibfnamefont {J.}~\bibnamefont
  {Allcock}}, \bibinfo {author} {\bibfnamefont {N.}~\bibnamefont {Brunner}},
  \bibinfo {author} {\bibfnamefont {N.}~\bibnamefont {Linden}}, \bibinfo
  {author} {\bibfnamefont {S.}~\bibnamefont {Popescu}}, \bibinfo {author}
  {\bibfnamefont {P.}~\bibnamefont {Skrzypczyk}}, \ and\ \bibinfo {author}
  {\bibfnamefont {T.}~\bibnamefont {V\'ertesi}},\ }\bibfield  {title} {\enquote
  {\bibinfo {title} {Closed sets of nonlocal correlations},}\ }\href {\doibase
  10.1103/PhysRevA.80.062107} {\bibfield  {journal} {\bibinfo  {journal} {Phys.
  Rev. A}\ }\textbf {\bibinfo {volume} {80}},\ \bibinfo {pages} {062107}
  (\bibinfo {year} {2009}{\natexlab{b}})}\BibitemShut {NoStop}%
\bibitem [{\citenamefont {Almeida}\ \emph {et~al.}(2010)\citenamefont
  {Almeida}, \citenamefont {Bancal}, \citenamefont {Brunner}, \citenamefont
  {Ac\'{\i}n}, \citenamefont {Gisin},\ and\ \citenamefont {Pironio}}]{GYNI}%
  \BibitemOpen
  \bibfield  {author} {\bibinfo {author} {\bibfnamefont {M.~L.}\ \bibnamefont
  {Almeida}}, \bibinfo {author} {\bibfnamefont {J.-D.}\ \bibnamefont {Bancal}},
  \bibinfo {author} {\bibfnamefont {N.}~\bibnamefont {Brunner}}, \bibinfo
  {author} {\bibfnamefont {A.}~\bibnamefont {Ac\'{\i}n}}, \bibinfo {author}
  {\bibfnamefont {N.}~\bibnamefont {Gisin}}, \ and\ \bibinfo {author}
  {\bibfnamefont {S.}~\bibnamefont {Pironio}},\ }\bibfield  {title} {\enquote
  {\bibinfo {title} {{Guess your neighbor's input: a multipartite nonlocal game
  with no quantum advantage}},}\ }\href
  {https://link.aps.org/doi/10.1103/PhysRevLett.104.230404} {\bibfield
  {journal} {\bibinfo  {journal} {Phys. Rev. Lett.}\ }\textbf {\bibinfo
  {volume} {104}},\ \bibinfo {pages} {230404} (\bibinfo {year}
  {2010})}\BibitemShut {NoStop}%
\bibitem [{\citenamefont {Yang}\ \emph {et~al.}(2011)\citenamefont {Yang},
  \citenamefont {Navascu\'es}, \citenamefont {Sheridan},\ and\ \citenamefont
  {Scarani}}]{Yang:2011}%
  \BibitemOpen
  \bibfield  {author} {\bibinfo {author} {\bibfnamefont {T.~H.}\ \bibnamefont
  {Yang}}, \bibinfo {author} {\bibfnamefont {M.}~\bibnamefont {Navascu\'es}},
  \bibinfo {author} {\bibfnamefont {L.}~\bibnamefont {Sheridan}}, \ and\
  \bibinfo {author} {\bibfnamefont {V.}~\bibnamefont {Scarani}},\ }\bibfield
  {title} {\enquote {\bibinfo {title} {{Quantum Bell inequalities from
  macroscopic locality}},}\ }\href {\doibase 10.1103/PhysRevA.83.022105}
  {\bibfield  {journal} {\bibinfo  {journal} {Phys. Rev. A}\ }\textbf {\bibinfo
  {volume} {83}},\ \bibinfo {pages} {022105} (\bibinfo {year}
  {2011})}\BibitemShut {NoStop}%
\bibitem [{\citenamefont {Fritz}\ \emph {et~al.}(2013)\citenamefont {Fritz},
  \citenamefont {Sainz}, \citenamefont {Augusiak}, \citenamefont {Brask},
  \citenamefont {Chaves}, \citenamefont {Leverrier},\ and\ \citenamefont
  {Ac\'{\i}n}}]{Fritz:2013}%
  \BibitemOpen
  \bibfield  {author} {\bibinfo {author} {\bibfnamefont {T.}~\bibnamefont
  {Fritz}}, \bibinfo {author} {\bibfnamefont {A.~B.}\ \bibnamefont {Sainz}},
  \bibinfo {author} {\bibfnamefont {R.}~\bibnamefont {Augusiak}}, \bibinfo
  {author} {\bibfnamefont {J.}~\bibnamefont {Brask}}, \bibinfo {author}
  {\bibfnamefont {R.}~\bibnamefont {Chaves}}, \bibinfo {author} {\bibfnamefont
  {A.}~\bibnamefont {Leverrier}}, \ and\ \bibinfo {author} {\bibfnamefont
  {A.}~\bibnamefont {Ac\'{\i}n}},\ }\bibfield  {title} {\enquote {\bibinfo
  {title} {Local orthogonality as a multipartite principle for quantum
  correlations},}\ }\href {\doibase 10.1038/ncomms3263} {\bibfield  {journal}
  {\bibinfo  {journal} {Nat. Commun.}\ }\textbf {\bibinfo {volume} {4}},\
  \bibinfo {pages} {2263} (\bibinfo {year} {2013})}\BibitemShut {NoStop}%
\bibitem [{\citenamefont {P{\"u}tz}\ \emph {et~al.}(2014)\citenamefont
  {P{\"u}tz}, \citenamefont {Rosset}, \citenamefont {Barnea}, \citenamefont
  {Liang},\ and\ \citenamefont {Gisin}}]{Putz2014}%
  \BibitemOpen
  \bibfield  {author} {\bibinfo {author} {\bibfnamefont {G.}~\bibnamefont
  {P{\"u}tz}}, \bibinfo {author} {\bibfnamefont {D.}~\bibnamefont {Rosset}},
  \bibinfo {author} {\bibfnamefont {T.~J.}\ \bibnamefont {Barnea}}, \bibinfo
  {author} {\bibfnamefont {Y.-C.}\ \bibnamefont {Liang}}, \ and\ \bibinfo
  {author} {\bibfnamefont {N.}~\bibnamefont {Gisin}},\ }\bibfield  {title}
  {\enquote {\bibinfo {title} {Arbitrarily small amount of measurement
  independence is sufficient to manifest quantum nonlocality},}\ }\href
  {https://link.aps.org/doi/10.1103/PhysRevLett.113.190402} {\bibfield
  {journal} {\bibinfo  {journal} {Physical Review Letters}\ }\textbf {\bibinfo
  {volume} {113}},\ \bibinfo {pages} {190402} (\bibinfo {year}
  {2014})}\BibitemShut {NoStop}%
\bibitem [{\citenamefont {Chaves}\ \emph {et~al.}(2015)\citenamefont {Chaves},
  \citenamefont {Majenz},\ and\ \citenamefont {Gross}}]{Chaves:2015aa}%
  \BibitemOpen
  \bibfield  {author} {\bibinfo {author} {\bibfnamefont {R.}~\bibnamefont
  {Chaves}}, \bibinfo {author} {\bibfnamefont {C.}~\bibnamefont {Majenz}}, \
  and\ \bibinfo {author} {\bibfnamefont {D.}~\bibnamefont {Gross}},\ }\bibfield
   {title} {\enquote {\bibinfo {title} {Information-theoretic implications of
  quantum causal structures},}\ }\href {http://dx.doi.org/10.1038/ncomms6766}
  {\bibfield  {journal} {\bibinfo  {journal} {Nat. Commun.}\ }\textbf {\bibinfo
  {volume} {6}},\ \bibinfo {pages} {5766} (\bibinfo {year} {2015})}\BibitemShut
  {NoStop}%
\bibitem [{\citenamefont {de~Vicente}(2015)}]{deVicente:2015}%
  \BibitemOpen
  \bibfield  {author} {\bibinfo {author} {\bibfnamefont {J.~I.}\ \bibnamefont
  {de~Vicente}},\ }\bibfield  {title} {\enquote {\bibinfo {title} {Simple
  conditions constraining the set of quantum correlations},}\ }\href {\doibase
  10.1103/PhysRevA.92.032103} {\bibfield  {journal} {\bibinfo  {journal} {Phys.
  Rev. A}\ }\textbf {\bibinfo {volume} {92}},\ \bibinfo {pages} {032103}
  (\bibinfo {year} {2015})}\BibitemShut {NoStop}%
\bibitem [{\citenamefont {Christensen}\ \emph {et~al.}(2015)\citenamefont
  {Christensen}, \citenamefont {Liang}, \citenamefont {Brunner}, \citenamefont
  {Gisin},\ and\ \citenamefont {Kwiat}}]{Christensen:2015}%
  \BibitemOpen
  \bibfield  {author} {\bibinfo {author} {\bibfnamefont {B.~G.}\ \bibnamefont
  {Christensen}}, \bibinfo {author} {\bibfnamefont {Y.-C.}\ \bibnamefont
  {Liang}}, \bibinfo {author} {\bibfnamefont {N.}~\bibnamefont {Brunner}},
  \bibinfo {author} {\bibfnamefont {N.}~\bibnamefont {Gisin}}, \ and\ \bibinfo
  {author} {\bibfnamefont {P.~G.}\ \bibnamefont {Kwiat}},\ }\bibfield  {title}
  {\enquote {\bibinfo {title} {Exploring the limits of quantum nonlocality with
  entangled photons},}\ }\href {\doibase 10.1103/PhysRevX.5.041052} {\bibfield
  {journal} {\bibinfo  {journal} {Phys. Rev. X}\ }\textbf {\bibinfo {volume}
  {5}},\ \bibinfo {pages} {041052} (\bibinfo {year} {2015})}\BibitemShut
  {NoStop}%
\bibitem [{\citenamefont {Lin}\ \emph {et~al.}(2017)\citenamefont {Lin},
  \citenamefont {Rosset}, \citenamefont {Zhang}, \citenamefont {Bancal},\ and\
  \citenamefont {Liang}}]{Lin:2017}%
  \BibitemOpen
  \bibfield  {author} {\bibinfo {author} {\bibfnamefont {P.-S.}\ \bibnamefont
  {Lin}}, \bibinfo {author} {\bibfnamefont {D.}~\bibnamefont {Rosset}},
  \bibinfo {author} {\bibfnamefont {Y.}~\bibnamefont {Zhang}}, \bibinfo
  {author} {\bibfnamefont {J.-D.}\ \bibnamefont {Bancal}}, \ and\ \bibinfo
  {author} {\bibfnamefont {Y.-C.}\ \bibnamefont {Liang}},\ }\bibfield  {title}
  {\enquote {\bibinfo {title} {{Device-independent point estimation from finite
  data}},}\ }\href {https://arxiv.org/abs/1705.09245} {\bibfield  {journal}
  {\bibinfo  {journal} {arXiv:1705.09245}\ } (\bibinfo {year}
  {2017})}\BibitemShut {NoStop}%
\bibitem [{\citenamefont {Zhou}\ \emph {et~al.}(2017)\citenamefont {Zhou},
  \citenamefont {Cai}, \citenamefont {Bancal}, \citenamefont {Gao},\ and\
  \citenamefont {Scarani}}]{Zhou:2017}%
  \BibitemOpen
  \bibfield  {author} {\bibinfo {author} {\bibfnamefont {Y.}~\bibnamefont
  {Zhou}}, \bibinfo {author} {\bibfnamefont {Y.}~\bibnamefont {Cai}}, \bibinfo
  {author} {\bibfnamefont {J.-D.}\ \bibnamefont {Bancal}}, \bibinfo {author}
  {\bibfnamefont {F.}~\bibnamefont {Gao}}, \ and\ \bibinfo {author}
  {\bibfnamefont {V.}~\bibnamefont {Scarani}},\ }\bibfield  {title} {\enquote
  {\bibinfo {title} {{Many-box locality}},}\ }\href
  {https://arxiv.org/abs/1708.03067} {\bibfield  {journal} {\bibinfo  {journal}
  {arXiv:1708.03067}\ } (\bibinfo {year} {2017})}\BibitemShut {NoStop}%
\bibitem [{\citenamefont {Tsirelson}(1987)}]{Tsirelson:1987}%
  \BibitemOpen
  \bibfield  {author} {\bibinfo {author} {\bibfnamefont {B.~S.}\ \bibnamefont
  {Tsirelson}},\ }\bibfield  {title} {\enquote {\bibinfo {title} {{Quantum
  analogues of the Bell inequalities. The case of two spatially separated
  domains}},}\ }\href {\doibase 10.1007/BF01663472} {\bibfield  {journal}
  {\bibinfo  {journal} {J. Soviet Math.}\ }\textbf {\bibinfo {volume} {36}},\
  \bibinfo {pages} {557} (\bibinfo {year} {1987})}\BibitemShut {NoStop}%
\bibitem [{\citenamefont {Landau}(1988)}]{Landau:1988}%
  \BibitemOpen
  \bibfield  {author} {\bibinfo {author} {\bibfnamefont {L.~J.}\ \bibnamefont
  {Landau}},\ }\bibfield  {title} {\enquote {\bibinfo {title} {Empirical
  two-point correlation functions},}\ }\href {\doibase 10.1007/BF00732549}
  {\bibfield  {journal} {\bibinfo  {journal} {Found. Phys.}\ }\textbf {\bibinfo
  {volume} {18}},\ \bibinfo {pages} {449} (\bibinfo {year} {1988})}\BibitemShut
  {NoStop}%
\bibitem [{\citenamefont {Werner}\ and\ \citenamefont
  {Wolf}(2001)}]{Werner:2001}%
  \BibitemOpen
  \bibfield  {author} {\bibinfo {author} {\bibfnamefont {R.~F.}\ \bibnamefont
  {Werner}}\ and\ \bibinfo {author} {\bibfnamefont {M.~M.}\ \bibnamefont
  {Wolf}},\ }\bibfield  {title} {\enquote {\bibinfo {title} {All-multipartite
  bell-correlation inequalities for two dichotomic observables per site},}\
  }\href {https://link.aps.org/doi/10.1103/PhysRevA.64.032112} {\bibfield
  {journal} {\bibinfo  {journal} {Phys. Rev. A}\ }\textbf {\bibinfo {volume}
  {64}},\ \bibinfo {pages} {032112} (\bibinfo {year} {2001})}\BibitemShut
  {NoStop}%
\bibitem [{\citenamefont {Masanes}(2006)}]{Masanes:2006}%
  \BibitemOpen
  \bibfield  {author} {\bibinfo {author} {\bibfnamefont {L.}~\bibnamefont
  {Masanes}},\ }\bibfield  {title} {\enquote {\bibinfo {title} {{Asymptotic
  violation of Bell inequalities and distillability}},}\ }\href
  {https://link.aps.org/doi/10.1103/PhysRevLett.97.050503} {\bibfield
  {journal} {\bibinfo  {journal} {Phys. Rev. Lett.}\ }\textbf {\bibinfo
  {volume} {97}},\ \bibinfo {pages} {050503} (\bibinfo {year}
  {2006})}\BibitemShut {NoStop}%
\bibitem [{\citenamefont {Cabello}(2005)}]{Cabello:2005}%
  \BibitemOpen
  \bibfield  {author} {\bibinfo {author} {\bibfnamefont {A.}~\bibnamefont
  {Cabello}},\ }\bibfield  {title} {\enquote {\bibinfo {title} {How much larger
  quantum correlations are than classical ones},}\ }\href {\doibase
  10.1103/PhysRevA.72.012113} {\bibfield  {journal} {\bibinfo  {journal} {Phys.
  Rev. A}\ }\textbf {\bibinfo {volume} {72}},\ \bibinfo {pages} {012113}
  (\bibinfo {year} {2005})}\BibitemShut {NoStop}%
\bibitem [{\citenamefont {Pitowsky}(1989)}]{Pitowsky:Book}%
  \BibitemOpen
  \bibfield  {author} {\bibinfo {author} {\bibfnamefont {I.}~\bibnamefont
  {Pitowsky}},\ }\href@noop {} {\emph {\bibinfo {title} {Quantum probability
  --- quantum logic}}},\ \bibinfo {edition} {1st}\ ed.,\ \bibinfo {series}
  {Lecture Notes in Physics}, Vol.\ \bibinfo {volume} {321}\ (\bibinfo
  {publisher} {Springer-Verlag Berlin Heidelberg},\ \bibinfo {year}
  {1989})\BibitemShut {NoStop}%
\bibitem [{\citenamefont {Navascu{\'e}s}\ and\ \citenamefont
  {Wunderlich}(2010)}]{Navascues2009}%
  \BibitemOpen
  \bibfield  {author} {\bibinfo {author} {\bibfnamefont {M.}~\bibnamefont
  {Navascu{\'e}s}}\ and\ \bibinfo {author} {\bibfnamefont {H.}~\bibnamefont
  {Wunderlich}},\ }\bibfield  {title} {\enquote {\bibinfo {title} {A glance
  beyond the quantum model},}\ }\href
  {http://dx.doi.org/10.1098/rspa.2009.0453} {\bibfield  {journal} {\bibinfo
  {journal} {Proc. R. Soc. A}\ }\textbf {\bibinfo {volume} {466}},\ \bibinfo
  {pages} {881} (\bibinfo {year} {2010})}\BibitemShut {NoStop}%
\bibitem [{\citenamefont {Paw{\l}owski}\ \emph {et~al.}(2009)\citenamefont
  {Paw{\l}owski}, \citenamefont {Paterek}, \citenamefont {Kaszlikowski},
  \citenamefont {Scarani}, \citenamefont {Winter},\ and\ \citenamefont
  {{\.Z}ukowski}}]{Pawlowski2009}%
  \BibitemOpen
  \bibfield  {author} {\bibinfo {author} {\bibfnamefont {M.}~\bibnamefont
  {Paw{\l}owski}}, \bibinfo {author} {\bibfnamefont {T.}~\bibnamefont
  {Paterek}}, \bibinfo {author} {\bibfnamefont {D.}~\bibnamefont
  {Kaszlikowski}}, \bibinfo {author} {\bibfnamefont {V.}~\bibnamefont
  {Scarani}}, \bibinfo {author} {\bibfnamefont {A.}~\bibnamefont {Winter}}, \
  and\ \bibinfo {author} {\bibfnamefont {M.}~\bibnamefont {{\.Z}ukowski}},\
  }\bibfield  {title} {\enquote {\bibinfo {title} {Information causality as a
  physical principle},}\ }\href {http://dx.doi.org/10.1038/nature08400}
  {\bibfield  {journal} {\bibinfo  {journal} {Nature}\ }\textbf {\bibinfo
  {volume} {461}},\ \bibinfo {pages} {1101} (\bibinfo {year}
  {2009})}\BibitemShut {NoStop}%
\bibitem [{\citenamefont {Navascu{\'{e}}s}\ \emph {et~al.}(2015)\citenamefont
  {Navascu{\'{e}}s}, \citenamefont {Guryanova}, \citenamefont {Hoban},\ and\
  \citenamefont {Ac{\'{\i}}n}}]{Navascus:2015}%
  \BibitemOpen
  \bibfield  {author} {\bibinfo {author} {\bibfnamefont {M.}~\bibnamefont
  {Navascu{\'{e}}s}}, \bibinfo {author} {\bibfnamefont {Y.}~\bibnamefont
  {Guryanova}}, \bibinfo {author} {\bibfnamefont {M.~J.}\ \bibnamefont
  {Hoban}}, \ and\ \bibinfo {author} {\bibfnamefont {A.}~\bibnamefont
  {Ac{\'{\i}}n}},\ }\bibfield  {title} {\enquote {\bibinfo {title} {{Almost
  quantum correlations}},}\ }\href {\doibase 10.1038/ncomms7288} {\bibfield
  {journal} {\bibinfo  {journal} {Nat. Commun.}\ }\textbf {\bibinfo {volume}
  {6}},\ \bibinfo {pages} {6288} (\bibinfo {year} {2015})}\BibitemShut
  {NoStop}%
\bibitem [{\citenamefont {Popescu}\ and\ \citenamefont
  {Rohrlich}(1992)}]{Popescu:1992}%
  \BibitemOpen
  \bibfield  {author} {\bibinfo {author} {\bibfnamefont {S.}~\bibnamefont
  {Popescu}}\ and\ \bibinfo {author} {\bibfnamefont {D.}~\bibnamefont
  {Rohrlich}},\ }\bibfield  {title} {\enquote {\bibinfo {title} {{Which states
  violate Bell's inequality maximally?}}}\ }\href
  {http://dx.doi.org/10.1016/0375-9601(92)90819-8} {\bibfield  {journal}
  {\bibinfo  {journal} {Phys. Lett. A}\ }\textbf {\bibinfo {volume} {169}},\
  \bibinfo {pages} {411} (\bibinfo {year} {1992})}\BibitemShut {NoStop}%
\bibitem [{\citenamefont {Mayers}\ and\ \citenamefont
  {Yao}(2004)}]{Mayers:2004}%
  \BibitemOpen
  \bibfield  {author} {\bibinfo {author} {\bibfnamefont {D.}~\bibnamefont
  {Mayers}}\ and\ \bibinfo {author} {\bibfnamefont {A.}~\bibnamefont {Yao}},\
  }\bibfield  {title} {\enquote {\bibinfo {title} {Self-testing quantum
  apparatus},}\ }\href@noop {} {\bibfield  {journal} {\bibinfo  {journal}
  {Quantum Inf. Comput.}\ }\textbf {\bibinfo {volume} {4}},\ \bibinfo {pages}
  {273} (\bibinfo {year} {2004})}\BibitemShut {NoStop}%
\bibitem [{\citenamefont {Hardy}(1992)}]{Hardy1992}%
  \BibitemOpen
  \bibfield  {author} {\bibinfo {author} {\bibfnamefont {L.}~\bibnamefont
  {Hardy}},\ }\bibfield  {title} {\enquote {\bibinfo {title} {{Quantum
  mechanics, local realistic theories, and Lorentz-invariant realistic
  theories}},}\ }\href {http://dx.doi.org/10.1103/PhysRevLett.68.2981}
  {\bibfield  {journal} {\bibinfo  {journal} {Phys. Rev. Lett.}\ }\textbf
  {\bibinfo {volume} {68}},\ \bibinfo {pages} {2981} (\bibinfo {year}
  {1992})}\BibitemShut {NoStop}%
\bibitem [{\citenamefont {Donohue}\ and\ \citenamefont
  {Wolfe}(2015)}]{Donohue:2015}%
  \BibitemOpen
  \bibfield  {author} {\bibinfo {author} {\bibfnamefont {J.~M.}\ \bibnamefont
  {Donohue}}\ and\ \bibinfo {author} {\bibfnamefont {E.}~\bibnamefont
  {Wolfe}},\ }\bibfield  {title} {\enquote {\bibinfo {title} {{Identifying
  nonconvexity in the sets of limited-dimension quantum correlations}},}\
  }\href {https://link.aps.org/doi/10.1103/PhysRevA.92.062120} {\bibfield
  {journal} {\bibinfo  {journal} {Phys. Rev. A}\ }\textbf {\bibinfo {volume}
  {92}},\ \bibinfo {pages} {062120} (\bibinfo {year} {2015})}\BibitemShut
  {NoStop}%
\bibitem [{\citenamefont {Navascu{\'e}s}\ \emph {et~al.}(2008)\citenamefont
  {Navascu{\'e}s}, \citenamefont {Pironio},\ and\ \citenamefont
  {Ac{\'\i}n}}]{Navascues:2008}%
  \BibitemOpen
  \bibfield  {author} {\bibinfo {author} {\bibfnamefont {M.}~\bibnamefont
  {Navascu{\'e}s}}, \bibinfo {author} {\bibfnamefont {S.}~\bibnamefont
  {Pironio}}, \ and\ \bibinfo {author} {\bibfnamefont {A.}~\bibnamefont
  {Ac{\'\i}n}},\ }\bibfield  {title} {\enquote {\bibinfo {title} {{A convergent
  hierarchy of semidefinite programs characterizing the set of quantum
  correlations}},}\ }\href {\doibase 10.1088/1367-2630/10/7/073013} {\bibfield
  {journal} {\bibinfo  {journal} {New J. Phys.}\ }\textbf {\bibinfo {volume}
  {10}},\ \bibinfo {pages} {073013} (\bibinfo {year} {2008})}\BibitemShut
  {NoStop}%
\bibitem [{\citenamefont {Barrett}\ \emph {et~al.}(2005)\citenamefont
  {Barrett}, \citenamefont {Linden}, \citenamefont {Massar}, \citenamefont
  {Pironio}, \citenamefont {Popescu},\ and\ \citenamefont
  {Roberts}}]{Barrett:2005:Feb}%
  \BibitemOpen
  \bibfield  {author} {\bibinfo {author} {\bibfnamefont {J.}~\bibnamefont
  {Barrett}}, \bibinfo {author} {\bibfnamefont {N.}~\bibnamefont {Linden}},
  \bibinfo {author} {\bibfnamefont {S.}~\bibnamefont {Massar}}, \bibinfo
  {author} {\bibfnamefont {S.}~\bibnamefont {Pironio}}, \bibinfo {author}
  {\bibfnamefont {S.}~\bibnamefont {Popescu}}, \ and\ \bibinfo {author}
  {\bibfnamefont {D.}~\bibnamefont {Roberts}},\ }\bibfield  {title} {\enquote
  {\bibinfo {title} {Nonlocal correlations as an information-theoretic
  resource},}\ }\href {https://link.aps.org/doi/10.1103/PhysRevA.71.022101}
  {\bibfield  {journal} {\bibinfo  {journal} {Phys. Rev. A}\ }\textbf {\bibinfo
  {volume} {71}},\ \bibinfo {pages} {022101} (\bibinfo {year}
  {2005})}\BibitemShut {NoStop}%
\bibitem [{\citenamefont {Barrett}\ and\ \citenamefont
  {Pironio}(2005)}]{Barrett:2005:Sep}%
  \BibitemOpen
  \bibfield  {author} {\bibinfo {author} {\bibfnamefont {J.}~\bibnamefont
  {Barrett}}\ and\ \bibinfo {author} {\bibfnamefont {S.}~\bibnamefont
  {Pironio}},\ }\bibfield  {title} {\enquote {\bibinfo {title}
  {{Popescu-Rohrlich correlations as a unit of nonlocality}},}\ }\href
  {https://link.aps.org/doi/10.1103/PhysRevLett.95.140401} {\bibfield
  {journal} {\bibinfo  {journal} {Phys. Rev. Lett.}\ }\textbf {\bibinfo
  {volume} {95}},\ \bibinfo {pages} {140401} (\bibinfo {year}
  {2005})}\BibitemShut {NoStop}%
\bibitem [{\citenamefont {Ac{\'\i}n}\ \emph {et~al.}(2015)\citenamefont
  {Ac{\'\i}n}, \citenamefont {Fritz}, \citenamefont {Leverrier},\ and\
  \citenamefont {Sainz}}]{Acin:2015}%
  \BibitemOpen
  \bibfield  {author} {\bibinfo {author} {\bibfnamefont {A.}~\bibnamefont
  {Ac{\'\i}n}}, \bibinfo {author} {\bibfnamefont {T.}~\bibnamefont {Fritz}},
  \bibinfo {author} {\bibfnamefont {A.}~\bibnamefont {Leverrier}}, \ and\
  \bibinfo {author} {\bibfnamefont {A.~B.}\ \bibnamefont {Sainz}},\ }\bibfield
  {title} {\enquote {\bibinfo {title} {{A combinatorial approach to nonlocality
  and contextuality}},}\ }\href {\doibase 10.1007/s00220-014-2260-1} {\bibfield
   {journal} {\bibinfo  {journal} {Commun. Math. Phys.}\ }\textbf {\bibinfo
  {volume} {334}},\ \bibinfo {pages} {533} (\bibinfo {year}
  {2015})}\BibitemShut {NoStop}%
\bibitem [{\citenamefont {Jones}\ and\ \citenamefont
  {Masanes}(2005)}]{Jones:2005}%
  \BibitemOpen
  \bibfield  {author} {\bibinfo {author} {\bibfnamefont {N.~S.}\ \bibnamefont
  {Jones}}\ and\ \bibinfo {author} {\bibfnamefont {L.}~\bibnamefont
  {Masanes}},\ }\bibfield  {title} {\enquote {\bibinfo {title} {Interconversion
  of nonlocal correlations},}\ }\href
  {https://link.aps.org/doi/10.1103/PhysRevA.72.052312} {\bibfield  {journal}
  {\bibinfo  {journal} {Phys. Rev. A}\ }\textbf {\bibinfo {volume} {72}},\
  \bibinfo {pages} {052312} (\bibinfo {year} {2005})}\BibitemShut {NoStop}%
\bibitem [{\citenamefont {Barnum}\ \emph {et~al.}(2010)\citenamefont {Barnum},
  \citenamefont {Beigi}, \citenamefont {Boixo}, \citenamefont {Elliott},\ and\
  \citenamefont {Wehner}}]{Barnum:2010}%
  \BibitemOpen
  \bibfield  {author} {\bibinfo {author} {\bibfnamefont {H.}~\bibnamefont
  {Barnum}}, \bibinfo {author} {\bibfnamefont {S.}~\bibnamefont {Beigi}},
  \bibinfo {author} {\bibfnamefont {S.}~\bibnamefont {Boixo}}, \bibinfo
  {author} {\bibfnamefont {M.~B.}\ \bibnamefont {Elliott}}, \ and\ \bibinfo
  {author} {\bibfnamefont {S.}~\bibnamefont {Wehner}},\ }\bibfield  {title}
  {\enquote {\bibinfo {title} {Local quantum measurement and no-signaling imply
  quantum correlations},}\ }\href
  {https://link.aps.org/doi/10.1103/PhysRevLett.104.140401} {\bibfield
  {journal} {\bibinfo  {journal} {Phys. Rev. Lett.}\ }\textbf {\bibinfo
  {volume} {104}},\ \bibinfo {pages} {140401} (\bibinfo {year}
  {2010})}\BibitemShut {NoStop}%
\bibitem [{\citenamefont {{Slofstra}}(2017)}]{Slofstra:2017}%
  \BibitemOpen
  \bibfield  {author} {\bibinfo {author} {\bibfnamefont {W.}~\bibnamefont
  {{Slofstra}}},\ }\bibfield  {title} {\enquote {\bibinfo {title} {{The set of
  quantum correlations is not closed}},}\ }\href
  {https://arxiv.org/abs/1703.08618} {\bibfield  {journal} {\bibinfo  {journal}
  {arXiv:1703.08618}\ } (\bibinfo {year} {2017})}\BibitemShut {NoStop}%
\bibitem [{\citenamefont {Dykema}\ \emph {et~al.}(2017)\citenamefont {Dykema},
  \citenamefont {Paulsen},\ and\ \citenamefont {Prakash}}]{dykema17a}%
  \BibitemOpen
  \bibfield  {author} {\bibinfo {author} {\bibfnamefont {K.}~\bibnamefont
  {Dykema}}, \bibinfo {author} {\bibfnamefont {V.~I.}\ \bibnamefont {Paulsen}},
  \ and\ \bibinfo {author} {\bibfnamefont {J.}~\bibnamefont {Prakash}},\
  }\bibfield  {title} {\enquote {\bibinfo {title} {{Non-closure of the set of
  quantum correlations via graphs}},}\ }\href
  {https://arxiv.org/abs/1709.05032} {\bibfield  {journal} {\bibinfo  {journal}
  {arXiv:1709.05032}\ } (\bibinfo {year} {2017})}\BibitemShut {NoStop}%
\bibitem [{\citenamefont {Rosset}\ \emph {et~al.}(2014)\citenamefont {Rosset},
  \citenamefont {Bancal},\ and\ \citenamefont {Gisin}}]{Rosset2014a}%
  \BibitemOpen
  \bibfield  {author} {\bibinfo {author} {\bibfnamefont {D.}~\bibnamefont
  {Rosset}}, \bibinfo {author} {\bibfnamefont {J.-D.}\ \bibnamefont {Bancal}},
  \ and\ \bibinfo {author} {\bibfnamefont {N.}~\bibnamefont {Gisin}},\
  }\bibfield  {title} {\enquote {\bibinfo {title} {Classifying 50 years of
  {{Bell}} inequalities},}\ }\href {\doibase 10.1088/1751-8113/47/42/424022}
  {\bibfield  {journal} {\bibinfo  {journal} {J. Phys. A: Math. Theo.}\
  }\textbf {\bibinfo {volume} {47}},\ \bibinfo {pages} {424022} (\bibinfo
  {year} {2014})}\BibitemShut {NoStop}%
\bibitem [{\citenamefont {Ramanathan}\ \emph {et~al.}(2016)\citenamefont
  {Ramanathan}, \citenamefont {Tuziemski}, \citenamefont {Horodecki},\ and\
  \citenamefont {Horodecki}}]{Ramanathan:2016}%
  \BibitemOpen
  \bibfield  {author} {\bibinfo {author} {\bibfnamefont {R.}~\bibnamefont
  {Ramanathan}}, \bibinfo {author} {\bibfnamefont {J.}~\bibnamefont
  {Tuziemski}}, \bibinfo {author} {\bibfnamefont {M.}~\bibnamefont
  {Horodecki}}, \ and\ \bibinfo {author} {\bibfnamefont {P.}~\bibnamefont
  {Horodecki}},\ }\bibfield  {title} {\enquote {\bibinfo {title} {No quantum
  realization of extremal no-signaling boxes},}\ }\href
  {https://link.aps.org/doi/10.1103/PhysRevLett.117.050401} {\bibfield
  {journal} {\bibinfo  {journal} {Phys. Rev. Lett.}\ }\textbf {\bibinfo
  {volume} {117}},\ \bibinfo {pages} {050401} (\bibinfo {year}
  {2016})}\BibitemShut {NoStop}%
\bibitem [{\citenamefont {Cleve}\ \emph {et~al.}(2004)\citenamefont {Cleve},
  \citenamefont {H{\o}yer}, \citenamefont {Toner},\ and\ \citenamefont
  {Watrous}}]{cleve04a}%
  \BibitemOpen
  \bibfield  {author} {\bibinfo {author} {\bibfnamefont {R.}~\bibnamefont
  {Cleve}}, \bibinfo {author} {\bibfnamefont {P.}~\bibnamefont {H{\o}yer}},
  \bibinfo {author} {\bibfnamefont {B.}~\bibnamefont {Toner}}, \ and\ \bibinfo
  {author} {\bibfnamefont {J.}~\bibnamefont {Watrous}},\ }\bibfield  {title}
  {\enquote {\bibinfo {title} {{Consequences and limits of nonlocal
  strategies}},}\ }\href {http://ieeexplore.ieee.org/document/1313847}
  {\bibfield  {journal} {\bibinfo  {journal} {Proceedings 19th IEEE Annual
  Conference on Computational Complexity}\ } (\bibinfo {year}
  {2004})}\BibitemShut {NoStop}%
\bibitem [{\citenamefont {Fine}(1982)}]{Fine1982}%
  \BibitemOpen
  \bibfield  {author} {\bibinfo {author} {\bibfnamefont {A.}~\bibnamefont
  {Fine}},\ }\bibfield  {title} {\enquote {\bibinfo {title} {{Hidden variables,
  joint probability, and the Bell inequalities}},}\ }\href {\doibase
  10.1103/PhysRevLett.48.291} {\bibfield  {journal} {\bibinfo  {journal} {Phys.
  Rev. Lett.}\ }\textbf {\bibinfo {volume} {48}},\ \bibinfo {pages} {291}
  (\bibinfo {year} {1982})}\BibitemShut {NoStop}%
\bibitem [{\citenamefont {Liang}\ \emph {et~al.}(2010)\citenamefont {Liang},
  \citenamefont {Harrigan}, \citenamefont {Bartlett},\ and\ \citenamefont
  {Rudolph}}]{Liang:2010}%
  \BibitemOpen
  \bibfield  {author} {\bibinfo {author} {\bibfnamefont {Y.-C.}\ \bibnamefont
  {Liang}}, \bibinfo {author} {\bibfnamefont {N.}~\bibnamefont {Harrigan}},
  \bibinfo {author} {\bibfnamefont {S.~D.}\ \bibnamefont {Bartlett}}, \ and\
  \bibinfo {author} {\bibfnamefont {T.}~\bibnamefont {Rudolph}},\ }\bibfield
  {title} {\enquote {\bibinfo {title} {{Nonclassical correlations from randomly
  chosen local measurements}},}\ }\href
  {https://link.aps.org/doi/10.1103/PhysRevLett.104.050401} {\bibfield
  {journal} {\bibinfo  {journal} {Phys. Rev. Lett.}\ }\textbf {\bibinfo
  {volume} {104}},\ \bibinfo {pages} {050401} (\bibinfo {year}
  {2010})}\BibitemShut {NoStop}%
\bibitem [{\citenamefont {Scarani}(2008)}]{Scarani2008}%
  \BibitemOpen
  \bibfield  {author} {\bibinfo {author} {\bibfnamefont {V.}~\bibnamefont
  {Scarani}},\ }\bibfield  {title} {\enquote {\bibinfo {title} {Local and
  nonlocal content of bipartite qubit and qutrit correlations},}\ }\href
  {\doibase 10.1103/PhysRevA.77.042112} {\bibfield  {journal} {\bibinfo
  {journal} {Phys. Rev. A}\ }\textbf {\bibinfo {volume} {77}},\ \bibinfo
  {pages} {042112} (\bibinfo {year} {2008})}\BibitemShut {NoStop}%
\bibitem [{\citenamefont {Bierhorst}(2016)}]{bierhorst16a}%
  \BibitemOpen
  \bibfield  {author} {\bibinfo {author} {\bibfnamefont {P.}~\bibnamefont
  {Bierhorst}},\ }\bibfield  {title} {\enquote {\bibinfo {title} {{Geometric
  decompositions of Bell polytopes with practical applications}},}\ }\href
  {http://dx.doi.org/10.1088/1751-8113/49/21/215301} {\bibfield  {journal}
  {\bibinfo  {journal} {J. Phys. A: Math. Theor.}\ }\textbf {\bibinfo {volume}
  {49}} (\bibinfo {year} {2016})}\BibitemShut {NoStop}%
\bibitem [{\citenamefont {McKague}\ \emph {et~al.}(2012)\citenamefont
  {McKague}, \citenamefont {Yang},\ and\ \citenamefont
  {Scarani}}]{McKague:2012}%
  \BibitemOpen
  \bibfield  {author} {\bibinfo {author} {\bibfnamefont {M.}~\bibnamefont
  {McKague}}, \bibinfo {author} {\bibfnamefont {T.~H.}\ \bibnamefont {Yang}}, \
  and\ \bibinfo {author} {\bibfnamefont {V.}~\bibnamefont {Scarani}},\
  }\bibfield  {title} {\enquote {\bibinfo {title} {Robust self-testing of the
  singlet},}\ }\href {\doibase 10.1088/1751-8113/45/45/455304} {\bibfield
  {journal} {\bibinfo  {journal} {J. Phys. A}\ }\textbf {\bibinfo {volume}
  {45}},\ \bibinfo {pages} {455304} (\bibinfo {year} {2012})}\BibitemShut
  {NoStop}%
\bibitem [{\citenamefont {Bancal}(2014)}]{Bancal:2014}%
  \BibitemOpen
  \bibfield  {author} {\bibinfo {author} {\bibfnamefont {J.-D.}\ \bibnamefont
  {Bancal}},\ }\href {\doibase 10.1007/978-3-319-01183-7} {\emph {\bibinfo
  {title} {On the device-independent approach to quantum physics}}}\ (\bibinfo
  {publisher} {Springer International Publishing},\ \bibinfo {year}
  {2014})\BibitemShut {NoStop}%
\bibitem [{\citenamefont {Wolfe}\ and\ \citenamefont
  {Yelin}(2012)}]{Wolfe:2012}%
  \BibitemOpen
  \bibfield  {author} {\bibinfo {author} {\bibfnamefont {E.}~\bibnamefont
  {Wolfe}}\ and\ \bibinfo {author} {\bibfnamefont {S.~F.}\ \bibnamefont
  {Yelin}},\ }\bibfield  {title} {\enquote {\bibinfo {title} {{Quantum bounds
  for inequalities involving marginal expectation values}},}\ }\href
  {https://link.aps.org/doi/10.1103/PhysRevA.86.012123} {\bibfield  {journal}
  {\bibinfo  {journal} {Phys. Rev. A}\ }\textbf {\bibinfo {volume} {86}},\
  \bibinfo {pages} {012123} (\bibinfo {year} {2012})}\BibitemShut {NoStop}%
\bibitem [{\citenamefont {Masanes}(2005)}]{Masanes:2005}%
  \BibitemOpen
  \bibfield  {author} {\bibinfo {author} {\bibfnamefont {L.}~\bibnamefont
  {Masanes}},\ }\bibfield  {title} {\enquote {\bibinfo {title} {{Extremal
  quantum correlations for N parties with two dichotomic observables per
  site}},}\ }\href {https://arxiv.org/abs/quant-ph/0512100} {\bibfield
  {journal} {\bibinfo  {journal} {quant-ph/0512100}\ } (\bibinfo {year}
  {2005})}\BibitemShut {NoStop}%
\bibitem [{\citenamefont {Cirel'son}(1980)}]{Tsirelson:1980}%
  \BibitemOpen
  \bibfield  {author} {\bibinfo {author} {\bibfnamefont {B.~S.}\ \bibnamefont
  {Cirel'son}},\ }\bibfield  {title} {\enquote {\bibinfo {title} {{Quantum
  generalizations of Bell's inequality}},}\ }\href {\doibase
  10.1007/BF00417500} {\bibfield  {journal} {\bibinfo  {journal} {Lett. Math.
  Phys.}\ }\textbf {\bibinfo {volume} {4}},\ \bibinfo {pages} {93} (\bibinfo
  {year} {1980})}\BibitemShut {NoStop}%
\bibitem [{\citenamefont {Wang}\ \emph {et~al.}(2016)\citenamefont {Wang},
  \citenamefont {Wu},\ and\ \citenamefont {Scarani}}]{wang16a}%
  \BibitemOpen
  \bibfield  {author} {\bibinfo {author} {\bibfnamefont {Y.}~\bibnamefont
  {Wang}}, \bibinfo {author} {\bibfnamefont {X.}~\bibnamefont {Wu}}, \ and\
  \bibinfo {author} {\bibfnamefont {V.}~\bibnamefont {Scarani}},\ }\bibfield
  {title} {\enquote {\bibinfo {title} {{All the self-testings of the singlet
  for two binary measurements}},}\ }\href
  {http://dx.doi.org/10.1088/1367-2630/18/2/025021} {\bibfield  {journal}
  {\bibinfo  {journal} {New J. Phys.}\ }\textbf {\bibinfo {volume} {18}}
  (\bibinfo {year} {2016})}\BibitemShut {NoStop}%
\bibitem [{\citenamefont {Rabelo}\ \emph {et~al.}(2012)\citenamefont {Rabelo},
  \citenamefont {Zhi},\ and\ \citenamefont {Scarani}}]{Rabelo:2012}%
  \BibitemOpen
  \bibfield  {author} {\bibinfo {author} {\bibfnamefont {R.}~\bibnamefont
  {Rabelo}}, \bibinfo {author} {\bibfnamefont {L.~Y.}\ \bibnamefont {Zhi}}, \
  and\ \bibinfo {author} {\bibfnamefont {V.}~\bibnamefont {Scarani}},\
  }\bibfield  {title} {\enquote {\bibinfo {title} {{Device-independent bounds
  for Hardy's experiment}},}\ }\href
  {https://link.aps.org/doi/10.1103/PhysRevLett.109.180401} {\bibfield
  {journal} {\bibinfo  {journal} {Phys. Rev. Lett.}\ }\textbf {\bibinfo
  {volume} {109}},\ \bibinfo {pages} {180401} (\bibinfo {year}
  {2012})}\BibitemShut {NoStop}%
\bibitem [{\citenamefont {Man\v{c}inska}\ and\ \citenamefont
  {Wehner}(2014)}]{Mancinska2014}%
  \BibitemOpen
  \bibfield  {author} {\bibinfo {author} {\bibfnamefont {L.}~\bibnamefont
  {Man\v{c}inska}}\ and\ \bibinfo {author} {\bibfnamefont {S.}~\bibnamefont
  {Wehner}},\ }\bibfield  {title} {\enquote {\bibinfo {title} {{A unified view
  on Hardy's paradox and the Clauser-Horne-Shimony-Holt inequality}},}\ }\href
  {https://dx.doi.org/10.1088/1751-8113/47/42/424027} {\bibfield  {journal}
  {\bibinfo  {journal} {J. Phys. A: Math. Theor.}\ }\textbf {\bibinfo {volume}
  {47}},\ \bibinfo {pages} {424027} (\bibinfo {year} {2014})}\BibitemShut
  {NoStop}%
\bibitem [{\citenamefont {Slofstra}(2011)}]{Slofstra:2011}%
  \BibitemOpen
  \bibfield  {author} {\bibinfo {author} {\bibfnamefont {W.}~\bibnamefont
  {Slofstra}},\ }\bibfield  {title} {\enquote {\bibinfo {title} {{Lower bounds
  on the entanglement needed to play XOR non-local games}},}\ }\href {\doibase
  10.1063/1.3652924} {\bibfield  {journal} {\bibinfo  {journal} {J. Math.
  Phys.}\ }\textbf {\bibinfo {volume} {52}},\ \bibinfo {pages} {102202}
  (\bibinfo {year} {2011})}\BibitemShut {NoStop}%
\bibitem [{\citenamefont {Ramanathan}\ and\ \citenamefont
  {Mironowicz}(2017)}]{ramanathan17a}%
  \BibitemOpen
  \bibfield  {author} {\bibinfo {author} {\bibfnamefont {R.}~\bibnamefont
  {Ramanathan}}\ and\ \bibinfo {author} {\bibfnamefont {P.}~\bibnamefont
  {Mironowicz}},\ }\bibfield  {title} {\enquote {\bibinfo {title} {{Trade-offs
  in multi-party Bell inequality violations in qubit networks}},}\ }\href
  {https://arxiv.org/abs/1704.03790} {\bibfield  {journal} {\bibinfo  {journal}
  {arXiv:1704.03790}\ } (\bibinfo {year} {2017})}\BibitemShut {NoStop}%
\bibitem [{\citenamefont {Froissart}(1981)}]{Froissart1981}%
  \BibitemOpen
  \bibfield  {author} {\bibinfo {author} {\bibfnamefont {M.}~\bibnamefont
  {Froissart}},\ }\bibfield  {title} {\enquote {\bibinfo {title} {{Constructive
  generalization of Bell's inequalities}},}\ }\href
  {https://doi.org/10.1007/BF02903286} {\bibfield  {journal} {\bibinfo
  {journal} {Il Nuovo Cimento B}\ }\textbf {\bibinfo {volume} {64}},\ \bibinfo
  {pages} {241} (\bibinfo {year} {1981})}\BibitemShut {NoStop}%
\bibitem [{\citenamefont {Collins}\ and\ \citenamefont
  {Gisin}(2004)}]{Collins2004}%
  \BibitemOpen
  \bibfield  {author} {\bibinfo {author} {\bibfnamefont {D.}~\bibnamefont
  {Collins}}\ and\ \bibinfo {author} {\bibfnamefont {N.}~\bibnamefont
  {Gisin}},\ }\bibfield  {title} {\enquote {\bibinfo {title} {{A relevant two
  qubit Bell inequality inequivalent to the CHSH inequality}},}\ }\href
  {https://doi.org/10.1088/0305-4470/37/5/021} {\bibfield  {journal} {\bibinfo
  {journal} {J. Phys. A: Math. Gen.}\ }\textbf {\bibinfo {volume} {37}},\
  \bibinfo {pages} {1775} (\bibinfo {year} {2004})}\BibitemShut {NoStop}%
\bibitem [{\citenamefont {Wehner}(2006)}]{wehner06a}%
  \BibitemOpen
  \bibfield  {author} {\bibinfo {author} {\bibfnamefont {S.}~\bibnamefont
  {Wehner}},\ }\bibfield  {title} {\enquote {\bibinfo {title} {{Tsirelson
  bounds for generalized Clauser-Horne-Shimony-Holt inequalities}},}\ }\href
  {http://journals.aps.org/pra/abstract/10.1103/PhysRevA.73.022110} {\bibfield
  {journal} {\bibinfo  {journal} {Phys. Rev. A}\ }\textbf {\bibinfo {volume}
  {73}},\ \bibinfo {pages} {022110} (\bibinfo {year} {2006})}\BibitemShut
  {NoStop}%
\bibitem [{\citenamefont {McKague}\ and\ \citenamefont
  {Mosca}(2011)}]{mckague11a}%
  \BibitemOpen
  \bibfield  {author} {\bibinfo {author} {\bibfnamefont {M.}~\bibnamefont
  {McKague}}\ and\ \bibinfo {author} {\bibfnamefont {M.}~\bibnamefont
  {Mosca}},\ }\bibfield  {title} {\enquote {\bibinfo {title} {{Generalized
  self-testing and the security of the 6-state protocol}},}\ }\href
  {http://dx.doi.org/10.1007/978-3-642-18073-6\_10} {\bibfield  {journal}
  {\bibinfo  {journal} {Theory of Quantum Computation, Communication, and
  Cryptography. TQC 2010. Lecture Notes in Computer Science}\ }\textbf
  {\bibinfo {volume} {6519}} (\bibinfo {year} {2011})}\BibitemShut {NoStop}%
\bibitem [{\citenamefont {Andersson}\ \emph {et~al.}(2017)\citenamefont
  {Andersson}, \citenamefont {Badziąg}, \citenamefont {Bengtsson},
  \citenamefont {Dumitru},\ and\ \citenamefont {Cabello}}]{andersson17a}%
  \BibitemOpen
  \bibfield  {author} {\bibinfo {author} {\bibfnamefont {O.}~\bibnamefont
  {Andersson}}, \bibinfo {author} {\bibfnamefont {P.}~\bibnamefont {Badziąg}},
  \bibinfo {author} {\bibfnamefont {I.}~\bibnamefont {Bengtsson}}, \bibinfo
  {author} {\bibfnamefont {I.}~\bibnamefont {Dumitru}}, \ and\ \bibinfo
  {author} {\bibfnamefont {A.}~\bibnamefont {Cabello}},\ }\bibfield  {title}
  {\enquote {\bibinfo {title} {{Self-testing properties of Gisin's elegant Bell
  inequality}},}\ }\href
  {http://journals.aps.org/pra/abstract/10.1103/PhysRevA.96.032119} {\bibfield
  {journal} {\bibinfo  {journal} {Phys. Rev. A}\ }\textbf {\bibinfo {volume}
  {96}} (\bibinfo {year} {2017})}\BibitemShut {NoStop}%
\bibitem [{\citenamefont {Kaniewski}(2017)}]{kaniewski17a}%
  \BibitemOpen
  \bibfield  {author} {\bibinfo {author} {\bibfnamefont {J.}~\bibnamefont
  {Kaniewski}},\ }\bibfield  {title} {\enquote {\bibinfo {title} {{Self-testing
  of binary observables based on commutation}},}\ }\href
  {http://journals.aps.org/pra/abstract/10.1103/PhysRevA.95.062323} {\bibfield
  {journal} {\bibinfo  {journal} {Phys. Rev. A}\ }\textbf {\bibinfo {volume}
  {95}} (\bibinfo {year} {2017})}\BibitemShut {NoStop}%
\bibitem [{\citenamefont {Colbeck}(2006)}]{colbeck06a}%
  \BibitemOpen
  \bibfield  {author} {\bibinfo {author} {\bibfnamefont {R.}~\bibnamefont
  {Colbeck}},\ }\emph {\bibinfo {title} {{Quantum and relativistic protocols
  for secure multi-party computation}}},\ \href@noop {} {Ph.D. thesis},\
  \bibinfo  {school} {University of Cambridge} (\bibinfo {year}
  {2006})\BibitemShut {NoStop}%
\bibitem [{\citenamefont {P\'al}\ \emph {et~al.}(2014)\citenamefont {P\'al},
  \citenamefont {V\'ertesi},\ and\ \citenamefont {Navascu\'es}}]{Pal:2014}%
  \BibitemOpen
  \bibfield  {author} {\bibinfo {author} {\bibfnamefont {K.~F.}\ \bibnamefont
  {P\'al}}, \bibinfo {author} {\bibfnamefont {T.}~\bibnamefont {V\'ertesi}}, \
  and\ \bibinfo {author} {\bibfnamefont {M.}~\bibnamefont {Navascu\'es}},\
  }\bibfield  {title} {\enquote {\bibinfo {title} {Device-independent
  tomography of multipartite quantum states},}\ }\href
  {https://link.aps.org/doi/10.1103/PhysRevA.90.042340} {\bibfield  {journal}
  {\bibinfo  {journal} {Phys. Rev. A}\ }\textbf {\bibinfo {volume} {90}},\
  \bibinfo {pages} {042340} (\bibinfo {year} {2014})}\BibitemShut {NoStop}%
\bibitem [{\citenamefont {Pironio}\ \emph {et~al.}(2011)\citenamefont
  {Pironio}, \citenamefont {Bancal},\ and\ \citenamefont
  {Scarani}}]{Pironio:2011}%
  \BibitemOpen
  \bibfield  {author} {\bibinfo {author} {\bibfnamefont {S.}~\bibnamefont
  {Pironio}}, \bibinfo {author} {\bibfnamefont {J.-D.}\ \bibnamefont {Bancal}},
  \ and\ \bibinfo {author} {\bibfnamefont {V.}~\bibnamefont {Scarani}},\
  }\bibfield  {title} {\enquote {\bibinfo {title} {Extremal correlations of the
  tripartite no-signaling polytope},}\ }\href
  {http://stacks.iop.org/1751-8121/44/i=6/a=065303} {\bibfield  {journal}
  {\bibinfo  {journal} {J. Phys. A}\ }\textbf {\bibinfo {volume} {44}},\
  \bibinfo {pages} {065303} (\bibinfo {year} {2011})}\BibitemShut {NoStop}%
\bibitem [{\citenamefont {Krisnanda}\ \emph {et~al.}(2017)\citenamefont
  {Krisnanda}, \citenamefont {Zuppardo}, \citenamefont {Paternostro},\ and\
  \citenamefont {Paterek}}]{krisnanda17a}%
  \BibitemOpen
  \bibfield  {author} {\bibinfo {author} {\bibfnamefont {T.}~\bibnamefont
  {Krisnanda}}, \bibinfo {author} {\bibfnamefont {M.}~\bibnamefont {Zuppardo}},
  \bibinfo {author} {\bibfnamefont {M.}~\bibnamefont {Paternostro}}, \ and\
  \bibinfo {author} {\bibfnamefont {T.}~\bibnamefont {Paterek}},\ }\bibfield
  {title} {\enquote {\bibinfo {title} {{Revealing nonclassicality of
  inaccessible objects}},}\ }\href
  {http://dx.doi.org/10.1103/PhysRevLett.119.120402} {\bibfield  {journal}
  {\bibinfo  {journal} {Phys. Rev. Lett.}\ }\textbf {\bibinfo {volume} {119}}
  (\bibinfo {year} {2017})}\BibitemShut {NoStop}%
\bibitem [{\citenamefont {{\'S}liwa}(2003)}]{Sliwa:2003}%
  \BibitemOpen
  \bibfield  {author} {\bibinfo {author} {\bibfnamefont {C.}~\bibnamefont
  {{\'S}liwa}},\ }\bibfield  {title} {\enquote {\bibinfo {title} {{Symmetries
  of the Bell correlation inequalities}},}\ }\href {\doibase
  10.1016/S0375-9601(03)01115-0} {\bibfield  {journal} {\bibinfo  {journal}
  {Phys. Lett. A}\ }\textbf {\bibinfo {volume} {317}},\ \bibinfo {pages} {165 }
  (\bibinfo {year} {2003})}\BibitemShut {NoStop}%
\bibitem [{\citenamefont {Greenberger}\ \emph {et~al.}(1989)\citenamefont
  {Greenberger}, \citenamefont {Horne},\ and\ \citenamefont
  {Zeilinger}}]{greenberger89a}%
  \BibitemOpen
  \bibfield  {author} {\bibinfo {author} {\bibfnamefont {D.~M.}\ \bibnamefont
  {Greenberger}}, \bibinfo {author} {\bibfnamefont {M.~A.}\ \bibnamefont
  {Horne}}, \ and\ \bibinfo {author} {\bibfnamefont {A.}~\bibnamefont
  {Zeilinger}},\ }\bibfield  {title} {\enquote {\bibinfo {title} {{Going beyond
  Bell's theorem}},}\ }\href {\doibase 10.1007/978-94-017-0849-4\_10}
  {\bibfield  {journal} {\bibinfo  {journal} {Bell's Theorem, Quantum Theory
  and Conceptions of the Universe}\ }\textbf {\bibinfo {volume} {37}} (\bibinfo
  {year} {1989}),\ 10.1007/978-94-017-0849-4\_10}\BibitemShut {NoStop}%
\bibitem [{\citenamefont {Simon}(2011)}]{Simon:2011}%
  \BibitemOpen
  \bibfield  {author} {\bibinfo {author} {\bibfnamefont {B.}~\bibnamefont
  {Simon}},\ }\href@noop {} {\emph {\bibinfo {title} {{Convexity: an analytic
  viewpoint}}}}\ (\bibinfo  {publisher} {Cambridge University Press},\ \bibinfo
  {year} {2011})\BibitemShut {NoStop}%
\bibitem [{\citenamefont {Basu}\ \emph {et~al.}(2011)\citenamefont {Basu},
  \citenamefont {Cornu{\'e}jols},\ and\ \citenamefont {Zambelli}}]{Basu:2011}%
  \BibitemOpen
  \bibfield  {author} {\bibinfo {author} {\bibfnamefont {A.}~\bibnamefont
  {Basu}}, \bibinfo {author} {\bibfnamefont {G.}~\bibnamefont
  {Cornu{\'e}jols}}, \ and\ \bibinfo {author} {\bibfnamefont {G.}~\bibnamefont
  {Zambelli}},\ }\bibfield  {title} {\enquote {\bibinfo {title} {{Convex sets
  and minimal sublinear functions}},}\ }\href
  {https://arxiv.org/abs/1701.06550} {\bibfield  {journal} {\bibinfo  {journal}
  {J. Convex Analysis}\ }\textbf {\bibinfo {volume} {18}},\ \bibinfo {pages}
  {427} (\bibinfo {year} {2011})}\BibitemShut {NoStop}%
\bibitem [{\citenamefont {Yang}\ and\ \citenamefont
  {Navascu{\'{e}}s}(2013)}]{yang13a}%
  \BibitemOpen
  \bibfield  {author} {\bibinfo {author} {\bibfnamefont {T.~H.}\ \bibnamefont
  {Yang}}\ and\ \bibinfo {author} {\bibfnamefont {M.}~\bibnamefont
  {Navascu{\'{e}}s}},\ }\bibfield  {title} {\enquote {\bibinfo {title} {{Robust
  self-testing of unknown quantum systems into any entangled two-qubit
  states}},}\ }\href
  {http://journals.aps.org/pra/abstract/10.1103/PhysRevA.87.050102} {\bibfield
  {journal} {\bibinfo  {journal} {Phys. Rev. A}\ }\textbf {\bibinfo {volume}
  {87}},\ \bibinfo {pages} {050102(R)} (\bibinfo {year} {2013})}\BibitemShut
  {NoStop}%
\bibitem [{\citenamefont {Bamps}\ and\ \citenamefont
  {Pironio}(2015)}]{Bamps:2015}%
  \BibitemOpen
  \bibfield  {author} {\bibinfo {author} {\bibfnamefont {C.}~\bibnamefont
  {Bamps}}\ and\ \bibinfo {author} {\bibfnamefont {S.}~\bibnamefont
  {Pironio}},\ }\bibfield  {title} {\enquote {\bibinfo {title} {{Sum-of-squares
  decompositions for a family of Clauser-Horne-Shimony-Holt-like inequalities
  and their application to self-testing}},}\ }\href
  {https://link.aps.org/doi/10.1103/PhysRevA.91.052111} {\bibfield  {journal}
  {\bibinfo  {journal} {Phys. Rev. A}\ }\textbf {\bibinfo {volume} {91}},\
  \bibinfo {pages} {052111} (\bibinfo {year} {2015})}\BibitemShut {NoStop}%
\bibitem [{\citenamefont {Doherty}\ \emph {et~al.}(2008)\citenamefont
  {Doherty}, \citenamefont {Liang}, \citenamefont {Toner},\ and\ \citenamefont
  {Wehner}}]{Doherty2008}%
  \BibitemOpen
  \bibfield  {author} {\bibinfo {author} {\bibfnamefont {A.~C.}\ \bibnamefont
  {Doherty}}, \bibinfo {author} {\bibfnamefont {Y.-C.}\ \bibnamefont {Liang}},
  \bibinfo {author} {\bibfnamefont {B.}~\bibnamefont {Toner}}, \ and\ \bibinfo
  {author} {\bibfnamefont {S.}~\bibnamefont {Wehner}},\ }\bibfield  {title}
  {\enquote {\bibinfo {title} {The quantum moment problem and bounds on
  entangled multi-prover games},}\ }in\ \href
  {http://dx.doi.org/10.1109/CCC.2008.26} {\emph {\bibinfo {booktitle}
  {Proceedings of the 2008 IEEE 23rd Annual Conference on Computational
  Complexity}}},\ \bibinfo {series and number} {CCC '08}\ (\bibinfo
  {publisher} {IEEE Computer Society},\ \bibinfo {address} {Washington, DC,
  USA},\ \bibinfo {year} {2008})\ pp.\ \bibinfo {pages} {199--210}\BibitemShut
  {NoStop}%
\end{thebibliography}%
\appendix
\section{Convex sets}
\label{app:convex-sets}
In this appendix we introduce standard notions and definitions used in convex geometry. For more details, we refer the reader to Chapters 1 and 8 of Ref.~\cite{Simon:2011}.

Let $\cA$ be a convex subset of $\amsbb{R}^{d}$ and, moreover, suppose that $\cA$ is compact (i.e.~closed and bounded). For an arbitrary vector $\vecg \in \amsbb{R}^{d}$ let
\begin{equation*}
c(\vecg) := \max_{\vecu \in \cA} \vecg \cdot \vecu
\end{equation*}
and note that the hyperplane $\{ \vecu \in \amsbb{R}^{d} : \vecu \cdot \vecg = c(\vecg) \}$ is a \textbf{supporting hyperplane}, i.e.~it has a non-empty intersection with $\cA$ and it divides the space into two half-spaces such that $\cA$ is fully contained in one of them. The vector $\vecg$ represents a linear functional acting on $\amsbb{R}^{d}$. It is well known that every convex set can be described as an intersection of half-spaces (possibly infinite). Supporting hyperplanes help us to understand the boundary of the convex set. For an arbitrary functional $\vecg$ the set of points which maximise $\vecg$
\begin{equation*}
\cF(\vecg) := \{ \vecu \in \cA : \vecg \cdot \vecu = c(\vecg) \}
\end{equation*}
is called an \textbf{exposed face} of $\cA$ and since $\cA$ is compact, the face is always non-empty. An exposed face is called \textbf{proper} if $\cF(\vecg) \subsetneq \cA$.

A point $\vecu \in \cA$ is called a \textbf{boundary point} if it belongs to some proper exposed face and we denote the set of boundary points by $\cAbnd$. The set of interior points of $\cA$ is simply the complement of $\cAbnd$ (in $\cA$).

Some boundary points have the property that they cannot be written as a non-trivial convex combination of other points in the set. Such points are called \textbf{extremal} and we denote the set of extremal points by $\cAext$. The Krein-Milman theorem states that any convex compact set (in a finite-dimensional vector space) is equal to the convex hull of its extremal points
\begin{equation*}
\cA = \conv( \cAext ).
\end{equation*}
Therefore, when maximising a linear functional over the set, it suffices to perform the optimisation over its extremal points. In other words, for all $\vecg$ we have
\begin{equation*}
\max_{\vecg \in \cA} \vecg \cdot \vecu = \max_{\vecg \in \cAext} \vecg \cdot \vecu.
\end{equation*}
Knowing the extremal points of $\cA$ is also sufficient to determine its faces. Since a face is a convex compact set, it is equal to the convex hull of its extremal points and the extremal points of the face must also be extremal points of $\cA$. For exposed faces we have
\begin{equation*}
\cF(\vecg) = \conv \big( \{ \vecu \in \cAext : \vecg \cdot \vecu = c(\vecg) \} \big).
\end{equation*}

Among extremal points there are points which can be identified as \emph{unique} maximisers of some linear functional. We say that $\vecu$ is \textbf{exposed} if there exists a linear functional $\vecg$ such that
\begin{equation*}
\cF(\vecg) = \{ \vecu \}
\end{equation*}
and we denote the set of exposed points by $\cAexp$. From the definitions alone, we immediately establish the inclusions
\begin{equation*}
\cAexp \subseteq \cAext \subseteq \cAbnd \subseteq \cA
\end{equation*}
and it is well known that all of them are in general strict. However, it is worth pointing out that by Straszewicz's theorem in a finite-dimensional vector space the set of exposed points is dense in the set of extremal points~\citep[Theorem 3]{Basu:2011}. In other words extremal but non-exposed points should be regarded as exceptional. For a polytope the set of extremal and exposed points coincide, as they are simply the vertices of the polytope.
\section{Precise definition of the quantum set}
\label{app:quantum-set}
Here we give a precise definition of the quantum set using the notions introduced in Section~\ref{sec:quantum-set} (quantum state and local quantum measurements). We also show why the description becomes significantly simpler in any $\scen[n]{2}{2}$ scenario.

Let $\cQ_{d}$ be the set of all probability points which can be realised using systems of local dimension $d$. Since both the set of states and the set of measurements of fixed (local) dimension are compact and the trace is a continuous map, all these sets are closed, i.e.~for all $d \in \amsbb{N}$ we have
\begin{equation*}
\clos( \cQ_{d} ) = \cQ_{d}.
\end{equation*}
We then define the set $\cQfin$ as the infinite union
\begin{equation*}
\cQfin := \bigcup_{d \in \amsbb{N}} \cQ_{d}
\end{equation*}
and the quantum set $\cQ$ as the closure
\begin{equation*}
\cQ := \clos ( \cQfin ).
\end{equation*}
While $\cQ_{d}$ is not necessarily convex, the union $\cQfin$ is and so is the quantum set $\cQ$. Since $\cQ$ is bounded (all the components of the probability vector must belong to the interval $[0, 1]$) and closed (by definition), it is a compact set.

While in general we might have to consider quantum systems of arbitrary large dimensions, Jordan's lemma simplifies the problem in the \scen[n]{2}{2} scenario. Jordan's lemma states that any two (Hermitian) projectors $P$ and $Q$ can be simultaneously block-diagonalised such that the blocks are of size at most $2 \times 2$~\cite{Masanes:2006}. This implies that in any scenario with two binary measurements on each site for any $d \in \amsbb{N}$ we have
\begin{equation*}
\cQ_{d} \subseteq \conv( \cQ_{2} ),
\end{equation*}
which immediately implies
\begin{equation*}
\cQfin \subseteq \conv( \cQ_{2} ).
\end{equation*}
If we start with the inclusion relation
\begin{equation*}
\cQ_{2} \subseteq \cQfin \subseteq \conv( \cQ_{2} )
\end{equation*}
and take convex hulls recalling that $\conv( \cQfin ) = \cQfin$, we arrive at
\begin{equation*}
\cQfin = \conv( \cQ_{2} ).
\end{equation*}
Since $\cQ_{2}$ is closed and the convex hull of a closed set is still closed, we finally obtain
\begin{equation*}
\cQ = \conv( \cQ_{2} ).
\end{equation*}
Let us mention that this observation (for the special case of $n = 2$) was already made by Tsirelson in 1980~\cite{Tsirelson:1980}.
\section{Self-testing of quantum systems}
\label{app:self-testing}
In this appendix we give a formal definition of self-testing and prove a relation between self-testing and extremality.

Let $\vecp \in \cQfin$ be a quantum probability point. A quantum realisation of $\vecp$ consists of Hilbert spaces $\sH_{A}$ and $\sH_{B}$, a state $\rho_{AB}$ acting on $\sH_{A} \otimes \sH_{B}$ and local measurements $\{ E_{a}^{x} \}$ and $\{ F_{b}^{y} \}$ acting on $\sH_{A}$ and $\sH_{B}$, respectively, such that
\begin{equation*}
P(ab|xy) = \tr \big[ ( E_{a}^{x} \otimes F_{b}^{y} ) \rho_{AB} \big]
\end{equation*}
for all $a, b, x, y$. We denote this quantum realisation by {\realreal}.

It turns out that for certain quantum probability points all quantum realisations are closely related. This is conveniently formulated by finding a single realisation from which all other realisations can be generated and we will call such a realisation \emph{canonical}. In the standard formulation of self-testing one can never certify that the state is mixed or that the measurements are non-projective (every point in the quantum set can be obtained by performing projective measurements on a pure state). Therefore, the canonical realisation always consists of projective measurements acting on a pure state. Moreover, we embed it in local Hilbert spaces whose dimension is equal to the rank of the reduced state, which ensures that the reduced density matrices are full-rank. We denote the canonical realisation by {\idealreal}. We use the standard definition of self-testing (see e.g.~Ref.~\cite{McKague:2012}), but we formulate it at the level of density matrices.
\begin{df}
A quantum probability point $\vecp \in \cQfin$ self-tests the canonical quantum realisation {\idealreal} if for every realisation of $\vecp$, denoted by {\realreal}, we can find:
\begin{itemize}
\item Hilbert spaces $\sH_{A''}$ and $\sH_{B''}$,
\item local isometries
\begin{align*}
V_{A} &: \sH_{A} \to \sH_{A'} \otimes \sH_{A''},\\
V_{B} &: \sH_{B} \to \sH_{B'} \otimes \sH_{B''},
\end{align*}
\item an auxiliary quantum state $\sigma_{A''B''}$ acting on $\sH_{A''} \otimes \sH_{B''}$
\end{itemize}
such that for $V := V_{A} \otimes V_{B}$ we have
\begin{equation}
\label{eq:self-testing-measurements}
V \big[ ( E_{a}^{x} \otimes F_{b}^{y} ) \rho_{AB} \big] V^{\dagger} = \big[ ( P_{a}^{x} \otimes Q_{b}^{y} ) \Psi_{A'B'} \big] \otimes \sigma_{A''B''}
\end{equation}
for all $a, b, x, y$.
\end{df}
This equality ensures that applying the real measurement operators to the real state is equivalent to applying the ideal measurements to the ideal state. Moreover, summing over $a$ and $b$ (for any fixed $x$ and $y$) immediately gives
\begin{equation}
\label{eq:self-testing-state}
V \rho_{AB} V^{\dagger} = \Psi_{A'B'} \otimes \sigma_{A''B''},
\end{equation}
which means that by applying local isometries one can find the ideal state $\Psi_{A'B'}$ inside the real state $\rho_{AB}$.

Let us start with the following simple observation.
\begin{obs}
\label{obs:extremality-rank-1-operators}
Let $R_{GH}^{0}, R_{GH}^{1}$ be positive semidefinite operators acting on $\sH_{G} \otimes \sH_{H}$ such that
\begin{equation}
\label{eq:convex-combination}
R_{GH}^{0} + R_{GH}^{1} = S_{G} \otimes T_{H},
\end{equation}
where $S_{G}$ and $T_{H}$ are positive semidefinite operators acting on $\sH_{G}$ and $\sH_{H}$, respectively. If $\rank(S_{G}) = 1$, then the operators $R_{GH}^{0}$ and $R_{GH}^{1}$ must be of the form
\begin{equation*}
R_{GH}^{j} = S_{G} \otimes T_{H}^{j}
\end{equation*}
for some positive semidefinite $T_{H}^{j}$ acting on $\sH_{H}$.
\end{obs}
\begin{proof}
If $\tr(T_{H}) = 0$, we must have $T_{H} = 0$, which immediately implies $R_{GH}^{j} = 0$, i.e.~we can set $T_{H}^{j} = 0$. If $\tr(T_{H}) > 0$, tracing out $H$ in Eq.~\eqref{eq:convex-combination} implies that $R_{G}^{j} = \alpha_{j} S_{G}$ for some $\alpha_{j} \geq 0$. It is easy to check that a bipartite positive semidefinite operator whose marginal is proportional to a rank-1 projector must be a product operator.
\end{proof}
Now, we are ready to state and prove the main result of this appendix.
\begin{prop}
\label{prop:self-tests-extremal}
If a quantum probability point $\vecp \in \cQfin$ self-tests the canonical realisation {\idealreal}, then it must be an extremal point of $\cQfin$.
\end{prop}
\begin{proof}
We show that $\vecp$ cannot be written as a non-trivial convex combination of points in $\cQfin$. More specifically, we show that if $\vecp$ is a self-test and can be written as
\begin{equation}
\label{eq:convex-combination-probabilities}
\vecp = q_{0} \vecp_{0} + q_{1} \vecp_{1}
\end{equation}
for $q_{0}, q_{1} \in (0, 1), q_{0} + q_{1} = 1$ and $\vecp_{0}, \vecp_{1} \in \cQfin$, then we must necessarily have $\vecp_{0} = \vecp_{1} = \vecp$.

Since $\vecp_{j} \in \cQfin$, it has a finite-dimensional realisation on $\sH_{A_{j}}$ and $\sH_{B_{j}}$ given by
\begin{equation*}
\big( \sH_{A_{j}}, \sH_{B_{j}}, \rho_{ A_{j} B_{j} }^{j}, \{ E_{a}^{x, j} \}, \{ F_{b}^{y, j} \} \big)
\end{equation*}
and we choose a realisation in which the reduced states are full-rank, i.e.~$\rank( \rho_{A_{j}} ) = \dim( \sH_{A_{j}})$ and $\rank( \rho_{B_{j}} ) = \dim( \sH_{B_{j}})$. Clearly, the convex combination given in Eq.~\eqref{eq:convex-combination-probabilities} can be realised on $\sH_{A} := \sH_{A_{0}} \oplus \sH_{A_{1}}$ and $\sH_{B} := \sH_{B_{0}} \oplus \sH_{B_{1}}$. Writing out $\sH_{A} \otimes \sH_{B}$ as a direct sum gives
\begin{align*}
\sH_{A} \otimes \sH_{B} &= ( \sH_{A_{0}} \oplus \sH_{A_{1}} ) \otimes ( \sH_{B_{0}} \oplus \sH_{B_{1}} )\\
&= \sH_{A_{0} B_{0}} \oplus \sH_{A_{0} B_{1}} \oplus \sH_{A_{1} B_{0}} \oplus \sH_{A_{1} B_{1}}.
\end{align*}
We embed $\rho_{A_{0} B_{0}}^{0}$ and $\rho_{A_{1} B_{1}}^{1}$ as
\begin{align*}
\rho_{AB}^{0} &:= \rho_{A_{0} B_{0}}^{0} \oplus 0_{ A_{0} B_{1} } \oplus 0_{ A_{1} B_{0} } \oplus 0_{A_{1} B_{1}},\\
\rho_{AB}^{1} &:= 0_{A_{0} B_{0}} \oplus 0_{ A_{0} B_{1} } \oplus 0_{ A_{1} B_{0} } \oplus \rho_{A_{1} B_{1}}^{1}
\end{align*}
and the overall state is given by
\begin{equation*}
\rho_{AB} := q_{0} \rho_{AB}^{0} + q_{1} \rho_{A B}^{1}.
\end{equation*}
The measurement operators are given by
\begin{align*}
E_{a}^{x} &:= ( E_{a}^{x, 0} )_{A_{0}} \oplus ( E_{a}^{x, 1} )_{A_{1}},\\
F_{b}^{y} &:= ( F_{b}^{y, 0} )_{B_{0}} \oplus ( F_{b}^{y, 1} )_{B_{1}}.
\end{align*}
Since {\realreal} is a quantum realisation of $\vecp$ and $\vecp$ self-tests the canonical realisation {\idealreal}, there exist Hilbert spaces $\sH_{A''}, \sH_{B''}$, local isometries
\begin{align*}
V_{A} &: \sH_{A} \to \sH_{A'} \otimes \sH_{A''},\\
V_{B} &: \sH_{B} \to \sH_{B'} \otimes \sH_{B''}
\end{align*}
and an auxiliary state $\sigma_{A''B''}$ such that
\begin{equation*}
V \rho_{AB} V^{\dagger} = \Psi_{A'B'} \otimes \sigma_{A''B''},
\end{equation*}
where $V = V_{A} \otimes V_{B}$ is the combined isometry. If we write out the sum
\begin{equation}
\label{eq:self-testing-sum}
q_{0} V \rho_{AB}^{0} V^{\dagger} + q_{1} V \rho_{AB}^{1} V^{\dagger} = \Psi_{A'B'} \otimes \sigma_{A''B''}
\end{equation}
we obtain an equality to which Observation~\ref{obs:extremality-rank-1-operators} can be applied. To see that all the conditions are satisfied we identify
\begin{align*}
\sH_{A'} \otimes \sH_{B'} &\leftrightarrow \sH_{G},\\
\sH_{A''} \otimes \sH_{B''} &\leftrightarrow \sH_{H},\\
q_{j} V \rho_{AB}^{j} V^{\dagger} &\leftrightarrow R_{GH}^{j},\\
\Psi_{A'B'} &\leftrightarrow S_{G},\\
\sigma_{A''B''} &\leftrightarrow T_{H},
\end{align*}
which allows us to conclude that
\begin{equation}
\label{eq:self-testing-j}
q_{j} V \rho_{AB}^{j} V^{\dagger} = \Psi_{A'B'} \otimes q_{j} \sigma_{A''B''}^{j}
\end{equation}
for some normalised states $\sigma_{A''B''}^{j}$. Tracing out Bob's part of the state and dividing through by $q_{j}$ (recall that $q_{j} > 0$) leads to
\begin{equation}
\label{eq:local-equality}
V_{A} \rho_{A}^{j} V_{A}^{\dagger} = \Psi_{A'} \otimes \sigma_{A''}^{j}.
\end{equation}
Since the quantum realisations of $\vecp_{1}$ and $\vecp_{2}$ and the canonical realisation are locally full-rank, the projectors on the supports of the reduced states are given by
\begin{align*}
\rho_{A}^{0} &\to \mathbb{1}_{A_{0}} \oplus 0_{A_{1}},\\
\rho_{A}^{1} &\to 0_{A_{0}} \oplus \mathbb{1}_{A_{1}},\\
\Psi_{A'} &\to \mathbb{1}_{A'}.
\end{align*}
Equation~\eqref{eq:local-equality} implies that the supports of both sides coincide, i.e.
\begin{align*}
\Pi_{A}^{0} := V_{A} ( \mathbb{1}_{A_{0}} \oplus 0_{A_{1}} ) V_{A}^{\dagger} &= \mathbb{1}_{A'} \otimes \Pi_{A''}^{0},\\
\Pi_{A}^{1} := V_{A} ( 0_{A_{0}} \oplus \mathbb{1}_{A_{1}} ) V_{A}^{\dagger} &= \mathbb{1}_{A'} \otimes \Pi_{A''}^{1},
\end{align*}
where $\Pi_{A''}^{j}$ is the projector on the support of $\sigma_{A''}^{j}$. Similarly, for Bob we obtain
\begin{align*}
\Pi_{B}^{0} := V_{B} ( \mathbb{1}_{B_{0}} \oplus 0_{B_{1}} ) V_{B}^{\dagger} &= \mathbb{1}_{B'} \otimes \Pi_{B''}^{0},\\
\Pi_{B}^{1} := V_{B} ( 0_{B_{0}} \oplus \mathbb{1}_{B_{1}} ) V_{B}^{\dagger} &= \mathbb{1}_{B'} \otimes \Pi_{B''}^{1}.
\end{align*}
By applying the projector $ \Pi_{A}^{j} \otimes \Pi_{B}^{j} $ to both sides of Eq.~\eqref{eq:self-testing-sum} and taking the trace we obtain
\begin{equation*}
q_{j} = \tr \big[ ( \Pi_{A''}^{j} \otimes \Pi_{B''}^{j} ) \sigma_{A''B''} \big].
\end{equation*}
The self-testing condition~\eqref{eq:self-testing-measurements} states that
\begin{equation*}
V \big[ ( E_{a}^{x} \otimes F_{b}^{y} ) \rho_{AB} \big] V^{\dagger} = \big[ ( P_{a}^{x} \otimes Q_{b}^{y} ) \Psi_{A'B'} \big] \otimes \sigma_{A''B''}.
\end{equation*}
Applying the projector $\Pi_{A}^{j} \otimes \Pi_{B}^{j}$ to both sides and tracing out gives
\begin{align*}
q_{j} \tr \big[ ( E_{a}^{x} &\otimes F_{b}^{y} ) \rho_{AB}^{j} \big]\\
&= \tr \big[ ( P_{a}^{x} \otimes Q_{b}^{y} ) \Psi_{A'B'} \big] \cdot \tr \big[ ( \Pi_{A''}^{j} \otimes \Pi_{B''}^{j} ) \sigma_{A''B''} \big]\\
&= q_{j} \tr \big[ ( P_{a}^{x} \otimes Q_{b}^{y} ) \Psi_{A'B'} \big],
\end{align*}
which immediately implies that $\vecp_{j} = \vecp$.
\end{proof}
Since in the \scen{2}{2} scenario we have $\cQ = \cQfin$, this result is sufficient for our purposes. To prove a stronger result in which extremality in $\cQfin$ is replaced by extremality in $\cQ$, one needs a slightly stronger promise, namely that the self-testing property is \emph{robust} (i.e.~that all probability points lying sufficiently close to $\vecp$ ``approximately'' self-test the canonical realisation). We leave this more general statement as an open problem for future work.
\section{Class (3a) does not appear in the \scen{2}{2} scenario}
\label{app:no-3a-in-2222}
We show here that in the $\scen{2}{2}$ scenario the equality $\beta_{\cQ} = \beta_{\NS}$ implies $\beta_{\cL} = \beta_{\cQ} = \beta_{\NS}$. Note that a similar result has been proven in a more general scenario (binary outcomes but an arbitrary number of settings), but for a special family of Bell functions (see Theorem 5.12 of Ref.~\cite{cleve04a}).

For an arbitrary Bell function consider the vertices of the no-signalling polytope which saturate the no-signalling bound $\beta_{\NS}$. If any one of them is local, we immediately have $\beta_{\cL} = \beta_{\NS}$, so we can without loss of generality assume that they are all nonlocal. If the bound is saturated by a single nonlocal vertex, then the quantum bound must be strictly smaller $\beta_{\cQ} < \beta_{\NS}$, because the PR box lies outside of the quantum set. On the other hand, if the bound is saturated by two (or more) nonlocal vertices, we must have $\beta_{\cL} = \beta_{\NS}$. This is because in the \scen{2}{2} scenario the convex hull of any two nonlocal vertices of the no-signalling set always contains a local point (in fact, it suffices to mix the two vertices with equal weights).
\section{The CHSH violation vs.~distance measures}
\label{app:chsh-violation-distance-measure}
In this appendix we show that in the CHSH scenario various visibilities are simple functions of the CHSH violation $\beta$. It is known that any no-signalling point in the CHSH scenario can violate at most one CHSH inequality (see the last paragraph of the supplementary material of Ref.~\cite{Liang:2010}), so the violation is well-defined.

Since the local set is invariant under the relabelling of inputs and outputs, we can without loss of generality assume that it is the standard CHSH inequality, cf.~Eq.~\eqref{eq:b1}, that is violated. Our results rely crucially on the following result proved by Bierhorst~\cite{bierhorst16a}.
\begin{prop}
\label{prop:convex-decomposition}
Let $\vecp \in \NS$ be a no-signalling point which violates the CHSH inequality, i.e.~$\beta > 2$. Then $\vecp$ can be written as
\begin{equation}
\label{eq:2222-convex-decomposition}
\vecp = v_{0} \PPR + \sum_{j = 1}^{8} v_{j} \vecp_{j},
\end{equation}
where $v_{j} \geq 0$, $\sum_{j} v_{j} = 1$ and $\vecp_{j}$ correspond to the 8 deterministic points which give the CHSH value of 2. Moreover, $v_{0} = (\beta - 2)/2$.
\end{prop}
For a no-signalling behaviour we define the \emph{visibility} against noise coming from the set $\cS$ as
\begin{equation*}
v_{\cS}(\vecp) := \inf \{ \lambda \in [0, 1] : (1 - \lambda) \vecp + \lambda \Pnoise \in \cL \},
\end{equation*}
where $\Pnoise \in \cS$. The three cases of interest are (i) visibility against  white noise $\cS = \{ \vecp_{0} \}$, (ii) visibility against local noise $\cS = \cL$ and (iii) visibility against no-signalling noise $\cS = \NS$ and the results are
\begin{equation}
\label{Eq:VisCHSHExplicit}
\begin{split}
v_{ \{ \vecp_{0} \} }(\vecp) &= \frac{ \beta - 2 }{ \beta } \, ,\\
v_{\cL}(\vecp) &= \frac{ \beta - 2 }{ \beta + 2 } \, ,\\
v_{\NS}(\vecp) &= \frac{ \beta - 2 }{ \beta + 4 } \, .
\end{split}
\end{equation}
These relations follow almost immediately from Proposition~\ref{prop:convex-decomposition}. Writing $\vecp$ in the convex decomposition~\eqref{eq:2222-convex-decomposition} and adding noise leads to
\begin{equation*}
(1 - \lambda) v_{0} \PPR + (1 - \lambda) \sum_{j = 1}^{8} v_{j} \vecp_{j} + \lambda \Pnoise.
\end{equation*}
Requiring that the CHSH value of the resulting point does not exceed the classical value of $2$ is equivalent to
\begin{equation*}
\lambda ( \beta - \zeta ) \geq \beta - 2,
\end{equation*}
where $\zeta$ is the CHSH value of $\Pnoise$. As the right-hand side is strictly positive, we must have $\beta - \zeta > 0$, which allows us to rewrite this lower bound as
\begin{equation*}
\lambda \geq \frac{\beta - 2}{ \beta - \zeta }.
\end{equation*}
For white, local and no-signalling noise we have $\zeta = 0$, $\zeta \geq -2$ and $\zeta \geq -4$, respectively, thereby showing that the right-hand-side of Eq.~\eqref{Eq:VisCHSHExplicit} is a legitimate lower bound on the visibilities. To see that this amount of noise is also sufficient choose $\Pnoise = \vecp_{0}$, $\Pnoise = (\vecp_{0} + \PPRk{2})/2$ and $\Pnoise = \PPRk{2}$, respectively.

The fact that the CHSH violation can be interpreted as a measure of distance from the local set might be useful in guiding us towards finding new robust self-tests. Suppose we would like to find a Bell inequality which self-tests a specific partially entangled state of two-qubits in the \scen{2}{2} scenario. Intuitively, we would like the probability point saturating this inequality to lie as far as possible from the local set. This reduces the problem to finding the maximal CHSH violation achievable using the fixed two-qubit state. It is worth pointing out that the self-tests for partially entangled two-qubit states based on the tilted CHSH inequality~\cite{yang13a, Bamps:2015} satisfy this property.
\section{Identifying and certifying flat boundary regions}
\label{app:identifying-certifying}
To the best of our knowledge the only rigorous method to certify the presence of a flat region on the boundary of the quantum set is to find a Bell function whose quantum value is saturated by distinct probability points. Since finding the quantum value of a Bell function is a well-studied problem, let us focus solely on the problem of identifying the relevant Bell function.

We are not aware of any systematic method of finding flat boundary regions of the quantum set. Instead, one has to start with some guesses and in our case they are predominantly of two types:
\begin{enumerate}
\item[(i)] Two Bell functions: we are given two Bell functions and we suspect that their maximal quantum values saturate some linear tradeoff.
\item[(ii)] A set of points: we are given a set of points and we suspect that they all lie on the same quantum face.
\end{enumerate}
In the next two sections we discuss how to handle cases (i) and (ii), respectively.
\subsection{Making a projection plot}
\label{app:bell-tradeoff}
We are given two Bell functions $\vecb_{1}$ and $\vecb_{2}$ and we suspect that they saturate a linear tradeoff. To confirm this we should produce a projection plot (similar to Fig.~\ref{fig:pchsh-pd-proj}) and check for flat boundary regions. There is no exact method of performing such a projection, but we can compute an outer approximation using the Navascu{\'e}s-Pironio-Ac{\'i}n hierarchy and an inner approximation by providing explicit quantum realisations. Finding good inner approximations is particularly feasible in any \scen[n]{2}{2} scenario, since we know that all extremal points of the quantum set can be achieved by performing projective measurements on an $n$-qubit state. We use the fact that projecting a convex compact set is equivalent to projecting its extremal points and then taking the convex hull.\footnote{Note that this is not true for slices: it is in general not sufficient to take the convex hull of the extremal points in the slice.} For instance in the \scen[2]{2}{2} scenario we can fix the Schmidt basis to be the computational basis, i.e.~assume that the bipartite state is of the form $\ket{\psi} = \cos{\theta} \, \ket{00} + \sin{ \theta } \, \ket{11}$ for some $\theta \in [0, \pi/4 ]$. A rank-1 projective observable corresponds to a unit vector on the Bloch sphere, which is specified by two independent parameters. For two observables on each side this gives 9 parameters in total. Therefore, to generate the inner approximation of the projected quantum set we must solve a series of non-linear optimisation problem in $9$ real variables. Such optimisation problems can be solved numerically using standard numerical packages, but we are never guaranteed to converge to the global optimum. However, we have found that repeating the optimisation with random starting points usually yields the correct answer.

Having identified a projection which contains a flat line on the boundary, we should look at its extremal points and find probability distributions that project down to these points. These need not be unique, but it suffices to find one for each of the endpoints. It is easy to see that a line connecting these two probability distributions lies on the boundary of the quantum set.
\subsection{Finding the Bell function}
Let $\{ \vecp_{j} \}_{j}$ be the points which we suspect to belong to the same quantum face. We are looking for a Bell function $\vecb$ whose quantum value is saturated by the points $\{ \vecp_{j} \}_{j}$, i.e.~we require that for all $j$
\begin{equation}
\label{eq:bell-value-constraint}
\vecb \cdot \vecp_{j} = 1
\end{equation}
and $\beta_{\cQ}(\vecb) = 1$.

The most primitive approach is to generate candidate Bell functions satisfying constraints~\eqref{eq:bell-value-constraint} and compute their quantum value. The tool used for computing the quantum value should ideally yield an analytic expression, as one can never distinguish between flat and almost-flat regions using numerical values. If all probability points exhibit a certain symmetry, we might also impose that symmetry on $\vecb$.

This method can be refined by looking at the neighbourhood of the given probability points. Given a quantum realisation of $\vecp_{1}$ it is easy to find some neighbouring quantum points (e.g.~applying rotations to the observables and/or the state). In the limit of infinitesimal change this will give us a family of tangent vectors $V_{k}$ such that $\vecp_{1} + \delta V_{k} \in \cQ$ for sufficiently small $\delta$. Clearly, we must have
\begin{equation}
\vecb \cdot \vec{V}_{k} = 0,
\end{equation}
which can significantly reduce the search space.
\section{Additional examples of quantum faces}
\subsection{A Bell function with $\cF_{\cL} = \cF_{\cQ} = \cF_{\NS}$}
\label{app:type-4d}
We give here another example of a Bell function whose local, quantum and no-signalling faces coincide, but in contrast to the examples given in Eq.~\eqref{eq:exposing-family} the maximiser is not unique. Consider the Bell function:
\begin{equation*}
    \vecb_{7} :=
    \begin{array}{c|c|c}
        & 0 &  0 \\ \hline
        0  &  1 &  1 \\ \hline
        0 &  1 & 1
    \end{array} \; , \;
    \vecb_{7} \cdot \vecp \leq
    \begin{cases}
        4 & \cL\\
        \bm{4} & \cQ\\
        4 & \NS
    \end{cases}.
\end{equation*}
It is straightforward to verify that the only extremal no-signalling points which saturate this inequality are $\Pdet{1}$ and $\Pdet{2}$ (specified in Eqs.~\eqref{eq:Pdet1} and~\eqref{eq:Pdet2}, respectively), i.e.~that the resulting face is a line. Since both of these points are local we immediately deduce that $\cF_{\cL} = \cF_{\cQ} = \cF_{\NS}$.
\subsection{Quantum faces containing the Hardy point}
\label{app:positivity}
Let us start by presenting a quantum face of maximal dimension. Writing the non-negativity of $P(11|11)$ as a Bell inequality gives
\begin{equation}
\label{eq:b8}
    \vecb_{8} :=
    \begin{array}{c|c|r}
        & 0 &  1 \\ \hline
        0  &  0 & 0 \\ \hline
        1  &  0 & -1
    \end{array} \; , \;
    \vecb_{8} \cdot \vecp \leq
    \begin{cases}
        1 & \cL\\
        \bm{1} & \cQ\\
        1 & \NS
    \end{cases}.
\end{equation}
It is easy to verify that the corresponding face of the local set is a (positivity) facet, i.e.~a face of maximal dimension. This immediately implies (by inequalities~\eqref{eq:dimQlb} and \eqref{eq:dimQub}) that the resulting quantum and no-signalling faces are also of maximal dimension.

While the three faces have the same dimension, they are all different, i.e.~the inclusions $\cF_{\cL} \subsetneq \cF_{\cQ} \subsetneq \cF_{\NS} $ are strict. To see this observe that the function is saturated by the Hardy point
\begin{equation}
\label{eq:PHardy}
    \PHardy :=
    \begin{array}{r|r|r}
        & 5-2\sqrt{5} & \sqrt{5} - 2  \\ \hline
        5-2\sqrt{5} & 6\sqrt{5}-13 & 3\sqrt{5} - 6\\ \hline
        \sqrt{5} - 2 &  3\sqrt{5} - 6& 2\sqrt{5}-5
    \end{array} \; ,
\end{equation}
which is quantum (but nonlocal) and also by the (non-quantum) PR box $\PPR$.

As shown in Ref.~\cite{Rabelo:2012}, the Hardy point $\PHardy$ is a self-test and, hence, an extremal point of the quantum set. However, as seen in Fig.~\ref{fig:hardy} and proved in Appendix~\ref{app:hardy-not-exposed} it is not exposed.

It is easy to check that $\PHardy$ saturates two other positivity facets: $P(01|10) \geq 0$ and $P(10|01) \geq 0$. Thus, $\PHardy$ must also saturate
\begin{align}
\label{eq:inequality-HasP}
    a_{1} P(10|01) + a_{2} P(01|10) + a_{3} P(11|11) \geq 0
\end{align}
for arbitrary $a_{1}, a_{2}, a_{3} \geq 0$, which can be written in terms of expectation values as
\begin{equation*}
\begin{split}
    \vecb_{9} :=
    &\begin{array}{c|c|c}
        & a_2 &  a_3-a_1 \\ \hline
        a_1 &  0 & a_1 \\ \hline
        a_3-a_2 & a_2 & -a_3
    \end{array} \; ,\\
    \vecb_{9} \cdot \vecp \leq&
    \begin{cases}
        a_1 + a_2 + a_3 & \cL\\
        \bm{a_1+a_2+a_3} & \cQ\\
        a_1+a_2+a_3 & \NS
    \end{cases} \; .
\end{split}	
\end{equation*}
For $a_{1}, a_{2}, a_{3} > 0$ the Bell function $\vecb_{9}$ is saturated only by points which simultaneously saturate three positivity facets, corresponding to the three terms in Eq.~\eqref{eq:inequality-HasP}. Each of the faces identified by Eq.~\eqref{eq:inequality-HasP} is thus at most 5-dimensional, corresponding to the \emph{intersection} of three 7-dimensional (positivity) faces.

It is easy to check that $\vecb_{9}$ is saturated by 5 local points, namely $\Pdet{1}$ of Eq.~\eqref{eq:Pdet1} and
\begin{equation*}
\begin{split}
    &\Pdet{5} :=
    \begin{array}{r|r|r}
        & \phantom{-}1 &  -1\\ \hline
        \phantom{-}1 &  1 & -1\\ \hline
        1 &  1 & -1
    \end{array} \; , \quad
    \Pdet{6} := \begin{array}{r|r|r}
        &  1 & 1\\ \hline
        1 & 1 & 1\\ \hline
        -1 & -1 & -1
    \end{array} \; ,\\
    &\Pdet{7} :=
    \begin{array}{r|r|r}
        & 1 &  -1\\ \hline
        -1 &  -1 &  1\\ \hline
        1 &  1 &  -1
    \end{array} \; ,\quad
    \Pdet{8} :=
    \begin{array}{r|r|r}
        &  -1 &  1\\ \hline
        1 &  -1 & 1\\ \hline
        -1 & 1 & -1
    \end{array} \; .
\end{split}	
\end{equation*}
The local face is the convex hull of these 5 points and, since they are affinely independent, we obtain a 4-dimensional polytope. The Bell function $\vecb_{9}$ is also saturated by the PR box, which implies that the no-signalling face is a 5-dimensional polytope. The quantum face corresponding to the Bell function $\vecb_{9}$ contains the 5 deterministic points and the Hardy point $\PHardy$, so it must be of dimension 5. We do not know, however, whether it is a polytope or not.
\begin{figure}[h!t]
    \includegraphics[width=9cm]{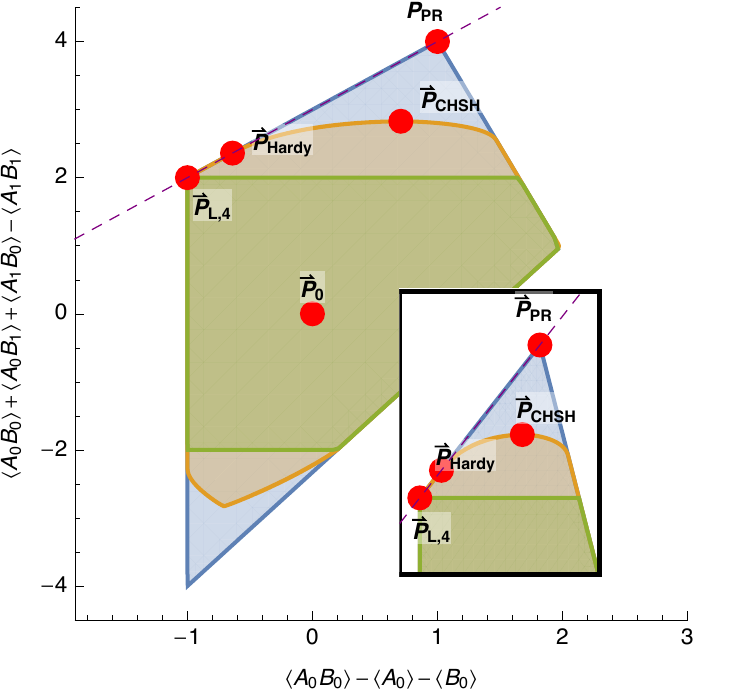}
    \caption{A slice of the quantum set containing the maximally mixed point $\vecp_{0}$, the PR box $\PPR$ and the Hardy point $\PHardy$. The dashed line corresponds to saturating the Bell function $\vecb_{9}$. The point $\PL{4}$ is defined as the intersection of the line going through the PR box and the Hardy point and the hyperplane of the CHSH value of $2$. This local point is not deterministic, but it has a unique decomposition in terms of the five deterministic strategies, namely $\PL{4} = \frac{ 1 + 2 \sqrt{5} }{19} \Pdet{1} + \frac{ 9 - \sqrt{5} }{38} ( \Pdet{5} + \Pdet{6} + \Pdet{7} + \Pdet{8} )$.}
    \label{fig:hardy}
\end{figure}
\section{The Hardy point is not exposed}
\label{app:hardy-not-exposed}
To prove that the Hardy point $\PHardy$ defined in Eq.~\eqref{eq:PHardy} is not exposed in the quantum set we show that any Bell function $\vecb$ maximised by the Hardy point satisfies $\beta_{\cL}(\vecb) = \beta_{\cQ}(\vecb)$, which implies that the Hardy point is not the unique maximiser.

Rabelo et al.~showed that the maximal violation of the Hardy paradox self-tests the following two-qubit state and measurements~\cite{Rabelo:2012}:
\begin{align*}
\ket{\psiH} &= \sqrt{\frac{ 1 - a^{2} }{2}} \, \big( \ket{01} + \ket{10} \big) + a \ket{11},\\
A_{0} &= B_{0} = 2 a \cdot \sigma_{x} + \sqrt{ 1 - 4 a^{2} } \cdot \sigma_{z},\\
A_{1} & = B_{1} = -\sigma_{z}
\end{align*}
for $a := \sqrt{ \sqrt{5} - 2 }$ . Define operators
\begin{equation*}
\begin{array}{ll}
G_{1} = A_{0} \otimes \mathbb{1}, & G_{5} = A_{0} \otimes B_{0},\\
G_{2} = A_{1} \otimes \mathbb{1}, & G_{6} = A_{0} \otimes B_{1},\\
G_{3} = \mathbb{1} \otimes B_{0}, & G_{7} = A_{1} \otimes B_{0},\\
G_{4} = \mathbb{1} \otimes B_{1}, & G_{8} = A_{1} \otimes B_{1}.
\end{array}
\end{equation*}
Let $\vecb$ be an arbitrary Bell function
\begin{equation*}
\vecb :=
    \begin{array}{c|c|c}
        & b_{3} &  b_{4} \\ \hline
        b_{1} &  b_{5} &  b_{6} \\ \hline
        b_{2} &  b_{7} & b_{8}
    \end{array}
\end{equation*}
and the corresponding Bell operator equals
\begin{equation*}
W = \sum_{j = 1}^{8} b_{j} G_{j}.
\end{equation*}
If the Bell function $\vecb$ is maximised by the Hardy point, then in particular the state $\ket{\psiH}$ must be an eigenstate of the Bell operator $W$, i.e.~it must satisfy the linear constraint $W \ket{\psiH} = \lambda \ket{\psiH} $. This forces the Bell function to be tangent to the boundary of the quantum set at the Hardy point. In the proof we show that \emph{every} such Bell function is maximised by at least two points: $\PHardy$ and a local point.

The eigenvalue equation $ W \ket{\psiH} = \lambda \ket{\psiH} $ implies $\bramatket{00}{W}{\psiH} = 0$, because $\braket{00}{\psiH} = 0$. This can be written as a linear constraint $\vecb \cdot \vec{T} = 0$, where the entries of $\vec{T}$ are proportional to $\bramatket{00}{ G_{j} }{\psiH}$. More specifically, we set
\begin{equation*}
    \vec{T} :=
    \begin{array}{c|c|c}
        & 1 &  0 \\ \hline
        1 & \sqrt{5}-1 &  -1 \\ \hline
        0 &  -1 & 0
    \end{array} \; .
\end{equation*}
Our goal is to find the largest value of $\vecb \cdot \PHardy$ for a Bell function maximised by the Hardy point. To write this as a linear program it is convenient to impose some normalisation, e.g.~that the local bound does not exceed $1$ (this is simply a matter of scaling the coefficients). The resulting linear program reads
\begin{gather*}
\begin{array}{ll}
\max & \vecb \cdot \PHardy\\
\textnormal{over} & \vecb \in \amsbb{R}^{8}\\
\textnormal{subject to} & \vecb \cdot \vec{T} = 0\\
& \vecb \cdot \Pdet{j} \leq 1 \nbox[4]{for} j = 1, 2, \ldots, 16,
\end{array}
\end{gather*}
where $\Pdet{j}$ are the deterministic behaviours.\footnote{The deterministic behaviours for $j = \{1, 2, \ldots, 8\}$ were defined in Sections~\ref{sec:b2},~\ref{sec:b3} and Appendix~\ref{app:positivity}. The remaining ones turn out to be irrelevant and the corresponding constraints could be removed without affecting the value of the problem.} The maximum value of the linear program is found to be identically 1; the optimal Bell function returned by the program is $\vecb_{8}$ as specified in Eq.~\eqref{eq:b8}, which achieves no smaller a value at $\PL{4}$ than it does at $\PHardy$, proving that the Hardy point is not exposed.  

The optimality of $\vecb_{8}$ can be shown analytically as follows. First, note that  $\vecb_{8}$ satisfies the constraints defining the linear program. To further show that $\vecb \cdot \PHardy=1$ is the optimal max-value, we write the dual program
\begin{gather*}
\begin{array}{ll}
\min & \sum_{k = 1}^{16} y_{k}\\
\textnormal{over} & y_{k} \geq 0, \; z \in \amsbb{R}\\
\textnormal{subject to} & \sum_{k = 1}^{16} y_{k} \Pdet{j} + z \vec{T} = \PHardy.
\end{array}
\end{gather*}
The assignment
\begin{align*}
y_{k} &=
\begin{cases}
\sqrt{5} - 2 &\nbox{if} k = 1,\\
( 3 - \sqrt{5} )/2 &\nbox{if} k \in \{5, 6\},\\
0 &\nbox{otherwise,}
\end{cases}\\
z &= 4 - 2 \sqrt{5}
\end{align*}
is a valid solution to the dual and the resulting value is $1$. This completes the proof that any Bell function maximised by the Hardy point must satisfy $\beta_{\cL}(\vecb) = \beta_{\cQ}(\vecb)$.
\section{The quantum value of the $B_{6}$ function}
\label{app:b7-quantum-value}
Let $A_{x}$ and $B_{y}$ denote the observables of Alice and Bob, respectively. The maximal quantum value of the Bell function $\vecb_{6}$ can be determined by finding the maximum eigenvalue of the Bell operator
\begin{equation*}
	W = A_{0} ( B_{0} + B_{1} +B_{2} ) + A_{1} ( B_{0} + B_{1} - B_{2} ) + A_{2} ( B_{0} - B_{1} ),
\end{equation*}
where for ease of presentation we assume that $A_{x}$ and $B_{y}$ act on the same composite Hilbert space while satisfying the commutation relations $[A_{x}, B_{y}]=0$ for all $x, y$.

By solving the semidefinite program proposed in Ref.~\cite{wehner06a}, as illustrated in Ref.~\cite{Doherty2008}, one essentially obtains a sum-of-squares decomposition for the operator $\gamma \mathbb{1} - W$ for the smallest possible $\gamma$, thereby showing that the maximal quantum violation of $B_{6}$ is upper bounded by $\gamma$. In particular, it is easy to verify that whenever $A_{x}^{2} = B_{y}^{2} = \mathbb{1}$ and the above-mentioned commutation relations hold true, the following holds
\begin{equation}
\label{eq:sos}
    5 \cdot \mathbb{1} - W = \frac{1}{2} \sum_{j = 1}^{3} V_{j}^\dagger V_{j},
\end{equation}
where $V_{1} = A_{0} + A_{1} - B_{0} - B_{1}$, $V_{2} = A_{0} - A_{1} - B_{2}$ and $V_{3} = A_{2} - B_{0} + B_{1}$. Since the right-hand side of Eq.~\eqref{eq:sos} is a non-negative operator, we see that 5 must be an upper bound on the maximal quantum value of $B_6$. Indeed, this upper bound is saturated by the family of quantum realisations presented in the main text.
\section{Quantum faces in the tripartite scenarios}
\label{app:tripartite-faces}
In this appendix we derive the quantum face corresponding to the Bell function $\vecb_{8}$ discussed in Section~\ref{sec:tripartite-scenarios}.

As explained in Appendix~\ref{app:convex-sets} to determine an exposed quantum face, it suffices to find its extremal points. Since these must also be extremal in the entire quantum set, we simply need to find the extremal points of the quantum set which saturate the quantum bound. In a scenario with two binary observables on each site, every extremal point can be obtained by performing projective rank-1 measurements on an $n$-qubit state~\cite{Masanes:2006}. For Alice we parametrise the observables as
\begin{equation*}
A_{0} = \sigma_{x}, \quad A_{1} = \cos a \cdot \sigma_{x} + \sin a \cdot \sigma_{y}
\end{equation*}
for some angle $a \in [0, \pi]$, while for Bob and Charlie we write
\begin{align*}
&B_{0} = \cos b \cdot \sigma_{x} + \sin b \cdot \sigma_{y}, \quad B_{1} = \cos b \cdot \sigma_{x} - \sin b \cdot \sigma_{y},\\
&C_{0} = \cos c \cdot \sigma_{x} + \sin c \cdot \sigma_{y}, \quad C_{1} = \cos c \cdot \sigma_{x} - \sin c \cdot \sigma_{y}.
\end{align*}
for some angles $b, c \in [0, \pi/2]$. This parametrisation ensures that the resulting Bell operator is easy to diagonalise.

The Bell operator corresponding to the Bell function $\vecb_{8}$ can be seen as the CHSH operator
\begin{equation*}
W = A_{0} \otimes ( T_{0} + T_{1} ) + A_{1} \otimes ( T_{0} - T_{1} ),
\end{equation*}
where $T_{j} := B_{j} \otimes C_{j}$. This implies that the eigenvalue of $2 \sqrt{2}$ is possible only for $a = \pi/2$. Having established the form of Alice's observables we are ready to find the eigenvectors of $W$. It is easy to check that
\begin{align*}
T_{0} + T_{1} = 2 ( \cos b \, \cos c \cdot \sigma_{x} \otimes \sigma_{x} + \sin b \, \sin c \cdot \sigma_{y} \otimes \sigma_{y} ),\\
T_{0} - T_{1} = 2 ( \cos b \, \sin c \cdot \sigma_{x} \otimes \sigma_{y} + \sin b \, \cos c \cdot \sigma_{y} \otimes \sigma_{x} ),
\end{align*}
which implies that
\begin{align*}
W \ket{000} &= - 2 \sqrt{2} \, \sin (b + c - \pi/4) \, \ket{111},\\
W \ket{001} &= - 2 \sqrt{2} \, \sin (b - c - \pi/4) \, \ket{110},\\
W \ket{010} &= 2 \sqrt{2} \, \sin (b - c + \pi/4) \, \ket{101},\\
W \ket{011} &= 2 \sqrt{2} \, \sin (b + c + \pi/4) \, \ket{100},\\
W \ket{100} &= 2 \sqrt{2} \, \sin (b + c + \pi/4) \, \ket{011},\\
W \ket{101} &= 2 \sqrt{2} \, \sin (b - c + \pi/4) \, \ket{010},\\
W \ket{110} &= -2 \sqrt{2} \, \sin (b - c - \pi/4) \, \ket{001},\\
W \ket{111} &= - 2 \sqrt{2} \, \sin (b + c - \pi/4) \, \ket{000}.
\end{align*}
The eigenvectors of $W$ are simply the following GHZ states:
\begin{align*}
\ket{\Omega_{\pm 1}} &= \frac{1}{\sqrt{2}} \big( \ket{000} \pm \ket{111} \big),\\
\ket{\Omega_{\pm 2}} &= \frac{1}{\sqrt{2}} \big( \ket{001} \pm \ket{110} \big),\\
\ket{\Omega_{\pm 3}} &= \frac{1}{\sqrt{2}} \big( \ket{010} \pm \ket{101} \big),\\
\ket{\Omega_{\pm 4}} &= \frac{1}{\sqrt{2}} \big( \ket{011} \pm \ket{100} \big)
\end{align*}
and the corresponding eigenvalues are:
\begin{align*}
\lambda_{\pm 1} &= \mp 2 \sqrt{2} \, \sin (b + c - \pi/4),\\
\lambda_{\pm 2} &= \mp 2 \sqrt{2} \, \sin (b - c - \pi/4),\\
\lambda_{\pm 3} &= \pm 2 \sqrt{2} \, \sin (b - c + \pi/4),\\
\lambda_{\pm 4} &= \pm 2 \sqrt{2} \, \sin (b + c + \pi/4).
\end{align*}
Since we are restricted to the square $b, c \in [0, \pi/2]$, the maximum eigenvalue of $2 \sqrt{2}$ appears if and only if (at least) one of the following equations is satisfied:
\begin{align*}
b + c &= 3 \pi/4,\\
b - c &= - \pi/4,\\
b - c &= \pi/4,\\
b + c &= \pi/4.
\end{align*}
Inside the square, i.e.~for $b, c \in (0, \pi)$, the maximum eigenvalue is non-degenerate and the corresponding eigenvector is unique. The one- and two-body marginals of a GHZ state are fully mixed, which implies that the one- and two-body expectation values must vanish:
\begin{equation*}
\expec{ A_{x} }= \expec{ B_{y} } = \expec{ C_{z} } = \expec{ A_{x} B_{y} } = \expec{ A_{x} C_{z} } = \expec{ B_{y} C_{z} } = 0.
\end{equation*}
Determining the three-body correlations is a simple exercise. For the branch $b + c = 3 \pi / 4 $ we obtain
\begin{align*}
\expec{ A_{0} B_{0} C_{0} } &= \expec{ A_{0} B_{1} C_{1} } = - \cos (b + c) = \frac{1}{\sqrt{2}},\\
\expec{ A_{0} B_{0} C_{1} } &= \expec{ A_{0} B_{1} C_{0} } = - \cos (b - c),\\
\expec{ A_{1} B_{0} C_{0} } &= - \expec{ A_{1} B_{1} C_{1} } = \sin (b + c) = \frac{1}{\sqrt{2}},\\
\expec{ A_{1} B_{0} C_{1} } &= - \expec{ A_{1} B_{1} C_{0} } = \sin (b - c),
\end{align*}
which corresponds to Eq.~\eqref{eq:one-parameter-family} for $\alpha \in [3 \pi/4, 5 \pi/4]$. The other branches give analogous results and cover the rest of the range.

To complete the analysis we must also look at the four special points where the maximal eigenvalue is degenerate, i.e.~$(b, c) = (\pi/4, 0), (0, \pi/4), (\pi/2, 0)$ and $(0, \pi/2)$. At each of these points the subspace corresponding to the maximal eigenvalue is 2-dimensional and the resulting statistics form a line, i.e.~we obtain two extremal points. Computing the extremal points for each pair $(b, c)$ yields the 8 points $\{ P_{j} \}_{j = 1}^{8}$ presented in the main text.
\end{document}